\documentclass[12pt]{amsart}
\usepackage[latin1]{inputenc}
\usepackage[all]{xy}
\usepackage{amssymb, amsmath, amscd, geometry}
\usepackage{graphicx}
\usepackage{setspace}
\numberwithin{equation}{section}
\input xy
\xyoption{all}
\usepackage{latexsym}
\usepackage[usenames]{color}
\usepackage{epic} 
\usepackage{mathrsfs}
\usepackage{euscript}
\usepackage{rotating}

\theoremstyle{plain}
\newtheorem{lemma}{Lemma}[section]
\newtheorem{theorem}[lemma]{Theorem}
\newtheorem{corollary}[lemma]{Corollary}
\newtheorem{problem}[lemma]{Problem}
\newtheorem{prop}[lemma]{Proposition}
\newtheorem{conc}[lemma]{Conclusion}

\newtheorem{OP}[lemma]{OP}
\newtheorem{step}{Step}

\theoremstyle{definition}
\newtheorem{definition}{Definition}[section]
\newtheorem{remark}{Remark}[section]

      \newcommand{\N}{{\mathbb N}}
      \newcommand{\R}{{\mathbb R}}
      \newcommand{\C}{{\mathbb C}}

\newcommand{\lra}{\longrightarrow}

\newcommand{\id}{\operatorname{id}}

\newcommand{\viol}{\operatorname{viol}}

\newcommand{\iviol}{\operatorname{iviol}}

\newcommand{\al}{\alpha}
\newcommand{\la}{\lambda}
\newcommand{\si}{\sigma}
\newcommand{\eps}{\epsilon}
\newcommand{\pl }{\,}
\newcommand{\lel }{\, =\, }
\newcommand{\kl}{\, \le \, }
\newcommand{\gl}{\, \ge \, }
\newcommand{\ten}{\otimes}

\allowdisplaybreaks

\begin{document}

\baselineskip=17pt

\title[Large violation of Bell inequalities with low entanglement]{Large violation of Bell inequalities with low entanglement}

\author{M. Junge}
\email{junge@math.uiuc.edu}
\address{Department of Mathematics, University of Illinois, Urbana, IL 61801, USA}

\author{C. Palazuelos}
\email{cpalazue@illinois.edu}
\address{Department of Mathematics, University of Illinois, Urbana, IL 61801, USA}

\thanks{The authors are partially supported by National Science Foundation grant  DMS-0901457}

\maketitle

\begin{abstract}
In this paper we obtain violations of general bipartite Bell inequalities of order $\frac{\sqrt{n}}{\log n}$ with $n$ inputs,
$n$ outputs and $n$-dimensional Hilbert spaces. Moreover, we construct explicitly, up to a random choice of signs, all the
elements involved in such violations: the coefficients of the Bell inequalities, POVMs measurements and quantum states.
Analyzing this construction we find that, even though entanglement is necessary to obtain violation of Bell inequalities, the Entropy
of entanglement of the underlying state is essentially irrelevant in obtaining large violation. We also indicate why the maximally
entangled state is a rather poor candidate in producing large violations with arbitrary coefficients. However, we also show that
for Bell inequalities with positive coefficients (in particular, games) the maximally entangled state achieves the largest
violation up to a logarithmic factor.
\end{abstract}

\section{Introduction and main results}\label{State of the Main Results}

The study of quantum nonlocality dates back to the famous work of
Einstein, Podolsky and Rosen (EPR) in 1935. They presented an
argument which questioned the validity of quantum mechanics as a
complete theory of Nature (\cite{EPR}). However, it took almost 30
years to understand that the apparently dilemma presented in
\cite{EPR} could be formulated in terms of assumptions which
naturally lead to a refutable prediction (\cite{WernerWolf}). Bell
showed that the assumption of a local hidden variable model
implies some inequalities on the set of probabilities, since then
called \emph{Bell inequalities}, which are violated by certain
quantum probabilities produced with an \emph{entangled state}
(\cite{Bell}). For a long time after this, entanglement and
violation of Bell inequalities were thought to be parts of the
same concept. This changed in the late 1980s with a number of
surprising results (see \cite{Werner}, \cite{Popescu},
\cite{Gisin}) which showed that, although entanglement is
necessary for the violation of Bell inequalities, the converse is
not true. On the other hand, up to our knowledge, violation of
Bell inequalities is the only way to detect entanglement
experimentally without additional hypothesis on the experiment.

\

Nowadays, Bell inequalities is a fundamental subject in Quantum
Information Theory (QIT). Apart from the theoretical interest,
Bell inequalities have found applications in many areas of QIT:
quantum cryptography, where it opens the possibility of getting
unconditionally secure quantum key distribution
(\cite{Acin1,Acin,Mas2,Mas}), complexity theory, where it enriches
the theory of multipartite interactive proof systems
(\cite{Ben-Or,Cleve,Cleve2,Jain,DLTW,KRT,KKMTV}), communication
complexity (see the recent review \cite{Buhrman}); Estimates for
the dimension of the underlying  Hilbert space
(\cite{Briet,Brunner2,PWJPV,Vertesi,Wehner}), entangled games
(\cite{KRT,KR}), etc.

\

Bell inequalities and their connection to quantum entanglement
have remained quite mysterious despite the recent research on
this topic. In the few last years, the application of techniques
from different areas of mathematics has started to clarify the
situation. This includes the previous works of the authors, which
are based on operator space techniques. Indeed, in the consecutive
works \cite{JPPVW} and \cite{PWJPV}, the authors have shown the
\emph{operator space theory} as a natural framework for the study
of Bell inequalities (see also \cite{JPPVW2}). Using this
connection the authors proved in \cite{JPPVW} the existence of
unbounded violations of tripartite correlation Bell inequalities,
answering an old question stated by Tsirelson (\cite{Tsirelson}).
Moreover, in \cite{PWJPV} the authors used operator spaces
techniques to get unbounded violations of general bipartite Bell
inequalities.

\

In the present paper we improve the main results of \cite{PWJPV}.
In fact, we obtain violations of general bipartite Bell
inequalities of order $\frac{\sqrt{n}}{\log n}$ with $N=n$ inputs,
$K=n$ outputs and $d=n$-dimensional Hilbert spaces. We also
provide upper bounds for general Bell inequalities of order ${\rm
O}(N)$, ${\rm O}(K)$ and ${\rm O}(d)$. In addition of being almost
optimal in all the parameters of the problem ($\sqrt{n}$ instead
of $n$) our estimates are also very concrete. Indeed, we construct
explicitly, up to a random choice of signs, all the elements
involved in the violation. That is, the coefficients of the Bell
inequalities, the quantum state and the POVM's. We hope these
constructions can be used for further applications. Moreover, we
connect our estimates with the entropy of entanglement of the
underlying pure states. To our own surprise, violation and entropy
of entanglement appear to be almost independent. Also, the
maximally entangled state is only of very limited use in producing
violation. Moreover, we show that this limitation is not longer
true when considering Bell inequalities with positive coefficients
(in particular, games), where the maximal entangled state always
gives the largest violation up to a logarithmic factor.

\

Let us now state the results more explicitly. A standard scenario
to study quantum nonlocality consists on two spatially separated
and non communicating parties, usually called Alice and Bob. Each
of them can choose among different observables, labelled by
$x=1,\cdots , N$ in the case of Alice and $y=1,\cdots , N$ in the
case of Bob. The possible outcomes of this measurements are
labelled by $a=1,\cdots , K$ in the case of Alice and $b=1,\cdots
, K$ in the case of Bob. Following the standard notation, we will
refer the observables $x$ and $y$ as \emph{inputs} and call $a$
and $b$ \emph{outputs}. We are considering the same number of
inputs (resp. outputs) for Alice and Bob just for simplicity. For
fixed $x,y$, we will consider the probability distribution
$(P(a,b|x,y))_{a,b=1}^K$ of positive real number satisfying
 \[ \sum_{a,b=1}^KP(ab|xy)\, =\, 1  \]
for all $x,y$.  The collection $P=(P(a,b|x,y))_{x,y; a,b=1}^{N,K}$
are called \emph{probability distributions}.

\begin{center}
\begin{figure}[h]
  \includegraphics[width=14cm]{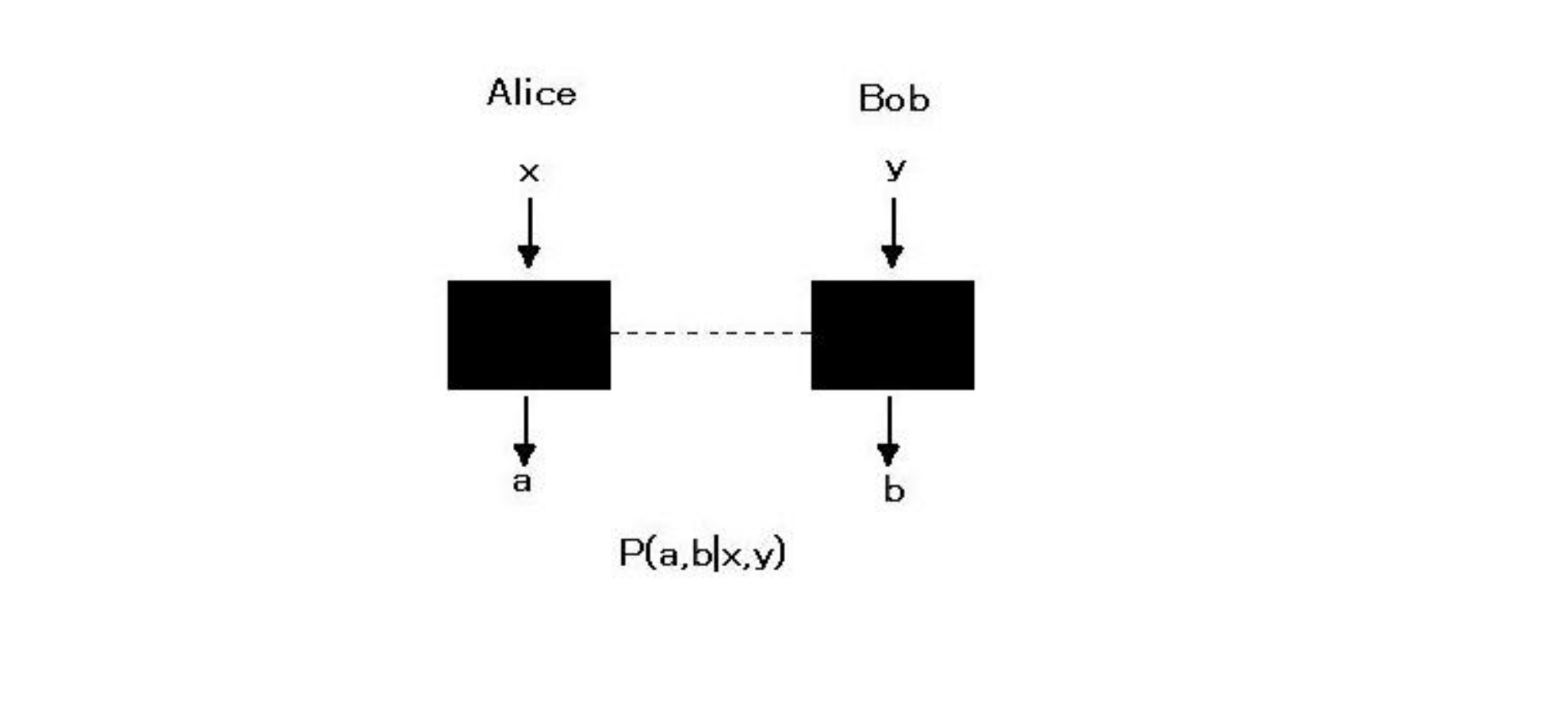}\\
  \caption{$\{P(a,b|x,y)\}_{a,b,x,y}$  is the probability distribution of the measurement outcomes $a,b$, when Alice and Bob choose the observables labeled by $x$ and $y$ respectively.}\label{figure1}
\end{figure}
\end{center}

Given a probability distribution $P=(P(a,b|x,y))_{x,y; a,b=1}^{N,K}$, we will say that $P$ is

\begin{itemize}

\item [a)] \emph{Non-signalling} if
\begin{align*}
\sum_{a=1}^K P(a,b|x,y)&=P(b|y) \text{ is independent of } x, \\
\sum_{b=1}^K P(a,b|x,y)&=P(a|x) \text{ is independent of } y.
\end{align*}
This condition means that Alice choice of inputs does not affect Bob's marginal probability distribution and viceversa. This is physically motivated by the principle of \emph{Einstein locality} which implies non-signalling if we assume that Alice and Bob are space-like separated. We denote the set of non-signalling probability distributions by $\mathcal{C}$.

We must point out that the elements in $\mathcal{C}$ were initially called \emph{behaviors} (see \cite{Tsirelson}). However, following the more recent literature (see \cite{DKLR} and \cite{JPPVW}), we will not use that terminology.

\item [b)] \emph{LHV (Local Hidden Variable)} if
\begin{equation*}\label{classical}
P(a,b|x,y)=\int_\Omega P_\omega(a|x)Q_\omega(b|y)d\mathbb{P}(\omega)
\end{equation*}
for every $x,y,a,b$, where $(\Omega,\Sigma,\mathbb{P})$ is a probability
space, $P_\omega(a|x)\ge 0$ for all $a,x,\omega$, $\sum_a
P_\omega(a|x)=1$ for all $x,\omega$ and the analogous conditions
for $Q_\omega(b|y)$. We denote the set of LHV probability
distributions by $\mathcal{L}$.

\item [c)] \emph{Quantum} if there exist two Hilbert spaces $H_1$,
$H_2$ such that
\begin{equation*}\label{quantum}
P(a,b|x,y)=tr(E_x^a\otimes F_y^b \rho)
\end{equation*}
for every $x,y,a,b$, where  $\rho\in B(H_1\otimes H_2)$ is  a density
operator and $(E_x^a)_{x,a}\subset B(H_1)$, $(F_y^b)_{y,b}\subset B(H_2)$ are two sets
of operators representing POVM
measurements on Alice and Bob systems. That is, $E_x^a\geq 0$ for every  $x,a$, $\sum
_{a}E_x^a=\id$ for every $x$, $F_y^b\geq 0$ for every $y,b$
and $\sum _{b}F_y^b=\id$ for every $y$. We denote the set of
quantum probability distributions by $\mathcal{Q}$.

\end{itemize}

It is well known (see \cite{Tsirelson}) that $\mathcal{L}\varsubsetneq\mathcal{Q}\varsubsetneq  \mathcal{C}\subset \R^{ N^2K^2}$.

\

We want to understand the ``distance between $\mathcal{L}$ and
$\mathcal{Q}$'' quantitatively. Following \cite{JPPVW}, we define
\emph{the largest Bell violation that a given $P\in \mathcal{C}$
may attain} as
\begin{equation*}\label{nu}
\nu(P)= \sup_M\frac{|\langle M,P\rangle|}{\sup_{P'\in\mathcal{L}} |\langle M,P'\rangle|}  ,
\end{equation*}where
$M=\{M_{x,y}^{a,b}\}_{x,y=1,a,b=1}^{N,K}$ is the \emph{Bell
inequality}
acting on $P$ by duality in the natural way: $\langle M,P
\rangle=\sum_{x,y;a,b=1}^{N,K}M_{x,y}^{a,b}P(a,b|x,y).$

Thus, in order to measure how far the elements in $\mathcal{Q}$ can be from $\mathcal{L}$, we are interested in computing the maximal possible
Bell violation

\begin{equation}\label{distanceviolation}
\sup_{P\in\mathcal{Q}}\nu(P).
\end{equation}

Beyond the theoretical interest of $\sup_{P\in\mathcal{Q}}\nu(P)$
as a measure of nonlocality, this term turns out to be a useful
measure regarding the applications in different contexts. Indeed,
in \cite{JPPVW2} (see also \cite{JPPVW}, \cite{DKLR}) the authors
showed its immediate application to dimension witness,
communication complexity or entangled games. Moreover, this term can be
used to measure nonlocality in the presence of noise or/and
detector inefficiencies. This is the key point in the search of a
loophole free Bell test (see \cite{JPPVW2}, \cite{JPPVW} for
details). The main result of this paper can be stated as follows:

\begin{theorem}\label{Theorem 1}
For every $n\in \N$ there exists a quantum probability distribution $P$ with $n$ inputs, $n$
outputs and Hilbert spaces of dimension $n$ such that
\begin{align*}
\nu(P)\succeq \frac{\sqrt{n}}{\log n}.
\end{align*}
Here, we use $\succeq$ to denote inequality up to a universal
constant independent of $n\in \mathbb{N}$.
\end{theorem}

The first unbounded violation of Bell inequalities dates back to
the Raz parallel repetition theorem (\cite{Raz}). Indeed, applying
this result to the repetition of the magic square game (or any
pseudo-telepathy game (\cite{Brassard-review})), one can deduce
the existence of an $x> 0$ such that for every $n$ we have quantum
probability distributions $P$ with $n$ inputs, $n$ outputs and
dimension $n$ such that $\nu(P)\succeq n^x.$ However, regarding
the sharpest estimates on the parallel repetition theorem
(\cite{Holenstein}, \cite{Rao}, \cite{Raz2}), the best known value
for the previous $x$ doesn't seem to be much better that
$10^{-5}$. In \cite{KRT} the authors made a great improvement of
the previous results. Via a highly non trivial construction of
Khot and Visnoi in the context of Complexity Theory (\cite{KV}),
the main result in \cite{KRT} shows the existence of a quantum
probability distribution $P$ with $n$ outputs and $\frac{2^n}{n}$
inputs, which verifies that $\nu(P)\succeq n^{\frac{1}{54}}$
\footnote{Actually, one can obtain $n^{\frac{1}{24}}$ up to terms
of lower order via a claim in (\cite{CMM}, pag. 3).}. The prize in
that estimate is a large number of inputs and no control in terms
of the dimension of the underlying Hilbert space. In the
recent paper \cite{JPPVW}, the authors showed the existence
of a quantum probability $P$ constructed with
$[2^\frac{\log^2n}{2}]^n$ inputs, $n$ outputs and Hilbert spaces
dimension $n$ which verifies $\nu(P)\succeq
\frac{\sqrt{n}}{\log^2n}$. This result highly improved the
previous ones, almost closing the gap to the known upper bounds in
the number of outputs and in the dimension of the Hilbert spaces. Finally, in the very recent work
\cite{BGW} the authors improved the previous estimate obtaining a quantum probability distribution $P$ verifying $\nu(P)\succeq
\frac{\sqrt{n}}{\log n}$ with $2^n$ inputs, $n$ outputs and $n$ dimensional Hilbert spaces.
As before these results required a large number of inputs.

\

Therefore, Theorem \ref{Theorem 1} significantly improves the
previously known results about unbounded violations of Bell
inequalities. It almost closes the gap to the known upper bounds
(see Section \ref{geometric}) in all the involved parameters of
the problem. Although the proof of the result relies on some
probabilistic estimates, all the ingredients are constructed
explicitly. Indeed, we  consider a fixed number of $\pm 1$ signs
$\eps_{x,a}^k$ with $x,a,k=1,\cdots ,n$. For a constant $K$ we
define
\begin{enumerate}
\item[a)] \emph{Bell inequality coefficients}:
   \[    \tilde{M}_{x,y}^{a,b}\, = \,
   \begin{cases}  \frac{1}{n^2}\sum_{k=1}^n\epsilon_{x,a}^k\epsilon_{y,b}^k &  x,y,a,b=1,\cdots, n \cr
   0 & a=n+1, \mbox{and } x,y,b=1,...,n. \cr
   0 & b=n+1, \mbox{and } x,y,a=1,...,n  \, .
   \end{cases}  \]
\item[b)] \emph{POVMs measurements}: $\{E_x^a\}_{x,a=1}^{n,n+1}$ in $M_{n+1}$ as
   \[
    E_x^a \, = \,
    \begin{cases}   \frac{1}{nK}\left( \begin{array}{ccccccc}
  1 &  \epsilon_{x,a}^1 & \cdots & \epsilon_{x,a}^n \\
  \epsilon_{x,a}^1 &  1  & \cdots & \epsilon_{x,a}^1\epsilon_{x,a}^n\\
  \vdots & \vdots & \vdots  & \vdots \\
  \epsilon_{x,a}^n &  \epsilon_{x,a}^n\epsilon_{x,a}^1  & \cdots & 1\\ \end{array}\right)
  & \mbox{ for }  a=1,\cdots ,n\, , \cr
  1-\sum_{a=1}^nE_x^a & \mbox{ for }  a=n+1\,
  \end{cases} \]
  for $x=1,\cdots ,n.$
\item[c)] \emph{States}: Let $(\al_i)_{i=1}^{n+1}$ be a decreasing and positive
sequence and
   \begin{align*}
 |\varphi_\alpha\rangle= \sum_{i=1}^{n+1} \al_i |ii\rangle.
 \end{align*}
\end{enumerate}

\begin{theorem}\label{Theorem-construction}
There exist universal constants $C$ and $K$ such that for every natural number $n$ there exists a choice of signs $\{\epsilon_{x,a}^k\}_{x,a,k=1}^{n}$ verifying that $\{E_x^a\}_{x,a=1}^{n,n+1}$ define POVMs measurements,
\begin{equation}\label{violation-constructive1}
\sup\left\{|\sum_{x,y; a,b=1}^{n,n+1}\tilde{M}_{x,y}^{a,b}P(a,b|x,y)|: P\in \mathcal L\right\}\leq C\log n
\end{equation} and
\begin{equation}\label{violation-constructive2}
\sum_{x,y; a,b=1}^{n,n+1}\tilde{M}_{x,y}^{a,b}\langle
 \varphi_\alpha|E_x^a\otimes E_y^b |\varphi_\alpha\rangle\geq
 \frac{2}{K^2} \,  \al_1\sum_{i=2}^{n+1} \al_i.
\end{equation}
Moreover, the probability of the elements (choices of signs) verifying this tends to $1$ exponentially fast as $n\rightarrow \infty$.
\end{theorem}








This explicit construction allows us to study the connection
between two concepts, violation of Bell inequalities and quantum
entanglement, which are at the heart of Quantum Information
Theory. Indeed, for bipartite pure states there exists a universal
measure of entanglement, the so called \emph{entropy of
entanglement}:

\begin{align*}
\mathcal{E}(|\psi\rangle)=S((|\psi\rangle\langle\psi|)_A),
\end{align*}
where $S$ denotes the usual von Neumann entropy (see \cite{DHR}).
It is easy to see that $\mathcal{E}(|\psi\rangle)\geq 0$ for every
state $|\psi\rangle$ and that the maximally entangled state in
dimension $n$ is
\begin{align*}
|\psi_{n}\rangle=\frac{1}{\sqrt{n}}\sum_{i=1}^{n}|ii\rangle,
\end{align*}
verifying $\mathcal{E}(|\psi_{n}\rangle)=\log_2(n)$. For a given
bipartite pure state $|\varphi\rangle$ in dimension $n$ and
$\delta> 0$, we will say that $|\varphi\rangle$ is
\emph{$\delta$-maximally entangled} (resp. \emph{$\delta$- non
entangled}) if $\log_2(n)-\mathcal{E}(|\varphi\rangle)< \delta$
(resp. $\mathcal{E}(|\varphi\rangle)< \delta$). As a consequence
of Theorem \ref{Theorem-construction} we have:

\begin{corollary}\label{delta-entangled}

For any $\delta> 0$ we can find a $\delta$-maximally entangled state (resp. $\delta$- non entangled state) $|\psi_\delta\rangle$ in a high enough dimension $n$, a Bell inequality $(M_{x,y}^{a,b})_{x,y,a,b=1}^n$ and POVMs measurement $\{E_x^a\}_{x,a=1}^n$  such that
\begin{align*}
\frac{|\langle M,Q_{|\psi_\delta\rangle}\rangle|}{\sup_{P\in\mathcal{L}} |\langle M,P\rangle|}\succeq \frac{\sqrt{n}}{(\log n)^2},
\end{align*}
where $Q_{|\psi_\delta\rangle}(a,b|x,y)=\langle \psi_\delta|E_x^a\otimes E_y^b|\psi_\delta\rangle$ for every $x,y,a,b=1,\cdots, n$.
\end{corollary}


The previous corollary shows that \emph{even though quantum
entanglement is needed to obtain violation of Bell inequalities,
the amount of entanglement is essentially irrelevant for large
violation}. Indeed, we can find states with entropy of
entanglement close to either $0$ or $\log_2 (n+1)$ and this only
decreases violation by a logarithmic factor.

\

It is interesting to note that the previous construction doesn't
say anything about the extremal cases: entanglement $0$ (which is
trivial) and maximal entanglement. This leads us to the following
result:

\begin{theorem}\label{Theorem 2}
There exist a Bell inequalities $\tilde{M}$ with $2^{n^2}$ inputs and $n+1$ outputs and POVMs $\{\tilde{E}_x^a\}_{x,a=1}^{2^{n^2},n+1}$ acting on $\ell_2^{n+1}$ with the following properties:

\begin{enumerate}
\item[a)] $\tilde{M}$ and $\{\tilde{E}_x^a\}_{x,a=1}^{2^{n^2},n+1}$ verify equations (\ref{violation-constructive1}) and (\ref{violation-constructive2}) in Theorem \ref{Theorem-construction} for every state $|\varphi_\alpha\rangle= \sum_{i=1}^{n+1} \al_i |ii\rangle$.

\

\item[b)] $\sup \{|\langle \tilde{M},Q_{max}\rangle|\}\preceq 1$, where this $\sup$ runs over all quantum probability distributions $Q_{max}$ constructed with the maximally entangled state in any dimension.

\end{enumerate}

\end{theorem}



In particular, Theorem \ref{Theorem 2} shows the existence of quantum probability distributions $P$ which can not be written as a quantum probability distribution by using the maximally entangled state, even when the dimension of the Hilbert spaces is not restricted (note the difference with the case of quantum correlations matrices, \cite{Tsirelson}). However, we will show that every $n$ dimensional diagonal state can be written, up to a $\sqrt{\log n}$ factor, as a \emph{superposition} of maximally entangled states in the same dimension. A very interesting consequence of this superposition result, is that Theorem \ref{Theorem 2} is not longer true if we restrict to Bell inequalities with positive coefficients (in particular, games), because in that case the maximal entangled state always gives the largest violation up to a $\log$ factor in the dimension of the Hilbert space (see Theorem \ref{positive coefficients}).



\


The paper is organized as follows. Section \ref{Mathematical Tools} is devoted to introduce the basic tools. In the first part
we will give a brief introduction to operator space theory. In the second part of this section we will summarize the connections
between operator spaces and Bell inequalities from \cite{JPPVW}. Furthermore, we will explain some new connections which will be
used in this work. In Section \ref{Main result} we will prove Theorem \ref{Theorem 1}. However, we will first give a direct less
explicit proof. This proof serves as a guideline for the strategy used throughout the paper. In the last part of the section, we
will discuss the optimality of our result. In Section \ref{The construction} we will present the proof of Theorem \ref{Theorem-construction} and we will investigate  the connection between the amount of violation and the entropy of entanglement of our states leading to Corollary \ref{delta-entangled}. In section \ref{Violation vs GHZ} we study the maximally entangled state. In the first part of the section we will prove Theorem \ref{Theorem 2}. Motivated by this result, we will clarify the role of the maximally entangled state in the context of violation of Bell
inequalities. In Section \ref{geometric} we will discuss the geometric meaning of violation of Bell inequalities. We will
show that the ``distance'' introduced in equation (\ref{distanceviolation}) and its dual version $LV$ (see Equation
(\ref{violation})) are not only the right ones regarding the applications of violation of Bell inequalities to different
contexts of QIT, but they are also the natural ones from a geometric point of view. As a consequence of the results
developed in this section we will obtain upper bounds for the largest violation of Bell inequalities. Finally, in Section \ref{section-gamma} we will study the role of the $\gamma_2^*$ tensor norm in the context of violation of Bell inequalities. The main motivation is the study of this norm as a relaxation of the problem of computing the classical and the quantum value of Bell inequalities. Actually, we will show that this relaxation is related to some well known SDP relaxations already used in the study of some problems of Complexity Theory. Throughout the section, we will give some optimal results for this norm.


\section{Basic tools}\label{Mathematical Tools}

\subsection{Operator spaces}

We will recall some basic facts from operator spaces theory. We recommend \cite{EffrosRuan} and \cite{Pisierbook} for further information and more detailed definitions. We will denote by $M_n$ (resp $M_{m,n}$) the space of complex $n\times n$ (resp $m\times n$) matrices.

\

The theory of operator spaces came to life through the work of Effros and Ruan in the 80's (see \cite{EffrosRuan,Pisierbook}). They provided an axiomatic characterization of closed subspaces of $B(H)$, the space of bounded linear operators on a Hilbert space, where the objects are Banach spaces $E$ combined with a tail of matrix norms on $M_n(E)$ attached to it. More formally, an operator space is a complex vector space $E$ and a sequence of
norms $\|\cdot\|_n$ in the space of $E$-valued matrices
$M_n(E)=M_n\otimes E$, verifying Ruan's axioms
\begin{enumerate}

\item For every $n,m\in\N$, $x\in M_m(E)$, $a\in M_{n,m}$ and $b\in M_{m,n}$ we have that
     \[  \|axb\|_n\le \|a\|\|x\|_m\|b\| \,  . \]

\item For every $n,m\in \N$, $x\in M_n(E)$, $y\in M_m(E)$, we have that
     \[ \left\|\left(\begin{array}{cc}
  x & 0 \\
  0 & y \\
 \end{array} \right)\right\|_{n+m}= \max\{\|x\|_n,\|y\|_m\} \, .  \]
  \end{enumerate}

In particular, every $C^*$-algebra $\mathcal{A}$ has a natural operator space structure induced by a faithful  embedding $j:\mathcal{A}\hookrightarrow B(H)$. Indeed, it is enough to consider the sequence of norms on $M_n\otimes \mathcal{A}$ defined by the embedding $id\otimes j:M_n\otimes \mathcal{A}\hookrightarrow M_n\otimes B(H)=B(\ell_2^n\otimes H)$. In
particular, $\ell_{\infty}^k$ has a natural operator space structure. Let us describe this  explicitly. We
embed $\ell_{\infty}^k$ as diagonal maps in $M_k$. Let $x=\sum_i A_i \otimes e_i\in M_n(\ell_{\infty}^k)=M_n\otimes \ell_{\infty}^k$. Then  we have
\begin{equation}\label{eq:op-space-infty}
\|x\|_n=\left\|\sum_{i} A_i \otimes
|i\rangle\langle i|\right\|_{M_{nk}}=\max_i\|A_i\|_{M_n}.
\end{equation}

The category of Banach spaces and the category of operator spaces essentially deal with the same objects, closed subspaces  of $B(H)$, but they differ through their morphisms. The morphisms in the category of operator spaces are those which
allow a uniform control of all matrix norms, so called
{\em completely bounded
maps}.  A  linear map $u:E\lra F$ between operator spaces is called completely bounded if all the amplifications
$u_n=\id_n \otimes u :M_n\otimes E=
M_n(E)\lra M_n\otimes F = M_n(F)$ remain uniformly  bounded. The cb-norm of $u$ is
then defined as $\|u\|_{cb}=\sup_n\|u_n\|$. We will call $CB(E,F)$
the resulting normed space. It has a natural operator space
structure given by $M_n(CB(E,F))=CB(E, M_n(F))$. We can
analogously define the notion of a complete isomorphism/isometry (see
\cite{EffrosRuan,Pisierbook}).

\

The minimal tensor product of two operator spaces
$E\subset B(H)$ and $F\subset B(K)$ is defined as the operator
space $E\otimes_{\min} F$ with the structure inherited from the
induced embedding $E\otimes F\subset B(H\otimes K).$ In
particular, $M_n(E)=M_n\otimes_{\min} E$ holds for every operator space
$E$. The tensor norm $\min$ in the category of operator spaces
will play the role of the so called $\epsilon$ norm in the
classical theory of tensor norms in Banach spaces \cite{Def}. In particular
$\min$ is injective, in the sense that if $E\subset X$ and
$F\subset Y$ completely isometric (isomorphic), then
$E\otimes_{\min} F\subset X\otimes_{\min} Y$ holds completely
isometrically (isomorphically). The analogue  of the largest
tensor norm $\pi$ for Banach spaces in operator spaces theory
is given by the {\em operator space projective norm $\wedge$}. The norm is defined as
 \[  \|u\|_{M_n(E\otimes_{\wedge} F)}=\inf\{\|\alpha\|_{M_{n,lm}}\|x\|_{M_l(E)}\|y\|_{M_m(F)}\|\beta\|_{M_{lm,n}}:u=\alpha(x\otimes
 y)\beta\}\, ,\]
where $u=\alpha(x\otimes y)\beta$ means the matrix
product  $$u=\sum_{rsijpq}\alpha_{r,ip} \beta_{jq,s}
 |r\rangle\langle s| \otimes x_{ij}\otimes y_{pq}  \in M_n\otimes
 E\otimes F.$$
Both tensor norms, ${\wedge}$ and $\min$, are associative and
commutative and they share the duality relations  given by ${\pi}$ and $\epsilon$. This means that for finite dimensional operator spaces we have the
natural completely isometric identifications
\begin{equation}\label{eq:dual-op-space}
(E\otimes_{\wedge} F)^*=CB^2(E,F; \C)=CB(E,F^*)=E^*\otimes_{\min}
F^*\, .
\end{equation}
Here the matrix norms of the dual operator space $E^*$ of an operator space $E$ are given by
$M_n(E^*)=CB(E,M_n)$.

\

A Banach space $X$ carries many different
operator space structures. This means that there are different isometric inclusions in $B(H)$ with different tail of matrix norms.
Fundamental examples are the\emph{ row} and \emph{column} structures, defined on a Hilbert space $\ell_2^n$. For the row operator spaces $R_n$, we embed $\ell_2^n$ into $M_n$ as a row
 \[  R_n =\{ \sum_k \alpha_k |0\rangle \langle k| :\alpha_k\in
 \mathbb{C} \} \]
and similarly we define the column operator space $C_n$ via
  \[ C_n =\{ \sum_k \alpha_k |k\rangle \langle 0| :\alpha_k\in
 \mathbb{C} \}\, . \]
It can be seen that
\begin{align*}\|\sum_i A_i\otimes e_i\|_{M_m\otimes_{\min} R_n}=\|\sum_i
 A_iA_i^{\dagger}\|^{\frac{1}{2}}  \text{     and     }
 \|\sum_i A_i\otimes e_i \|_{M_m\otimes_{\min} C_n}=\|\sum_i
 A_i^{\dagger}A_i\|^{\frac{1}{2}}.
\end{align*}

We may also need the \emph{intersection} of two operator spaces. Assume that $X$ and $Y$ are injectively embedded in larger topological vector space $V$. Then we may define the norm on $M_n(X\cap Y)=M_n(X)\cap M_n(Y)$  by
 \[ \|(x_{ij})\|_{M_n(X\cap Y)} \, = \, \max\{\|(x_{ij})\|_{M_n(X)},\|(x_{ij})\|_{M_n(Y)}\} \,  .\]
It is easy to see that this definition satisfies Ruan axioms, and moreover for $X\subset B(H)$, $Y\subset B(K)$
 \[ X\cap Y\subset X\oplus Y \subset B(H)\oplus B(K) \subset B(H\oplus K),\]
where the last inclusion is given by diagonal operators, and the first inclusion by identifying elements which are considered equal in the ambient space $V$. Specifically, if we consider $R_n\cap C_n$, we obtain a new operator space structure on $\ell_2^n$, described by
\begin{align*}
\|\sum_i A_i\otimes e_i\|_{M_m\otimes_{\min} R_n\cap C_n}=\max\{\|\sum_i
 A_iA_i^{\dagger}\|^{\frac{1}{2}}, \|\sum_i
 A_i^{\dagger}A_i\|^{\frac{1}{2}}\}.
\end{align*}

In this work we will also use Pisier's operator space $OH$ as a technical tool. We refer to its definition as complex interpolation space $OH=(R,C)_{\frac12}$ and further properties to \cite[Chapter 7]{Pisierbook} and \cite{Pisierbook5}.

\

The operator space $\ell_1^n$ carries a natural operator space structure as the dual of $\ell_{\infty}^n$, i.e. $\ell_1^n=(\ell_{\infty}^n)^*$. Note that for any operator space $X$  the natural operator space structure on $\ell_1(X)\subset (c_0\otimes_{\min}X^*)^*$ is given by the norm  closure of $\ell_1\otimes X$. We will write $\ell_1^n(X)$ for the space given by $n$-tuples of elements in $X$ and observe that by definition \[  (\ell_{\infty}^n\otimes_{\min} X)^*= (\ell^n_\infty(X))^*=\ell^n_1(X^*) = \ell_1^n\otimes_{\wedge} X^*
  \quad \text{  holds completely isometrically. } \]

\

Let  $E_0$ and $E_1$ be two operator spaces so that
$(E_0,E_1)$ is a compatible couple of Banach spaces.  This means that $E_0$ and $E_1$ are injectively embedded in a topological vector space $V$.  On the complex interpolation space  $E_\theta=(E_0,E_1)_\theta$ we have a natural operator space structure given by the formula
 \[ M_n(E_\theta)=(M_n(E_0),M_n(E_1))_\theta \, . \]
This turns $E_\theta$ into an operator spaces (see \cite{Pisierbook}, Chapter 2). As an application we observe that \[  \ell_2^n(\ell_\infty)=(\ell_1^n(\ell_\infty), \ell_\infty^n(\ell_\infty))_\frac{1}{2} \]
carries a natural operator spaces structure. Using standard
interpolation theory (a clever application of the three line lemma), this implies
\begin{align} \label{intp1}
\|id:\ell_\infty^n(\ell_\infty)\rightarrow \ell_2^n(\ell_\infty)\|_{cb}=\sqrt{n} \text{   and    } \|\id:\ell_2^n(\ell_\infty)\rightarrow \ell_1^n(\ell_\infty)\|_{cb}= \sqrt{n}.
\end{align}

\smallskip

\subsection{Connections to the physical problem}

As it was shown in \cite{JPPVW}, we understand a Bell inequality (or more precisely the coefficients of a potential Bell inequality) as an element
 \[ M=\sum_{x,y; a,b=1}^{N,K}M_{x,y}^{a,b}(e_x\otimes e_a)\otimes (e_y\otimes e_b)\in \ell_1^N(\ell_\infty^K)\otimes \ell_1^N(\ell_\infty^K).\]
Looking at violations, we study the rate
 \begin{equation}\label{quotient min-epsilon}
  \viol(M) = \frac{\|M\|_{\min}}{\|M\|_{\epsilon}}.
 \end{equation}

Let us recall the following result.

\begin{theorem}\cite[Corollary 4 + Lemma 1]{JPPVW}\label{connection:min-epsilon}

Given an element $M=\sum_{x,y; a,b=1}^{N,K}M_{x,y}^{a,b}(e_x\otimes e_a)\otimes (e_y\otimes e_b)\in \ell_1^N(\ell_\infty^K)\otimes \ell_1^N(\ell_\infty^K)$ such that $\frac{\|M\|_{min}}{\|M\|_{\epsilon}}\geq C$, we can define a Bell inequality $\hat{M}$ with $N$ inputs and $K+1$ outputs (just completing with zeros) and verifying

\begin{equation}\label{violation}
LV(\hat{M})=\sup_{Q\in \mathcal{Q}}\frac{|\langle \hat{M},Q\rangle|}{\sup_{P\in\mathcal{L}} |\langle \hat{M},P\rangle|}\succeq \frac{C}{16}.
\end{equation}

Furthermore, if the Hilbert space dimension required in Equation (\ref{quotient min-epsilon}) is $n$, Equation (\ref{violation}) can be obtained with a Hilbert space dimension lower or equal than $2n$.

\end{theorem}

\begin{remark}
The reader should  note that the meaning of $LV(M)$ here is different to the one in \cite{JPPVW}. In the previous work, we used $LV(M)$ to denote the large violation of $M$ over incomplete probability distributions (see \cite[Definition 1]{JPPVW}). Here $LV(M)$ represents the violation of $M$ over the (complete) probability distributions. We will not deal with the incomplete probabilities here.

\end{remark}

In this paper we will also consider the problem of studying violation for a fixed state. This motivates the following definition

\begin{definition}\label{violation_M_over_rho}
Let $k$ be a natural number and let $\rho$ be a state acting on $\ell_2^k\otimes_2 \ell_2^k$. Given a Bell inequality $M=(M_{x,y}^{a,b})_{x,y; a,b=1}^{N,K}$, we define \emph{the largest violation of $M$ over $\rho$} as

\begin{equation}\label{violation-rho}
LV_\rho(M)=\sup_{Q\in \mathcal{Q}_\rho}\frac{|\langle M,Q\rangle|}{\sup_{P\in\mathcal{L}} |\langle M,P\rangle|},
\end{equation}
where $\mathcal{Q}_\rho=\{(tr(E_x^a\otimes F_y^b\rho)_{x,y;a,b=1}^{N,K}: \{E_x^a\}_{x,a=1}^{N,K}, \{F_y^b\}_{y,b=1}^{N,K} \text{  POVM's   }  \}.$

\end{definition}

We will be mainly interested in the particular case where our state is the maximal entangled state $\rho=|\psi_k\rangle\langle\psi_k|$, where $|\psi_k\rangle=\frac{1}{\sqrt{k}}\sum_{i=1}^k|ii\rangle$.  We will abuse the notation by writing  $LV_{|\psi\rangle}(M)$ instead of  $LV_{|\psi\rangle \langle \psi|}(M)$ for any pure state. Furthermore, we will be interested in identifying Bell inequalities where the dimension free version
\begin{equation}\label{GHZ}
LV_{|\psi\rangle}(M)=\sup_kLV_{|\psi_k\rangle}(M)
\end{equation}
remains bounded. In fact similar objects have been studied in
operator space theory (see \cite{JuPi}). Given two operator spaces $X$ and $Y$ and $a\in X\otimes Y$, we may define a modified $\min$-norm
 \[ \|a\|_{\psi-\min}=\sup \{ \langle\psi_k|(u\otimes v)(a)|\psi_k\rangle\},\]
where the $\sup$ runs over all $k\in\mathbb{N}$ and all completely contractions $u: X\rightarrow M_k$ and $v: Y\rightarrow M_k.$ The following connection between this modified min norm and violation follows easily from the fact
that for every POVM $\{E_x^a\}_{x,a=1}^{N,K}$ in $M_k$, the application $u:\ell_1^N(\ell_\infty^K)\rightarrow M_k$ given by $u(e_x\otimes e_a)=E_x^a$ for every $x,a$, is a completely contraction (see \cite[Section 8]{JPPVW}).

\begin{lemma}\label{connection:maxentanglement}
Given an element $M=\sum_{x,y; a,b=1}^{N,K}M_{x,y}^{a,b}(e_x\otimes e_a)\otimes (e_y\otimes e_b)\in \ell_1^N(\ell_\infty^K)\otimes \ell_1^N(\ell_\infty^K)$, such that $\|M\|_{\psi-\min}\leq C$, then $$\sup_k\sup_{\mathcal{Q}_{|\psi_k\rangle}}|\langle M,Q\rangle|\leq C.$$
\end{lemma}

The following lemma will be very useful in Section \ref{Violation vs GHZ}.

\begin{lemma}\label{maxentanglement}
For every $a\in (R_n\cap C_n)\otimes (R_n\cap C_n)$ we have $\|a\|_{\psi-\min}\leq 4\|a\|_\epsilon.$

\end{lemma}

\begin{proof}
First of all note that it is enough to consider $a=\id=\sum_{i=1}^n e_i\otimes e_i$. Indeed, for every $a$ we can consider the Hilbert Schmidt decomposition and write $a=\sum_{i=1}^n\alpha_ie_i\otimes U(e_i)$ for some unitary operator $U:\ell_2^n\rightarrow \ell_2^n$ and coefficients $(\alpha_i)_{i=1}^n$ verifying $|\alpha_i|\leq \|a\|_\epsilon$ for every $i=1,\cdots ,n$. But it is easy to see that we can consider the coefficients and the unitary operator as a part of the completely contractions in the definition of $\|a\|_{|\psi\rangle}$.

Let $u:R_n\cap C_n\to M_k$ be a complete contraction. According to the Wittstock extension theorem (due independently to Haagerup, Paulsen and Wittstock see \cite{Paulsen}) and the definition of $R_n\cap C_n\subset R_n\oplus C_n$,  we can extend $u:R_n\cap C_n\to M_k$ to $R_n\oplus C_n$ and find a decomposition $u=u_c+u_r$, where $u_c:C_n\to M_k$ and $u_r:R_n\to M_k$ are complete contractions. Therefore  it is enough to consider the norm $\|id \|_{\psi-\min}$ of the four spaces $R_n\otimes_{\psi-\min} C_n$, $R_n\otimes_{\psi-\min} R_n$, $C_n\otimes_{\psi-\min} R_n$, $C_n \otimes_{\psi-\min} C_n$. For this, let $(X_i)_{i=1}^n, (Y_i)_{i=1}^n\subset M_k$ be operators.
Then, we deduce from the Cauchy-Schwartz inequality that
\begin{align*}
 &|\langle\psi_k|\sum_{i=1}^nX_i\otimes Y_i |\psi_k\rangle|=\frac{1}{k}|\sum_{r,s=1}^k\sum_{i=1}^nX_i(r,s)Y_i(r,s)|\\
 &\leq |\frac{1}{k}\sum_{r,s=1}^k\sum_{i=1}^nX_i(r,s)\overline{X_i(r,s)}|^\frac{1}{2}\, \, |\frac{1}{k}\sum_{r,s=1}^k\sum_{i=1}^nY_i(r,s)\overline{Y_i(r,s)}|^\frac{1}{2}.
\end{align*}
On the other hand, for every  $(X_i)_{i=1}^n\subset M_k$,
 \[ \frac{tr}{k}(\sum_{i=1}^nX_iX_i^*)=\frac{tr}{k}(\sum_{i=1}^nX_i^*X_i)=
\frac{tr}{k}\sum_{r,s=1}^k\sum_{i=1}^nX_i(r,s)\overline{X_i(r,s)}.\] The statement follows from the trivial inequality
 \[  |\frac{tr}{k}(\sum_{i=1}^nX_iX_i^*)|\,\le \, \min\{\|\sum_{i=1}^n X_i^*X_i\|,\|\sum_{i=1}^n X_iX_i^*\|\} \,. \]
Hence for all four possibilities we obtain $\|id\|_{\psi-\min}\le 1$. \end{proof}


\section{Main result}\label{Main result}

\subsection{Proof of the main result}

In this section we will present a first proof of Theorem \ref{Theorem 1}. The key point of the proof is the construction of a complemented copy of $\ell_2^n$ into $\ell_1^n(\ell_\infty^n)$ (see Theorem \ref{complemented}). This result will be also very useful in Sections \ref{The construction} and \ref{Violation vs GHZ}. At the end of this section, we will discuss the optimality of our result and some interesting open questions.

We refer to Section \ref{The construction} for more information regarding the constants in Theorem \ref{Theorem 1}. Note that according to Theorem \ref{connection:min-epsilon}, Theorem \ref{Theorem 1} follows from the next result.


\begin{theorem}\label{Theorem min-epsilon}
For every $n\in \mathbb{N}$, there exists an element $M\in \ell_1^n(\ell_\infty^n)\otimes \ell_1^n(\ell_\infty^n)$ such that
 \[ \frac{\|M\|_{\min}}{\|M\|_\epsilon}\succeq \frac{\sqrt{n}}{\log n}\, .\]
Furthermore, this can be achieved with a Hilbert space of dimensions $n$.
\end{theorem}

\begin{remark}
Actually, the proof of the previous theorem guarantees that we can get Theorem \ref{Theorem 1} with a Hilbert space of dimension $2n$. However, we will see in Section \ref{The construction} that  $n+1$ dimensions suffices.
\end{remark}

The key point to prove Theorem \ref{Theorem min-epsilon} is the following result.

\begin{theorem}\label{complemented}
There exist $\delta\in (0, \frac{1}{2})$ and a universal constant $C$, such that, for every $n$, we have a Hilbert space $H_n$ of dimension $\delta n$ and applications $V: H_n\rightarrow \ell_1^n(\ell_\infty^n)$ and $V^*:\ell_1^n(\ell_\infty^n)\rightarrow H_n$ such that $\|V\|\leq C \sqrt{\log n}$, $\|V^*\|\leq 1$ and $V^*V=\id_{H_n}$. In particular, there exists a $C\sqrt{\log n}$- complemented copy of $\ell_2^{\delta n}$ into $\ell_1^n(\ell_\infty^n)$.
\end{theorem}

\begin{remark} {\rm
Theorem \ref{complemented} has a very interesting physical interpretation. Recall that the natural space hosting joint probability distributions $(P(a,b|x,y))_{x,y;a,b=1}^{N,K}$ of Alice and Bob is the tensor product $\ell_\infty^N(\ell_1^K)\otimes \ell_\infty^N(\ell_1^K)$. Following this idea different descriptions of Nature are expressed through different tensor norms on this space. Roughly speaking  $\epsilon$ corresponds to a local model of Nature and $\min$ to a quantum description (see Subsection \ref{NSG} for the rigorous formalization of this idea). However, it is conceivable that other models of nature are associated to other tensor norms (see Section \ref{section-gamma} for such an interpretation).
Our complementation  result shows that on the subspace constructed in Theorem \ref{complemented} the norms of the corresponding probabilities can be identified, up to a logarithmic term, by the $\alpha$-norm on the Hilbert space tensor product $\ell_2^n\otimes_{\alpha}\ell_2^n$. In general, tensor norms on Hilbert spaces are much easier to calculate.
}\end{remark}

We will recall Chevet's inequality, which will be frequently used in this paper (see e.g. \cite{LedouxTalagrand}). For its formulation, we should recall the notation  of the weak-$\ell_2$ norm which is defined as
  $$w_2((x_s)_s;
 X)=\sup \left\{ \left(\sum_s|x^*(x_s)|^2\right)^\frac{1}{2}:x^*\in X^*, \|x^*\|\leq 1\right\} $$
for every  sequence $(x_s)_s$ in a Banach space $X$.

\begin{theorem}[Chevet's inequality]\label{chevet}
There exists a universal constant $b$ such that for every Banach spaces $E, F$ and every sequence $(g_{s,t})_{s,t}$ of independent normalized gaussian random variables
 \begin{align*}
 \|\sum_{s,t} g_{s,t}
 x_s\otimes y_t\|_{E\otimes_{\varepsilon}F} \le bw_2((x_s)_s;
 E)\|\sum_{t} g_ty_t\|_F+bw_2((y_t)_t; F)\|\sum_s g_s x_s\|_E.
 \end{align*}
Here  $b=1$ for real Banach spaces and  $b=4$ for complex spaces.
\end{theorem}

To prove Theorem \ref{complemented}, we will use the following three lemmas.

\begin{lemma}\cite[Lemma 4]{JPPVW}\label{gaussian}
There exits $\delta\in (0,1/2)$ with the following property:
Given natural numbers $n\leq m$ and a family of normalized real gaussian random variables $(g_{ij})_{i,j=1}^{n,m}$, let $G=\sum_{i,j=1}^{n,m}g_{ij}e_i\otimes e_j$ be an operator from $\ell_2^n$ to $\ell_2^m$. Then, ``with highly probability''\footnote{High probability means here that the probability tends to $1$ exponentially fast as $m\rightarrow \infty$ (see Theorem 4.7 in \cite{Pisierbook3}).}, there exists an operator $v_n:H_n\lra \ell_2^n$ such that $v_n^*\frac{1}{m}G^*Gv_n=\id_{H_n}$ and $\|v_n\|\leq 2$, where we denote $H_n=\ell_2^{[\delta n]+1}.$ Here $[x]$ denotes the entire part of real number $x$.
\end{lemma}

\begin{lemma}\label{first}
Let $(g_{i,j}^k)_{i,j,k=1}^n$ be a family of independent and normalized real gaussian variables and the map $G$ defined by
 \[  G(e_k)=\sum_{i,j=1}^ng_{i,j}^ke_i\otimes e_j\]
for every $k=1,\cdots ,n$. Then, there exists a universal constant $C_1> 0$ such that
\begin{align*}
 \mathbb{E}\|G:\ell_2^n\rightarrow \ell_2^n(\ell_\infty^n)\|\leq C_1 \sqrt{n \log n}.
 \end{align*}
In particular, $\mathbb{E}\|G:\ell_2^n\rightarrow \ell_1^n(\ell_\infty^n)\|\leq C_1 n\sqrt{\log n}$.

\end{lemma}

\begin{proof}
Chevet's inequality implies that
\begin{align*}
 \mathbb{E} \|G\|
  \le   w_2((e_i)_i;\ell_2^n )\,\,  \mathbb{E} \|\sum_{i,j=1}^ng_{i,j}e_i\otimes e_j\|_{\ell_2^n(\ell_\infty^n)}
  +\omega_2((e_i\otimes e_j)_{i,j};\ell_2^n(\ell_\infty^n))\,\,  \mathbb{E} \|\sum_{i=1}^ng_ie_i\|_{\ell_2^n}.
\end{align*}
It is well known that $w_2((e_i)_i;\ell_2^n)=1$ and $\mathbb{E} \|\sum_{i=1}^ng_ie_i\|_{\ell_2^n}\le \sqrt{n}$. On the other hand, it is immediate to check that $\omega_2((e_i\otimes e_j)_{i,j}; \ell_2^n(\ell_\infty^n))=1 $. Then it suffices to show
 \[  \mathbb{E} \|\sum_{i,j=1}^ng_{i,j}e_i\otimes e_j\|_{\ell_2^n(\ell_\infty^n)} \preceq \sqrt{n\log n}\, .\]
Indeed, using the  well-known estimate
$\mathbb{E}\|\sum_{i=1}^ng_ie_i\|_{\ell_\infty^n}\preceq\sqrt{\log n}$ (see e.g.  \cite[Page 15]{Tomczak}), we have
 \begin{align*}
 \mathbb{E}\|\sum_{i,j=1}^ng_{i,j}e_i\otimes e_j\|_{\ell_2^n(\ell_\infty^n)} &\le \sqrt{n}\,
 \mathbb{E}\|\sum_{i,j=1}^ng_{i,j}e_i\otimes e_j\|_{\ell_{\infty}^n(\ell_\infty^n)} \preceq \sqrt{n} \sqrt{\log n^2}\, .
 \end{align*}
The second assertion follows from  $\|id:\ell_2^n(\ell_\infty^n)\rightarrow \ell_1^n(\ell_\infty^n)\|\leq \sqrt{n}$.\end{proof}

\begin{lemma}\label{second}
Let $(g_{i,j}^k)_{i,j,k=1}^n$ be a family of independent and normalized real gaussian variables and let the application defined by $G^*(e_i\otimes e_j)=\sum_{k=1}^ng_{i,j}^ke_k$ for every $i,j=1,\cdots ,n$. Then,
\begin{align*}
\mathbb{E}\|G^*:\ell_1^n(\ell_2^n)\rightarrow \ell_2^n\|\leq C_2 \sqrt{n}
\end{align*}
holds with a universal constant $C_2$. In particular
$\mathbb{E}\|G^*:\ell_1^n(\ell_\infty^n)\rightarrow \ell_2^n\|\leq C_2 n$.
\end{lemma}

\begin{proof} According to  Chevet's inequality, we have
\begin{align*}
 \mathbb{E} \|G^*\| \leq  \omega_2((e_i\otimes e_j)_{i,j}; \ell_\infty^n(\ell_2^n))\, \mathbb{E} \|\sum_{i=1}^ng_ie_i\|_{\ell_2^n} + \omega_2((e_i)_i;\ell_2^n)\, \mathbb{E} \|\sum_{i,j=1}^ng_{i,j}e_i\otimes e_j\|_{\ell_\infty^n(\ell_2^n)} .
\end{align*}
Using the simple estimates mentioned in the proof of the previous Lemma, it is enough to see that $\omega_2((e_i\otimes e_j)_{i,j}; \ell_\infty^n(\ell_2^n))\leq 1$  and   $\mathbb{E}\|\sum_{i,j=1}^ng_{i,j}e_i\otimes e_j\|_{\ell_\infty^n(\ell_2^n)}\preceq \sqrt{n}.$ Indeed, for the first one just note that $\|id: \ell_2^n(\ell_2^n)\rightarrow \ell_\infty^n(\ell_2^n)\|\leq 1$. The other inequality
 \[  \mathbb{E}\|\sum_{i,j=1}^ng_{i,j}e_i\otimes e_j\|_{\ell_\infty^n(\ell_2^n)}  =
 \mathbb{E}\|\sum_{i,j=1}^ng_{i,j}e_i\otimes e_j\|_{\ell_\infty^n\otimes_{\epsilon} \ell_2^n}
 \preceq n\]
follows easily from a further application of Chevet's inequality.  For the last assertion just note that $\|id:\ell_1^n(\ell_\infty^n)\rightarrow \ell_1^n(\ell_2^n)\|\leq \sqrt{n}$.
\end{proof}

Now we have all the required ingredients for the proof of  Theorem \ref{complemented}.

\begin{proof}[Proof of Theorem \ref{complemented}]
According to Chebyshev's inequality, we may find a random matrix  $(g(\omega)_{i,j}^k)_{i,j,k=1}$ verifying Lemmas \ref{gaussian}, \ref{first} and \ref{second} simultaneously with a slight modification in the constants for the expectation. Then, we define $W:\ell_2^n\rightarrow \ell_1^n(\ell_\infty^n)$ and $W^*:\ell_1^n(\ell_\infty^n)\rightarrow \ell_2^n$ respectively, as
  \begin{align*}
  W(e_k)& =\frac{1}{n}\sum_{i,j=1}^ng(\omega)_{i,j}^ke_i\otimes e_j    \quad \mbox{and} \quad
    W^*(e_i\otimes e_j) =\frac{1}{n}\sum_{l=1}^ng(\omega)_{i,j}^le_l
 \end{align*}
for $k=1,...,n$ and $1\le i,j\le n$. According to
Lemma \ref{first} and Lemma \ref{second} we know that $\|W\|\preceq \log n$ and $\|W^*\|\preceq 1$. On the other hand, by Lemma \ref{gaussian} applied to $m=n^2$, we obtain a subspace $H_n$ of $\ell_2^n$ and an operator $v_n:H_n\rightarrow \ell_2^n$ with $\|v_n\|\leq 2$ such that $v_n^*(V^*V)v_n=\id_{H_n}$. Thus, defining $V=Wv_n$ and $V^*=v_n^*W^*$, we obtain the assertion.
\end{proof}

In order to work with the operator space minimal tensor norm on $\ell_1^n(\ell_\infty^n)\otimes \ell_1^n(\ell_\infty^n)$, we will have to estimate the cb-norm of $V$ and $V^*$. The following lemma allows us to compute the cb-norm of an operator $T:\ell_1^n(\ell_\infty^n)\rightarrow R\cap C$.

\begin{lemma}\label{Grothendieck}
There exists a universal constant $C_3> 0$, such that for every measure space $(\Omega,\mu)$, for every natural numbers $k\in \mathbb{N}$ and every operator $T:L_1(\Omega,\ell_\infty^k)\rightarrow \ell_2$,
\begin{align*}
\|T:L_1(\Omega,\ell_\infty^k)\rightarrow R\cap C\|_{cb}\leq C_3 \|T\|.
\end{align*}
Furthermore, $C_3=K_L$ is the constant in little Grothendieck theorem.
\end{lemma}

\begin{proof} By approximation it suffices to prove the assertion  for $\ell_1^n(\ell_\infty^k)$ and arbitrary natural number $n$. Now, given an operator $T:\ell_1^n(\ell_\infty^k)\rightarrow \ell_2$, it is immediate that $\|T\|=\sup_i\|T_i\|$ and $\|T\|_{cb}=\sup_i\|T_i\|_{cb}$ where $T_i$ is the associated operator $T_i:\ell_\infty^k\rightarrow \ell_2$ defined by $T_i(e_j)=T(e_i\otimes e_j)$ for every $j=1,\cdots ,k$ and $i=1,..,n$. Now, according to the little Grothendieck theorem,
the $2$-summing norm of $T_i$ is bounded by $K\|T_i\|$, and hence  $T_i=u_iD_{\si_i}$
factors through a diagonal map $D_{\si_i}(e_j)=\si_i(j)e_j$ (see \cite{Pisierbook4}) with  $\|u_i\|\le 1$ and  $\|\si_i\|_2\le K\|T_i\|$. Let $a_j\in M_k$. Then we have
 \begin{align*}
  \|\sum_j |\si_i(j)|^2 a_j^*a_j\|_{M_k}
  &\le \sum_j |\si_i(j)|^2 \sup_j \|a_j\|^2 \, .
  \end{align*}
Hence we have $\|D_{\si_i}:\ell_{\infty}^n\to C_n\|_{cb}\le \|\si_i\|_2$. On the other hand, it is well known (see \cite{Pisierbook}) that $\|u:\ell_2^n\rightarrow \ell_2^n\|=\|u:C_n\rightarrow C_n\|_{cb}$ for every operator $u$. Thus, $\|T_i:\ell_{\infty}^n\rightarrow C_n\|_{cb}\leq K\|T_i\|$ for every $i=1,..,n$. The estimate for $R_n$ is similar. The result follows by the definition of $R_n\cap C_n$.
\end{proof}

\begin{remark}\label{Grothendieck2}
The same proof works if we replace $R_n\cap C_n$ by  $OH_n$.
\end{remark}

We prove now Theorem \ref{Theorem min-epsilon}:

\begin{proof}[Proof of Theorem \ref{Theorem min-epsilon}]
Let us consider the element $$M=(V\otimes V)(a)\in \ell_1^n(\ell_\infty^n)\otimes \ell_1^n(\ell_\infty^n),$$ where $a=\id_{\ell_2^{\delta n}\otimes \ell_2^{\delta n}}$. It is enough to show:
\begin{enumerate}
\item[a)] $\|M\|_{\ell_1^n(\ell_\infty^n)\otimes_\epsilon \ell_1^n(\ell_\infty^n)}\preceq \log n$ and
\item[b)] $\|M\|_{\ell_1^n(\ell_\infty^n)\otimes_{\min} \ell_1^n(\ell_\infty^n)}\succeq \sqrt{n}$.
\end{enumerate}
Observe that  a) follows from Lemma \ref{second}. Indeed,
\begin{align*}
\|M\|_{\ell_1^n(\ell_\infty^n)\otimes_\epsilon \ell_1^n(\ell_\infty^n)}= \|(V\otimes V)(a) \|_{\ell_1^n(\ell_\infty^n)\otimes_\epsilon \ell_1^n(\ell_\infty^n)}\leq \|V\|^2 \|a\|_{\ell_2^{\delta n}\otimes_\epsilon \ell_2^{\delta n}}\preceq \log n.
\end{align*}
For b) we recall that by definition
\begin{align*}
\|M\|_{\ell_1^n(\ell_\infty^n)\otimes_{\min} \ell_1^n(\ell_\infty^n)}=\sup \{\|(u\otimes v)(M)\|_{B(H)\otimes_{min}B(H)}\},
\end{align*}
where the $\sup$
runs over Hilbert spaces $H$ and completely contractions ${u,\! v\!:\!\ell_1^n(\ell_\infty^n)\!\rightarrow \! B(H)}$. Lemma \ref{first} and Lemma \ref{Grothendieck} tell us that $V^*:\ell_1^n(\ell_\infty^n)\rightarrow R_n\subseteq M_n$ verifies that $\|V\|_{cb}\preceq 1$. This implies
 \[  \|M\|_{\ell_1^n(\ell_\infty^n)\otimes_{\min} \ell_1^n(\ell_\infty^n)}\succeq \|(V^*\otimes V^*)(M)\|_{R_n\otimes_{\min} R_n} \,  . \]
But then
 \begin{align*}
 \|(V^*\otimes V^*)(V\otimes V) (a)\|_{R_n\otimes_{\min} R_n}= \|a\|_{R_{\delta n}\otimes_{\min} R_{\delta n}}=\|a\|_{\ell_2^{\delta n}\otimes_2 \ell_2^{\delta n}}=\sqrt {\delta n}\succeq \sqrt {n}.
\end{align*}
Thus, we conclude that $\|M\|_{\ell_1^n(\ell_\infty^n)\otimes_{\min} \ell_1^n(\ell_\infty^n)}\succeq \sqrt{n}$. \end{proof}

\begin{problem}\label{problem1}
Theorem \ref{Theorem min-epsilon} exactly says that:
\begin{align*}
\|\id\otimes \id:\ell_1^n(\ell_\infty^n)\otimes_\epsilon \ell_1^n(\ell_\infty^n)\rightarrow\ell_1^n(\ell_\infty^n)\otimes_{\min} \ell_1^n(\ell_\infty^n)\|\succeq \frac{\sqrt n}{\log n}.
\end{align*} We don't know whether it's possible to improve this order to $\frac{n}{\log^{\al}n}$. According to Theorem \ref{optimality-SDP rankII} we cannot remove a log term for elements $M$ of rank $n$.
\end{problem}

However, for our specific $M$  we cannot improve the violation estimate.

\begin{prop}\label{optimality}
The element $M=(V\otimes V)(a)\in \ell_1^n(\ell_\infty^n)\otimes \ell_1^n(\ell_\infty^n)$  in Theorem \ref{Theorem min-epsilon} verifies
$$\|M\|_{\ell_1^n(\ell_\infty^n)\otimes_{\min} \ell_1^n(\ell_\infty^n)}\preceq \log n \sqrt{n}.$$
\end{prop}
To prove Proposition \ref{optimality} we need the following technical lemma:

\pagebreak[1]

\begin{samepage}
\begin{lemma}\label{lemma-optimality}${\atop}$

{\rm a)} Let $X$ be an operator space and $T:X\rightarrow \ell_2^n(\ell_\infty^n)$ a linear map. Then $$\|T\|_{cb}\leq \sqrt{n}\, \|T\|.$$

{\rm b)} Let  $T:OH\rightarrow \ell_2^n(\ell_\infty^n)$ be a linear map. Then  $\|T\|_{cb}\leq n^\frac{1}{4}\|T\|$.
\end{lemma}\end{samepage}

\begin{proof} For the proof  of a), we consider the factorization
\begin{align*}
\id\circ \id\circ T:X\rightarrow \ell_2^n(\ell_\infty^n)\stackrel{id}{\rightarrow} \ell_\infty^n(\ell_\infty^n) \stackrel{id}{\rightarrow}\ell_2^n(\ell_\infty^n).
\end{align*}
Obviously,
\begin{align*}
 \|T\|_{cb}&\leq \|\id\circ T:X\rightarrow \ell_\infty^n(\ell_\infty^n)\|_{cb}\|id:\ell_\infty^n(\ell_\infty^n) \rightarrow\ell_2^n(\ell_\infty^n)\|_{cb} \\
 &\le \|\id\circ T:X\rightarrow \ell_\infty^n(\ell_\infty^n)\| \, \sqrt{n} \, \le \,
 \|T:X\rightarrow \ell_2^n(\ell_\infty^n)\| \, \sqrt{n}  \, .
 \end{align*}
Here we used  that $\ell_\infty^n(\ell_\infty^n)=\ell_{\infty}^{n^2}$ is a commutative $C^*$-algebra and estimate \eqref{intp1} from Section \ref{Mathematical Tools}.  For the proof of
b), we consider a map  $T:OH\rightarrow \ell_2^n(\ell_\infty^n)$. Then, the map
\begin{align*}
T\overline{T^*}:\overline{\ell_2^n(\ell_1^n)}\rightarrow \overline{OH^*}\simeq OH\rightarrow \ell_2^n(\ell_\infty^n)
\end{align*}
verifies  $$\|T\overline{T^*}\|_{cb}=\|T\|_{cb}^2$$
according to \cite[Proposition 7.2]{Pisierbook}. Applying a), we deduce
 \begin{align*}
 \|T\|_{cb}^2=\|T\overline{T^*}\|_{cb}\leq \sqrt{n}\|T\overline{T^*}\|=\sqrt{n}\|T\|^2.
 \end{align*}
Hence $\|T\|_{cb}\leq n^\frac{1}{4}\|T\|$.\end{proof}

Now, we can show the optimality of our result when we consider the element $M$.

\begin{proof}[Proof of Proposition \ref{optimality}]
According to Lemma \ref{first} and Lemma \ref{lemma-optimality},
\begin{align*}
 \|V:OH_n\rightarrow \ell_1^n(\ell_\infty^n)\|_{cb}& \leq\|V:OH_n\rightarrow \ell_2^n(\ell_\infty^n)\|_{cb}\|\id:\ell_2^n(\ell_\infty^n)\rightarrow \ell_1^n(\ell_\infty^n)\|_{cb}\\
 &\preceq n^{\frac14} \frac{\sqrt{\log n}}{\sqrt{n}}
  \, n^{\frac{1}{2}}=n^\frac{1}{4}\sqrt{\log n}.
\end{align*}
Therefore, we obtain
\begin{align*}
 \|M\|_{\ell_1^n(\ell_\infty^n)\otimes_{\min} \ell_1^n(\ell_\infty^n)}& =\|(V\otimes V)(a)\|_{\ell_1^n(\ell_\infty^n)\otimes_{\min} \ell_1^n(\ell_\infty^n)}\preceq \sqrt{n}\log n\|a\|_{OH_{\delta n}\otimes_{\min}OH_{\delta n}}\\
 & =\sqrt{n}\log n \|a\|_{\ell_2^{\delta n}\otimes_{\epsilon}\ell_2^{\delta n}}\preceq \sqrt{n}\log n. \qedhere
\end{align*}
\end{proof}




\section{Explicit form of the violation}\label{The construction}

In this section we will prove Theorem \ref{Theorem-construction}. It turns out that very little knowledge on a given state is required to  have violation and that, in a certain sense, entanglement and violation are rather independent.

\subsection{Constructing violation}\label{Constructing violations}

In our approach violation of Bell inequalities is obtained by constructing Bell inequalities and the corresponding positive operator valued measurements simultaneously. The explicit form is derived from an explicit form of Wittstock/Paulsen extension theorem. Another key ingredient is the factorization structure of the coefficients. Before proving the result, consider the following two remarks.

\begin{remark}(\emph{Concerning the universal constants in Theorem \ref{Theorem-construction}})
Although we are not going to give an explicit value of the constants $C$ and $K$,  we would like to point out that for all the constant we are going to use there are explicit bounds in the literature.
\end{remark}

\begin{remark}\label{Different Random variables}(\emph{Random variables})
Although we have proved Theorem \ref{Theorem 1} (via Theorem \ref{Theorem min-epsilon}) using gaussian variables, it is well known that the same estimations work for Bernoulli variables (\cite{Tomczak})  and random unitaries (\cite{MarcusPisier}, \cite{MariusHabi}) (in this last case one has to normalize by a factor $\sqrt{n}$). Thus Theorem \ref{Theorem-construction} can be stated using Bernoulli variables $(\epsilon_{x,a}^k)_{x,a,k}$ (as it is stated), gaussian variables $(g_{x,a}^k)_{x,a,k}$ and random unitaries $(U_k(x,a))_{x,a,k}$, where in this last case the probability is defined by the Haar measure on $\prod_{i=1}^k\mathbb{U}_n$. It turns out that Bernoulli lead to a slight simplification in the proof of \eqref{violation-constructive2}.
\end{remark}

We will divide the proof of Theorem \ref{Theorem-construction} into three steps.

\begin{step}\label{step}
$\sup\left\{|\sum_{x,y;a,b=1}^{n,n+1}\tilde{M}_{x,y}^{a,b}P(a,b|x,y)|:
P\in \mathcal L\right\}\leq K_1^2\log n.$
\end{step}

\begin{proof} Indeed, in \cite{JPPVW}, Proposition 4), we gave an elementary proof of
\begin{align*}
\sup\left\{|\sum_{x,y;a,b=1}^{n,n+1}\tilde{M}_{x,y}^{a,b}P(a,b|x,y)|: P\in \mathcal L\right\}\leq \|\tilde{M}\|_{\ell_1^n(\ell_\infty^{n+1})\otimes_\epsilon \ell_1^n(\ell_\infty^{n+1})}.
\end{align*}
Moreover,
\begin{align*}
\|\tilde{M}\|_{\ell_1^n(\ell_\infty^{n+1})\otimes_\epsilon \ell_1^n(\ell_\infty^{n+1})}=\|M\|_{\ell_1^n(\ell_\infty^n)\otimes_\epsilon \ell_1^n(\ell_\infty^n)}\leq K_1^2\log n.
\end{align*}
follows immediately from the injectivity of the $\eps$-norm for the first equality and  from  Lemma \ref{first}. Indeed, we refer to \cite{PiVar} for the simple fact that
 \[ \mathbb{E}  \|\sum_{x,a,k} \eps_{x,a}^k e_k\otimes (e_{x}\otimes e_a)\|_{\ell_2^n\otimes_{\eps} \ell_1^n(\ell_{\infty}^n)}
 \le \sqrt{\frac{\pi}{2}}\mathbb{E}
 \|\sum_{x,a,k} g_{x,a}^k e_k\otimes (e_{x}\otimes e_a)\|_{\ell_2^n\otimes_{\eps} \ell_2^n(\ell_{\infty}^n)} \, .\]
Let  $\mathcal{E}:\ell_2^n\to \ell_1^n(\ell_{\infty}^n)$ denotes
the  map $\mathcal{E}(e_k)=\eps_{x,a}^k e_{x}\otimes e_a$. Then
our estimate of order $K_1^2\log n $ follows from
$\|\mathcal{E}^*\mathcal{E}\|\le \|  \mathcal{E}\|^2\le \pi/2
\|G\|^2$.
\end{proof}

\begin{step}\label{step2} Moreover, the operators $\{E_x^a\}_{a,x=1}^{n+1,n}$ in the statement of Theorem \ref{Theorem-construction} define POVM's measurements in $M_{n+1}$.
\end{step}

\begin{proof} By the comments in Step 1, we know that the estimate of Lemma \ref{second} also works for Bernoulli variables replacing the constant $C_2$ with $\sqrt{\frac{\pi}{2}}C_2$. Furthermore, according to Chebyshev's inequality, we may find a random matrix  $(\epsilon_{x,a}^k)_{x,a,k=1}^n$ verifying Lemmas \ref{first} and \ref{second} simultaneously (replacing the constants $c$ for the expectation by $2c$). Let's call $K_1$ and $K_2$ the corresponding final constants.

Thanks to our definition of  $E_x^{n+1}$, it certainly suffices to show  $\sum_{a=1}^{n}E_x^a\leq 1$.
Our proof is motivated by Paulsen description of cb-maps from $\ell_{\infty}^n\to M_k$. Let us fix $x$ and consider the operators
\begin{align*}\tilde{E}_x^a = \frac{1}{nK_2}\sum_{k=1}^n\epsilon_{x,a}^ke_{1,k}=\frac{1}{nK_2}\left( \begin{array}{cccccc}
  \epsilon_{x,a}^1 & \epsilon_{x,a}^2 &\cdots & \epsilon_{x,a}^n \\
  0 & 0 &\cdots & 0 \\
    \vdots & \vdots & \vdots    \\
  0 & 0 &\cdots & 0 \\
\end{array}\right)
\end{align*} for every $a=1,\cdots, n$. We define  operators $\alpha_a^x=\frac{1}{\sqrt{n}}e_{1,1}$ and  $\beta_a^x=\frac{1}{\sqrt{n}}\sum_{k=1}^n\epsilon_{x,a}^ke_{1,k}$ and note that
 \[ \tilde{E}_x^a=\alpha_a^x\beta_a^x  \]
holds for every $a=1,\cdots, n$. Then we observe that
 \[ \sum_{a=1}^n \alpha_a^x(\alpha_a^x)^{\dag}  \,=\, e_{11}\le 1\, .\]
Let $R=(\frac{\eps_{x,a}^k}{\sqrt{n}})_{a,k=1}^n$. According to Lemma \ref{second}, we have $\|R\|\le K_2$ and hence
 \[
  \sum_{a=1}^n  (\beta_a^x)^{\dag} \beta_a^x
 \,=\, R^*R\le K_2^2. \, \]
Then, we define the positive operators
\begin{align*}\hat{E}_x^a&=\left( \begin{array}{ccccccc}
  \alpha_a^x & 0 \\
  (\beta_{a}^x)^{\dag} &  0\\
\end{array}\right)\left( \begin{array}{ccccccc}
  (\alpha_a^x)^{\dag} & \beta_{a}^x \\
  0 &  0\\
\end{array}\right) =\left( \begin{array}{ccccccc}
  \alpha_a^x(\alpha_a^x)^{\dag} & \alpha_a^x\beta_{a}^x \\
  (\alpha_a^x\beta_{a}^x)^{\dag} &  (\beta_{a}^x)^{\dag}\beta_{a}^x\\
\end{array}\right)\\
&=\frac{1}{n}\left( \begin{array}{ccccccc}
  1 &  0 & \cdots & 0 & \epsilon_{x,a}^1 & \cdots & \epsilon_{x,a}^n \\
  0 &  0 & \cdots & 0  & 0 & \cdots & 0\\
  \vdots & \vdots & \vdots  & \vdots & \vdots & \vdots  & \vdots \\
  0 &  0 &\cdots & 0  & 0 & \cdots & 0\\
  \epsilon_{x,a}^1&  0 & \cdots & 0 & 1  & \cdots & \epsilon_{x,a}^1\epsilon_{x,a}^n\\
  \vdots & \vdots & \vdots  & \vdots & \vdots & \vdots  & \vdots \\
  \epsilon_{x,a}^n&  0 & \cdots & 0 & \epsilon_{x,a}^n\epsilon_{x,a}^1  & \cdots & 1\\
\end{array}\right).
\end{align*}
Note that these operator are in $M_{2n}$ for every $a=1,\cdots ,n$. However, we may erase columns and rows and obtain the positive operators $E_x^a$ in $M_{n+1}$ stated in our assertion using the constant $K=2K_2^2$. Thus it remains to show $\sum_{a=1}^n \hat{E}_x^a\le 2$.  We recall that for a positive matrix
 \[ A\, = \, \left(\begin{array}{cc} a&b\\
  b^*&c \end{array}\right) \]
we have $b^*b+bb^*\le a+c$. This implies $A\, \le\,  2
\left(\begin{array}{cc} a& 0 \\ 0 & c \end{array}\right)$, and
hence
 \begin{align*}
 \sum_{a=1}^n \hat{E}_x^a &\le 2
 \left(\begin{array}{cc} \sum_a \al_a^x(\al_a^x)^{\dag} & 0 \\
  0 &  \sum_a (\beta_a^x)^{\dag}\beta_a^x \end{array} \right)
  \, \le \, 2 \, . \qedhere
  \end{align*}
\end{proof}


\begin{step} \label{step3} $\sum_{x,y;a,b=1}^{n,n+1}\tilde{M}_{x,y}^{a,b}\langle
\varphi_\alpha|E_x^a\otimes E_y^b |\varphi_\alpha\rangle\geq
\frac{2}{K^2} \al_1(\sum_{i=2}^{n+1} \al_i)$.
\end{step}

\begin{proof} Let us denote by $E_x^a=(E_x^a(k,l))_{k,l=1}^{n+1}$
the matrix coefficients of our operators from step 2. Using the
fact that the $E_x^a$ are selfadjoint, we deduce
\begin{align*}
 &\sum_{x,y;a,b=1}^{n,n+1}\tilde{M}_{x,y}^{a,b}\langle \varphi|E_x^a\otimes E_y^b
 |\varphi\rangle\\
 &=\al_1^2 \sum_{x,y,a,b=1}^nM_{x,y}^{a,b}E_x^a(1,1)E_y^b(1,1)
 + 2\al_1 \sum_{i=2}^{n+1} \al_i
 \sum_{x,y,a,b=1}^nM_{x,y}^{a,b}E_x^a(1,i)E_y^b(1,i)\\
 &\quad +\sum_{i,j=2}^{n+1}\al_i\al_j \sum_{x,y,a,b=1}^nM_{x,y}^{a,b}
 E_x^a(i,j)E_y^b(i,j)  \\
 &={\rm I} + {\rm II} + {\rm III} \, .
\end{align*}
Let us start with the main term
\begin{align*}
 {\rm II} &= 2\al_1 \sum_{i=2}^{n+1} \al_i
 \sum_{x,y,a,b=1}^nM_{x,y}^{a,b}E_x^a(1,i)E_y^b(1,i) \,=\,
  \frac{2}{K^2n^2} \al_1 \sum_{i=2}^{n+1} \al_i
 \sum_{x,y,a,b=1}^nM_{x,y}^{a,b}\eps_{x,a}^{i-1}\eps_{y,b}^{i-1} \\
 &= \frac{2}{K^2n^4} \al_1 \sum_{i=2}^{n+1} \al_i  \sum_{x,y,a,b=1}^n
 \sum_{k=1}^n
 \eps_{x,a}^k\eps_{y,b}^k \eps_{x,a}^{i-1}\eps_{y,b}^{i-1} \\
 &=
  \frac{2}{K^2n^4} \al_1 \sum_{k=1}^n \sum_{i=2}^{n+1} \al_i
 \sum_{x,y,a,b=1}^n \eps_{x,a}^k\eps_{y,b}^k
 \eps_{x,a}^{i-1}\eps_{y,b}^{i-1}\\
 &=  \frac{2}{K^2n^4} \al_1 \sum_{k=1}^n \sum_{i=2}^{n+1} \al_i
 (\sum_{x,a}\eps_{x,a}^k\eps_{x,a}^{i-1})^2 \\
  &\ge  \frac{2}{K^2n^4} \al_1 \sum_{i=2}^{n+1} \al_i
 (\sum_{x,a}\eps_{x,a}^k\eps_{x,a}^k)^2 \,=\,
 \frac{2}{K^2} \al_1 \sum_{i=1}^{n+1}\al_i \, .
 \end{align*}
To conclude the proof it remains to show that the other two terms
are positive. Indeed for the first term we have
 \begin{align*}
 {\rm II} &= \alpha_1^2\sum_{x,y,a,b=1}^nM_{x,y}^{a,b}E_x^a(1,1)E_y^b(1,1)
 \, =\, \frac{\alpha_1^2}{K^2n^4}\sum_{x,y,a,b,k=1}^n\epsilon_{x,a}^k\epsilon_{y,b}^k\\
 &= \frac{\alpha_1^2}{K^2n^4}\sum_{k=1}^n (\sum_{x,a} \epsilon_{x,a}^k)^2 \, \ge \, 0 \, .
 \end{align*}
For the third and last term we argue similarly,
\begin{align*}
  \,{\rm III}=\,   \sum_{i,j=2}^{n+1} \al_i\al_j
\sum_{x,y,a,b=1}^nM_{x,y}^{a,b}E_x^a(i,j)E_y^b(i,j) \,=\,
 \frac{1}{K^2n^2} \sum_{i,j=2}^{n+1} \al_i\al_j \sum_{x,y,a,b=1}^n M_{x,y}^{a,b}
 \eps_{x,a}^{i-1}k\eps_{y,b}^{j-1}\\
 = \frac{1}{K^2n^3} \sum_{k=1}^n \sum_{i,j=2}^{n+1} \al_i\al_j \!\!\!
 \sum_{x,y,a,b=1}^n\!\!\!
 \eps_{x,a}^k\eps_{y,b}^k \eps_{x,a}^{i-1}\eps_{y,b}^{j-1}
 =
 \frac{1}{K^2n^3} \sum_{k=1}^n \sum_{i,j=2}^{n+1} \al_i\al_j (\sum_{x,a}^n
 \eps_{x,a}^k\eps_{x,a}^{i-1})^2  \ge  0. \,
 \end{align*}

\end{proof}

\begin{remark}
Note that the lower estimate from Step 3 holds for every choice of signs. The exponential estimate for the $\eps$-norm of the Bell inequality in Step 1 and for the POVMs measurements in Step 2 for the Bernouilli variables  follows from the contraction principle and the corresponding exponential estimate for gaussian, an immediate consequence of the deviations inequalities (see \cite{PiVar} and \cite{Kw-Woj}).
\end{remark}


\subsection{Entanglement and quantum nonlocality}\label{Entanglement}

We will now exhibit a family of states $(|\varphi_{\al}\rangle)_{0<\al<1}$ on $\ell_2^{n+1}\otimes
\ell_2^{n+1}$ and show that in certain range $\mu_n\le \al \le \nu_n$ we find coefficients of a Bell inequalities $M\in
\ell_1^{n}(\ell_{\infty}^{n+1})\otimes \ell_1^{n}(\ell_{\infty}^{n+1})$ such that  \[ \frac{\sqrt{n}}{\log^{2}n} \preceq   LV_{|\psi_{\al}\rangle} (M) \,. \] We will study the entropy of entanglement of these states. We refer to Definition \ref{violation_M_over_rho} for the precise
definition of $LV_{|\psi_{\al}\rangle} (M)$ as the maximal violation for a given state $|\psi\rangle$ and propose a measure of
nonlocality, namely \emph{the largest Bell violation that $\rho$ may attain}:
\begin{equation*}\label{violation_rho}
 LV_{|\psi\rangle} = \sup_M LV_\rho(M) .
\end{equation*}
We were very surprised when comparing this measure of violation
with the entropy of entanglement of these states. Let us recall
that for a pure state $|\psi\rangle\in H_A\ten H_B$ the corresponding
density  is defined as
 \[ \rho_A \,=\,  tr_B(|\psi\rangle \langle \psi|) \, .
 \]
Here $tr_B=id\ten tr_{H_B}$ is the partial trace onto the first
Hilbert space $H_A$. For the convenience of the reader let us
recall that for $|\psi\rangle=\sum_{i=1}^{n+1}\al_i |ii\rangle$ we
have
 \[ \rho_{|\psi\rangle} \lel \sum_{i=1}^{n+1} \al_i^2 \pl |i\rangle
 \langle i| \, .\]
The entropy of entanglement is given by the von Neumann entropy of
$\rho_A$, i.e.
 \[ \mathcal{E}(\psi) \lel H((\rho_{\psi})_A) \lel -tr(\rho_A
 \log_2(\rho_A)) \pl .\]
The following Lemma is completely elementary.
\begin{lemma}\label{ealpha} Let
    \[ |\varphi_\alpha\rangle= \alpha |11\rangle+\frac{\sqrt{1-\alpha^2}}{\sqrt{n}}\sum_{i=2}^{n+1}|ii\rangle\in
\ell_2^{n+1}\otimes_2\ell_2^{n+1}\pl .\] Then the function
$f(\al)=\mathcal{E}(\varphi_{\al})$ satisfies
 \begin{enumerate}
  \item[i)] $f(\al)=-\alpha^2\log_2(\alpha^2)-\sum_{i=2}^{n+1}\frac{1-\alpha^2}{n}\log_2(\frac{1-\alpha^2}{n})
  =\alpha^2\log_2(\frac{1}{\alpha^2})+(1-\alpha^2)\log_2(\frac{n}{1-\alpha^2})$;
  \item[ii)] $f_n(0)=\log_2(n)$, $f_n(1)=0$;
  \item[iii)] the function $f$ has a unique maximum
  $f_n(\alpha_n)=\log_2(n+1)$ for the maximally entagled state $\varphi_{\al_n}$, $\alpha_n=\frac{1}{\sqrt{n+1}}$.
  \item[iv)] Let $\eps,\delta >0$, $n\geq 2$, $p\gl 1$ and $\eps\le \log_2(n)$, $2\delta\le
  \log_2^p(n)$. Set  $\mu_n^2=\frac{\eps}{\log_2(n)}$ and
   $1-\nu_n^2=\frac{\delta}{\log_2^p(n)}$. Then
  \[  \log_2(n+1)-\eps-\frac{1}{n\ln 2} \le \log_2(n)-\eps\le
  f(\mu_n) \quad \mbox{and}\quad
 f(\nu_n)\kl 4 (\delta
 \log_2^{1-p}n+\sqrt{\delta}\log_2^{-p/2}n)\pl .\]
 \end{enumerate}
\end{lemma}

Corollary \ref{delta-entangled} follows from the next theorem.

\begin{theorem} Let $n\geq 2$ and $0<\eps< \frac{1}{2}$. Then there exists a
family of pure states $(|\varphi_{\al}\rangle)_{0<\al<1}$ on $\ell_2^{n+1}\otimes
\ell_2^{n+1}$ such that
 \[ LV_{|\varphi_{\al}\rangle} \gl c  \eps  \frac{\sqrt{n}}{\log^2 n} \]
for all $\al$ and
 \[ \{\mathcal{E}(|\varphi_{\al}\rangle)\} \supset
 [ \frac{8\eps}{\log_2(n)}, \log_2(n)-\eps] \pl. \]

\end{theorem}

\begin{proof} Let $|\varphi_{\al}\rangle$ be defined as above, $0<\eps< \frac{1}{2}$,
$p=2$ and $\delta=\eps^2$. Since $n\geq 2$, we trivially have $2\delta\leq \log_2^2(n)$. We consider $\al$ in the range $[ \mu_n,\nu_n]$ and deduce from Theorem \ref{Theorem-construction} the existence of $Q\in \mathcal{Q}_{|\varphi_{\al}\rangle}$ and $\tilde{M}$ such that
 \begin{align*}
  \langle Q,\tilde{M}\rangle & \gl \frac{2}{K^2} \al_1
 \sum_{i=2}^{n+1} \al_i
 \gl \frac{\al \sqrt{1-\al^2}}{K^2} \sqrt{n} \\
 &\gl \frac{\sqrt{n}}{K^2}  \min\{\mu_n\sqrt{1-\mu_n^2},
 \nu_n\sqrt{1-\nu_n^2}\}
 \gl \frac{\sqrt{n}}{K^2} \min\{
 \frac{\sqrt{\eps}}{\sqrt{\log_2n}},
 \frac{\sqrt{\delta}}{\sqrt{\log_2^2n}}\} \pl .
\end{align*}
However, we should also normalize with the $\eps$-norm of
$\tilde{M}$ and have to consider $\tilde{M}=\frac{\tilde{M}}{C \log n}$
instead. This yields
 \[ LV_{|\varphi_{\al}\rangle} \gl c \eps \frac{\sqrt{n}}{\log_2^2n}
 \pl .\]
The intermediate value theorem concludes the proof thanks to (Lemma
\ref{ealpha}, iv)).\end{proof}

Theorem \ref{Theorem-construction} allows us to go a little
further and introduce the following indicator of violation of a
pure state $|\psi\rangle$ given by
 \[ \iviol(|\psi\rangle) \lel \||\psi\rangle\|_{\infty}\||\psi\rangle\|_1 \pl ,\]
where $\||\psi\rangle\|_p=(tr(||\psi\rangle|^p)^{1/p}$ is the $p$-norm of the
Hilbert-Schmidt operator associated with $|\psi\rangle$. Let us note that
for a state $\||\psi\rangle\|_2=1$ and hence, the well-known application
of H\"{o}lder's inequality
 \[ 1\lel \||\psi\rangle\|_2\kl \||\psi\rangle\|_1^{1/2} \||\psi\rangle\|_{\infty}^{1/2}
 \]
shows that $\iviol(\psi)\gl 1$. We may reformulate
Theorem \ref{Theorem-construction}  as follows:

\begin{corollary} Let $n\gl 2$ and  $|\psi\rangle \in \ell_2^n\ten \ell_2^n$ with
$\iviol(|\psi\rangle)\gl 2$. Then
 \[ LV_{|\varphi_{\al}\rangle} \gl \frac{c}{\log n} \iviol(|\psi\rangle) \]
holds for an absolute constant $c$.
\end{corollary}

\begin{proof} Let $\la_i$ be the singular values of $|\psi\rangle$. Then
we deduce from our assumption  that
 \begin{align*}
  \iviol(|\psi\rangle) &= \la_1 (\sum_{i=1}^n \la_i)
  \lel \la_1 (\sum_{i=2}^n \la_i) + \la_1^2
  \kl \la_1 (\sum_{i=2}^n \la_i) + 1 \kl
  \la_1 (\sum_{i=2}^n \la_i) + \frac{1}{2}  \iviol(|\psi\rangle) \pl .
  \end{align*}
Thus we find $\iviol(|\psi\rangle)\le 2\la_1 (\sum_{i=2}^n \la_i)$. We may
diagonalize $|\psi\rangle$ with local operations (i. e. unitary operations
in $u\ten 1$ and $1\ten v$ in $\ell_2^n\ten \ell_2^n$) and assume
that $|\psi\rangle=\la_1|11\rangle+\sum_{i=2}^{n} \la_i |ii\rangle$. Then
Theorem \ref{Theorem-construction} applied for $n-1$ yields the
assertion. \end{proof}

\begin{conc} Large violation of order $\sqrt{n}$ can occur independently of the entropy of entanglement of a pure state, as
long as the state is not too local $(rank(|\psi\rangle)\leq k)$ or too maximally entangled $(\max_i \{\la_i\}\leq \frac{k}{\sqrt{n}})$.
\end{conc}


\section{Nonlocality and the maximally entangled state}\label{Violation vs GHZ}

In this section we will investigate conditions for Bell
inequalities which either avoid violation in the maximal entangled
state or enforce violation to occur on the maximal entangled
state.

\subsection{Unbounded violation away from the maximally entangled state}

The fact that the maximally entangled state is not optimal in
terms of violation is not new. Indeed, there are many examples in
the context of quantum nonlocality where the maximally entangled
state has been shown not to be the most \emph{nonlocal} one. We
can find some of these \emph{anomalies} (\cite{MV}) in the study
of Bell inequalities (\cite{ADGL}), detection loophole
(\cite{Eber}), extractable secrete key (\cite{SGBMPA}), K-L
distance (\cite{AGG}), etc. Here we will show that there are Bell
inequalities which avoid violation of the maximally entangled
state in high dimension. The examples are closely related to
Theorem \ref{Theorem-construction}, but we need more inputs. From
an operator space perspective, we may say that we use a $C\log n$
completely complemented copy of $R_n\cap C_n$.

\begin{theorem}\label{completely complemented}
There exist universal constants $C,D>0$ with the following
property: For $n\in \mathbb{N}$ there are linear maps $S:R_n\cap
C_n\to \ell_{1}^k(\ell_{\infty}^{Dn})\to R_n\cap C_n$ and
$S^*:\ell_1^{k}(\ell_{\infty}^{Dn})\to $ such that
 \[ S^*S=id_{\ell_2^n} \quad \mbox{and}\quad \|S_{cb}^*\|\le C \quad \mbox{and}\quad
   \|S\|_{cb}  \kl  C \sqrt{\log n}  \pl .\]
Moreover, $k\le 2^{D^2n^2}$.
\end{theorem}


Before proving Theorem \ref{completely complemented}, let's see
how to obtain Theorem \ref{Theorem 2}:

\begin{proof}[Proof of Theorem \ref{Theorem 2}]
Following the same steps as in Theorem \ref{Theorem
min-epsilon}, we can prove that the element $G= (S\otimes S)(\sum_k e_k\ten e_k)\in
\ell_1^k(\ell_\infty^{Dn})\otimes \ell_1^k(\ell_\infty^{Dn})$ satisfies
 \[  \|G\|_\epsilon\preceq \log n  \text{       }\text{    and    }
 \text{      }  \|G\|_{\min}\succeq \sqrt{n} \pl .\]
According to Lemma \ref{maxentanglement} and the cb-estimate from Theorem \ref{completely complemented}, we obtain
 \[ \|M\|_{\psi-\min}= \|(S\otimes S)(\sum_k e_k\ten e_k)\|_{\psi-\min}\leq
 \|S\|_{cb}^2\|a\|_{\psi-\min}
 \leq  C \|S\|_{cb}^2 \leq C \log n \pl .\]
Taking $M=\frac{M}{\log n}$ and following Lemma
\ref{connection:min-epsilon}, we obtain $\tilde{M}\in
\ell_1^k(\ell_\infty^{Dn+1})\otimes \ell_1^k(\ell_\infty^{Dn+1})$
such that $LV(\tilde{M})\succeq \frac{\sqrt{n}}{\log n}$. Furthermore, adding $0$'s  does not change the $\|\pl \|_{\psi-\min}$-norm.
\end{proof}

We need the following well-known fact:

\begin{lemma}\label{copy}
Let $f_j$ be identically distributed independent copies of a
random matrix $f$ on a probability space and $1\le q\le \infty$.
Then
 \begin{equation}\label{q1}
  \|\sum_{l=1}^n f_l\ten e_l\|_{L_1(S_1^m(\ell_{\infty}^n))}
  \kl n^{1/q} \|f\|_{S_1^m(L_q)} \pl .
  \end{equation}
\end{lemma}

\begin{proof} Let us prove this for $q=1$. Then clearly the
triangle inequality implies the assertion. For $q=\infty$ we may
assume that $f$ is positive and it can be written as
 \[ f  \lel a^* F a \quad a\in S_2^m \pl ,\pl F\in L_{\infty}(M_m)
 \]
such that $\|a\|_2\le 1$ and $\|F\|_{\infty}\le
\|f\|_{S_1^m(L_{\infty})}$. Let $F_j$ be independent copies. Then
we find
 \[ f_l \lel a^* F_l a \quad \mbox{and}\quad \sup_l \|F_l\| \kl
 \|F\|_{\infty} \pl .\]
Since the underlying measure space is a probability space, we
obtain the assertion in the case $q=\infty$. For $1\le q\le
\infty$, this follows from a complex interpolation argument as in
\cite{JuPa}. Just note that in operator space jargon, the Lemma is
an immediate consequence of the fact that $id:\ell_q^n\to
\ell_{\infty}^n$ and $id:L_q(\Omega;X)\to L_1(\Omega;X)$ are
complete contractions. \end{proof}

Although the techniques to prove Theorem \ref{completely
complemented} are similar to the ones used before, we have to use a
random embedding of $\ell_2^n$ into $L_1(\Omega,\ell_\infty^n)$
contained in the following Lemma:
\begin{lemma}\label{random}
Let $(g_{i,j})_{i,j=1}^n$ be a normalized family of independent
and identically distributed real gaussian variables. The
application $u:R_n\cap C_n\rightarrow L_1(\Omega,\ell_\infty^n)$
defined by
 \[  u(e_i)\lel \sum_{j=1}^ng_{i,j}\otimes e_j  \text{   for
 every   } i=1,\cdots ,n\pl ,\]
satisfies  $\|u\|_{cb}\preceq \sqrt{\log n}$. Similarly, the map
$u_{RAD}(e_i)=\sum_{j=1}^n\epsilon_{i,j}\otimes e_j$ for every $i=1,\cdots ,n$ satisfies
$\|u_{RAD}\|_{cb}\le C\sqrt{\log n}$.
\end{lemma}\begin{proof} Let  $x_1,\cdots ,x_k\in S_1^m$. Then we have
 \begin{align*}
 \mathbb{E} \|\sum_{i,j=1}^ng_{i,j} x_i\otimes e_j\|_{S_1^m(\ell_\infty^n)}
 \kl C \sqrt{\log n}\pl  \|\sum_{k=1}^nx_i\otimes e_j\|_{S_1^m(R_n\cap C_n)}.
\end{align*}
Indeed, let $2< q<\infty$. From Lemma \ref{copy} we deduce that
 \begin{align*}
\|\sum_{i,j=1}^ng_{i,j} x_i\otimes
e_j\|_{L_1(S_1^m(\ell_\infty^n))}& =\|\sum_{j=1}^n(\pi_j\otimes
 \id)(\sum_{i=1}^ng_ix_i\otimes e_j
 )\|_{L_1(S_1^m(\ell_\infty^n))}\\\leq
n^{\frac{1}{q}}\|\sum_{i=1}^nx_i\otimes
g_i\|_{S_1^m(L_q(\Omega))},
\end{align*}
where $\pi_j:L_q(\Omega)\rightarrow L_1(\Omega^n)$ is defined by
$\pi_j(f)=\overbrace{\id\otimes\cdots \otimes\id}^{j-1}\otimes
f\otimes \overbrace{\id\otimes\cdots \otimes \id}^{n-j}.$

We need to recall the definition of $RC_q^n=[R_n\cap C_n,R_n+C_n]_{1/q}$. Then we have trivially
 \[ \|id:R_n\cap C_n \to RC_q^n\|_{cb} \le 1 \pl , \]
and the noncommutative Khintchine inequality can be reformulated (see \cite{Pisierbook})
as
\begin{align*}
\|id:RC_q^n\rightarrow {\rm span}\{g_i:1\leq i\leq n\}\subset L_q(\Omega)\|_{cb}\kl C \sqrt{q}.
\end{align*}
Hence $\|id:R_n\cap C_n\to {\rm span}\{g_i:1\leq i\leq
n\}\|_{cb}\le C\sqrt{q}$ implies

\begin{align*}
\|\sum_{i,j=1}^ng_{i,j} x_i\otimes
e_j\|_{L_1(S_1^m(\ell_\infty^n))}\leq
n^{\frac{1}{q}}\|\sum_{i=1}^nx_i\otimes
g_i\|_{S_1^m(L_q(\Omega))}\leq  C\sqrt{q}n^{\frac{1}{q}}\|\sum_{i=1}^nx_i\otimes
e_i\|_{S_1^m(R_n\cap C_n)}.
\end{align*}Taking $q=\log n$ we
obtain the first assertion. The second assertion follows
immediately from the contraction principle (see \cite{PiVar})
 \begin{align*}
 \mathbb{E}  \|\sum_{ij}  \eps_{ij}  x_{ij} \|_X &\le  \sqrt{\frac{\pi}{2}}
 \mathbb{E} \|\sum_{ij} g_{ij}  x_{ij}\|_X \pl. \qedhere
 \end{align*}
\end{proof}



\begin{proof}[Proof of Theorem \ref{completely complemented}]
Let $0<\delta<\frac{1}{2}$ and $n\in \mathbb{N}$. As a consequence
of Chevet's inequality and Chebyshev's inequality we can obtain a
constant $C(\delta)> 0$ and a set $A_\delta\subset \Omega$ such
that for every $\omega\in A_\delta$ we have
 \begin{equation}\label{norm-operator}
 \|\sum_{i,j=1}^n\epsilon_{i,j}(\omega)e_i\otimes
 e_j\|_{\ell_2^n \otimes_\epsilon \ell_2^n}\leq C(\delta)\sqrt{n},
 \end{equation}
and $\mu(A_\delta^c)< \delta$. In particular, we have
\begin{equation}\label{norm-v}
\|\sum_{i,j=1}^n\epsilon_{j,i}(\omega)e_i\otimes
e_j\|_{\ell_1^n\otimes_\epsilon \ell_2^n} \leq C(\delta)n
\end{equation}
for $\omega \in A_{\delta}$. Then, we define
$v:L_1(\Omega,\ell_\infty^n)\rightarrow R_n\cap C_n$ by
 \begin{align*}
 v(f)=\frac{1}{n}\sum_{k=1}^n(\int_{A_\delta}\sum_{j=1}^Nf_j(\omega)\epsilon_{k,j}(\omega)d\mu)e_k.
 \end{align*}
We deduce from Equation (\ref{norm-v}) that
 \begin{equation} \label{ness}
 \|v\|=\sup_{\omega\in A_\delta}\|\sum_{j,k=1}^n\epsilon_{k,j}(\omega)e_j\otimes
 e_k\|_{\ell_1^n\otimes_\epsilon \ell_2^n}\leq C(\delta)\pl .
 \end{equation}
Furthermore, by Lemma \ref{Grothendieck}, we infer that
$\|v\|_{cb}\leq C'(\delta) $. Recall that $u:R_n\cap
C_n\rightarrow L_1(\Omega,\ell_\infty^n)$ defined by
\begin{align*}
u(e_i)=\sum_{j=1}^n\epsilon_{i,j}\otimes e_j  \text{   for every }
i=1,\cdots , n;
\end{align*}
satisfies  $\|u\|_{cb}\le C \sqrt{\log n}$. Let us now show that
$v\circ u$ is invertible on a large subspace of dimension $\eta
n$. Indeed, we observe that that $(v\circ
u)(e_i)=\frac{1}{n}\sum_{j,k=1}^n(\int_{A_\delta}
\epsilon_{i,j}\epsilon_{k,j})e_k$. Then
\begin{align*}
 \|v\circ u\|_{M_n}\leq
 \frac{1}{n}(\int_{A_\delta}\|\sum_{i,k=1}^n(\sum_{j=1}^n\epsilon_{i,j}\epsilon_{k,j})e_i\otimes
 e_k\|_{\ell_2^n\otimes_\epsilon
\ell_2^n}d\mu)=\frac{1}{n}(\int_{A_\delta}\|R^tR\|_{\ell_2^n\otimes_\epsilon
 \ell_2^n}d\mu)\leq C(\delta)^2;
\end{align*}
follows from  equation (\ref{norm-operator}) for
$R=\sum_{i,j=1}^n\epsilon_{i,j} e_i\otimes e_j:\ell_2^n\rightarrow
\ell_2^n$. On the other hand,
\begin{align*}
 \|v\circ u\|_{S_1^n}\geq|tr(v\circ u)|=
 \frac{1}{n}|\int_{A_\delta}\sum_{i,j=1}^n\epsilon_{i,j}^2d\mu|=n
 \mu(A_\delta)> \frac{n}{2}.
\end{align*}
Take now $0< \eta< \frac{1}{4C(\delta)^2}$ and recall that
$s_i(v\circ u)$ denotes the $i^{th}$ singular value of $v\circ u$.
Then  we have
\begin{align*}\frac{n}{2}\leq \|v\circ u\|_{S_1^n}=\sum_{i=1}^n|s_i(v\circ u)|\leq |s_1(v\circ u)|([\eta
 n]-1)+n|s_{[\delta n]}(v\circ u)|\\
 = C(\delta)^2([\eta n]-1)+ n|s_{[\eta n]}(v\circ u)|\leq C(\delta)^2\eta n +n|s_{[\eta n]}(v\circ u)|.
 \end{align*}
We conclude that $|s_{[\eta n]}(v\circ u)|\geq \frac{1}{4}$. Let
$\tilde{n}=[\eta n]$ and  $v\circ u=o|v\circ u|$ the polar
decomposition. Let $w:\ell_2^{\tilde{n}}\to \ell_2^n$ be the
partial isometry sending the unit vectors to the eigenvalues. Then
$s_{\tilde{n}}(v\circ u)$ implies that $\|(w^*|v\circ u|w)^{-1} \|
\le s_{\tilde{n}}(v\circ u)^{-1}\le 4$. Therefore we find
 \[ w^*o^*v\circ u (w^*|v\circ u|w)^{-1}\lel
 id_{\ell_2^{\tilde{n}}} \pl .\]
On the other hand
 \[ \|u (w^*|v\circ u|w)^{-1}:R_{\tilde{n}}\cap C_{\tilde{n}}\to
 \ell_1^{2^{n^2}}(\ell_{\infty}^n)\|_{cb} \kl 4 C \sqrt{\log n} \]
and  $\|w^*o^*v\|_{cb} \le C'(\delta)$. Therefore our result is
proved for $\tilde{n}$ instead of $n$. Note that $n\le
\frac{2}{\eta}\tilde{n}$ and hence $D=2\eta^{-1}\le 8C(\delta)^2$
and $k\le 2^{D^2\tilde{n}^2}$. \end{proof}

\begin{remark}\emph{(Application to Bell inequalities)} We also obtain similar violation of Bell
inequalities as in Section \ref{The construction} with a similar
behavior for states. Indeed, the coefficients of the Bell
inequalities with $k\le 2^{n^2}$ inputs  and $n$ outputs are given
by
\begin{align*}
 M_{a,b}^{\omega,\omega'} \lel 2^{-2n^2} \sum_{k=1}^n
 \epsilon_{k,a}(\omega) \epsilon_{k,b}(\omega').
 \end{align*}
Then the estimate $\|u_{RAD}\|_{cb}\le C \sqrt{\log n}$ implies
 \[ \| \sum_{a,b,\omega,\omega'}
  M_{a,b}^{\omega,\omega'} (e_{\omega}\ten e_{a})\ten (e_{\omega'}\ten
  e_b)\|_{\ell_1^{2^{n^2}}(\ell_\infty^n)\otimes_{\eps}
  \ell_1^{2^{n^2}}(\ell_\infty^n)} \kl C^2 \log n \pl .\]
Moreover, with the correct  normalization we have
 \[ \sum_{a,b,\omega,\omega'}  M_{a,b}^{\omega,\omega'} (e_{\omega}\ten e_{a})\ten (e_{\omega'}\ten
  e_b) \lel \sum_{k=1}^n u_{RAD}(e_k)\ten u_{RAD}(e_{k}). \]
Therefore we deduce from Lemma \ref{maxentanglement} and Lemma
\ref{random} that
 \begin{equation}\label{mtan}
  \|M\|_{\psi-\min}\kl 4 \|u_{RAD}\|_{cb}^2 \preceq \log n \pl .
  \end{equation}
On the other hand, we find exactly the same
behavior in terms of states as in section \ref{The construction}.
Indeed, using \eqref{ness} and the argument from Step \ref{step2}, we
deduce that
  \begin{align*}
 E_\omega^a= \frac{1_{A_\delta}(\omega)}{2C(\delta)^2n}\left(
 \begin{array}{ccccccc}
   1 &  \epsilon_{1,a}(\omega) & \cdots & \epsilon_{n,a}(\omega) \\
  \epsilon_{1,a}(\omega) &  1 & \cdots & \epsilon_{1,a}(\omega)\epsilon_{n,a}(\omega)\\
  \vdots & \vdots & \vdots  & \vdots \\
  \epsilon_{n,a}(\omega) &  \epsilon_{n,a}(\omega)\epsilon_{1,a}(\omega)  & \cdots & 1\\
\end{array}\right)\in M_{n+1}, a=1,\cdots ,n;
\end{align*}
satisfies
 \[ \sup_{\omega} \sum_a E_\omega^a \le 1 \pl .\]
Thus we may define $E_{\omega}^{n+1}=1-\sum_a E_\omega^a$ for all
$\omega$ and obtain POVMs indexed by $\omega\in
\{-1,1\}^{2^{n^2}}$. For $\phi=\al_1|11\rangle+\sum_{i=2}^n
\al_i|ii\rangle$ we see that the same argument as in Step \ref{step3}
yields
\begin{align*}
 &\sum_{\omega,\omega'} \sum_{a,b=1}^n M_{a,b}^{\omega,\omega'}\langle
 \phi |E_\omega^a\otimes E_{\omega'}^{b}|\phi\rangle
 \gl \frac{2}{(2C(\delta)^2)^2} \al_1\sum_{i=2}^{n+1} \al_{i}
 \pl .
 \end{align*}
\emph{Hence the matrix $(\frac{1}{C \log
n}\tilde{M}_{a,b}^{\omega,\omega'})_{a,b=1,..,n+1,\omega,\omega'}$
is an example of a Bell inequality which can not
produce large violation on the maximal entangled state, but produces large
violation for all states with $\iviol(\phi)\gg  \log n$}.
\end{remark}

\begin{conc}
We have provided an example of a Bell inequality which gives violations of order $\frac{\sqrt{n}}{\log n}$ but only bounded violations can be obtained with any maximally entangled state. This example suggests that the maximally entangled state is a poor candidate to get large violations. A similar statement holds in the context of tripartite correlations, see \cite{PWJPV} and the recent generalization to diagonal states in \cite{BBLV}. However, in \cite{BGW} the authors showed the existence of a Bell inequality (constructed with $2^n$ inputs and $n$ outputs) with positive coefficients for which the maximally entangled state in dimension $n$ gives violations of order $\frac{\sqrt{n}}{\log n}$. Therefore, we can not expect to have condition b) in Theorem \ref{Theorem 2} for every Bell inequality $M$.
\end{conc}

\begin{problem}\label{problem2} {\rm
The key point to prove condition b) in Theorem \ref{Theorem 2} is
the fact that the operator $u_{RAD}:R_n\cap C_n\rightarrow
L_1(\Omega,\ell_\infty^n)$ admits a good estimate of the cb-norm.
It would be nice to know whether the operator $V:R_n\cap C_n
\rightarrow\ell_1^n(\ell_\infty^n)$, defined by
\begin{align*}
V(e_k)=\frac{1}{n}\sum_{i,j=1}^n\epsilon_{i,j}^ke_i\otimes e_j
\text{   for every   } k=1,\cdots ,n;
\end{align*} also satisfies a similar estimate  $\|V\|_{cb}\preceq \sqrt{\log n}$.
We refer to \cite[Remark 9]{JPPVW} for the interest in negative
answer, which would imply that there are violations of Bell
inequality involving POVM's only for Alice or Bob, but not both.
On the other hand, an affirmative answer would imply that Theorem
\ref{complemented} gives a $\sqrt{\log n}$- completely
complemented copy of $R^{\delta n}\cap C^{\delta n}$ into
$\ell_1^n(\ell_\infty^n)$. This would imply, in particular, that
the element $M$ in Theorem \ref{Theorem-construction} also
verifies property b) in Theorem \ref{Theorem 2}.
}\end{problem}


\subsection{The role of the maximally entangled state in violation of Bell inequalities}

Theorem \ref{Theorem 2} says that there exist violations of Bell
inequalities which can not be obtained from the maximally
entangled state. This is completely different from the case of
quantum correlations (see \cite{Tsirelson}). We will now show that
for positive Bell inequalities  and violation in dimension $n$,
the maximal entangled state plays a prominent role. Using standard
tools from interpolation theory we show that every pure state can
be written, up to a $\sqrt{\log n}$ factor, as a
\emph{superposition} of maximal entangled states in lower
dimensions.

\begin{lemma}\label{diagonal} Let $|\psi\rangle=\sum_{i=1}^n a_ie_i\in \ell_2^n$ be
such that $a_1\geq a_2\geq \cdots \geq a_n\geq 0$ and
$\sum_{i=1}^na_i^2\leq 1$. Then, there exist positive coefficients
$\beta_s\geq 0$ such that
 \[ |\psi\rangle=\sum_{s=1}^l \beta_s|\varphi_s\rangle\]
for some $l\in \mathbb{N}$ with $l\kl C n+\log n$, and some
vectors $(|\varphi_s\rangle)_{s=1}^l$ satisfying
\begin{enumerate}
 \item[a)] $\sum_{s=1}^l\beta_s\leq 2\sqrt{\log n}$;
 \item[b)] For every $s=1,\cdots ,l$,
 $|\varphi_s\rangle=\frac{1}{\sqrt{|A_s|}}\sum_{i\in A_s}|i\rangle$,
 where $A_s$ is a set contained in $\{1,\cdots , n\}.$
\end{enumerate}
\end{lemma}

\begin{proof}[Proof of Lemma \ref{diagonal}]
We use the dyadic intervals $I_k=[2^{k-1},2^k-1]$ for $k=1,\cdots,
\log (n+1)$ and note $|I_k|=2^{k-1}$ for every $k$. Then, we can
write
\begin{align*}
|\psi\rangle=\sum_{k=1}^{\log {(n+1)}}(\sum_{i\in
I_k}a_i|i\rangle)=\sum_{k=1}^{\log
(n+1)}a_{2^{k-1}}2^{\frac{k-1}{2}}(\sum_{i\in
I_k}\frac{a_i}{a_{2^{k-1}}}|i\rangle)2^{-\frac{k-1}{2}}.
\end{align*}
First, note that
 \begin{align*}
 \sum_{k=1}^{\log {(n+1)}}a_{2^{k-1}}2^{\frac{k-1}{2}}&\leq   \sqrt{\log (n+1)}(\sum_{k=1}^{\log
 n}a_{2^{k-1}}^22^{k-1})^\frac{1}{2}\\
  & \leq \sqrt{2}\sqrt{\log (n+1)} (\sum_{i=1}^na_i^2)^\frac{1}{2}\leq
 \sqrt{2}\sqrt{\log (n+1)}.
 \end{align*}
On the other hand, we see that for fixed $k$, the set
$\{\frac{a_i}{a_{2^{k-1}}}: i\in I_k\}$ is a set of $2^{k-1}$
positive numbers of value lower or equal one. By  Caratheodory's
Theorem there are  positive convex combinations
$\sum_{j=1}^{2^{k-1}+1}\alpha_j^k=1$ such that
 \[ \frac{a_i}{a_{2^{k-1}}}\lel
 \sum_{j=1}^{2^{k-1}+1}\alpha_j^kA_j^k(i)\pl .\]
Here $A_j^k(i)\in \{0,1\}$ for every $k,j$ and $i$. Thus, we
obtain
\begin{align*}
 |\psi\rangle&=\sum_{k=1}^{\log
 (n+1)}a_{2^{k-1}}2^{\frac{k-1}{2}}(\sum_{i\in I_k}\sum_{j=1}^{2^{k-1}+1}\alpha_j^kA_j^k(i)|i\rangle)2^{-\frac{k-1}{2}}
 \\
 &=\sum_{k=1}^{\log (n+1)}\sum_{j=1}^{2^{k-1}+1}\alpha_j^ka_{2^{k-1}}2^{\frac{k-1}{2}}(\sum_{i\in
 I_k}A_j^k(i)|i\rangle)2^{-\frac{k-1}{2}}\pl .
\end{align*}
Given $k$ and $j$, we denote by $|A_j^k|={\rm card}\{i\in I_k:
A_j^k(i)=1\}$ the cardinality. Obviously, we have $|A_j^k|\leq
2^{k-1}=|I_k|$. This yields
\begin{align*}
 |\psi\rangle&=\sum_{k=1}^{\log
 (n+1)}\sum_{j=1}^{2^{k-1}+1}\alpha_j^ka_{2^{k-1}}2^{\frac{k-1}{2}}\frac{\sqrt{|A_j^k|}}{\sqrt{2^{k-1}}}\left(\frac{1}{\sqrt{|A_j^k|}}\sum_{i\in
 I_k}A_j^k(i)|i\rangle\right) \pl .
\end{align*}
By construction we have
\begin{align*}
 &\sum_{k=1}^{\log
(n+1)}\sum_{j=1}^{2^{k-1}+1}\alpha_j^ka_{2^{k-1}}2^{\frac{k-1}{2}}\frac{\sqrt{|A_j^k|}}{\sqrt{2^{k-1}}}\leq
 \sum_{k=1}^{\log
(n+1)}\sum_{j=1}^{2^{k-1}+1}\alpha_j^ka_{2^{k-1}}2^{\frac{k-1}{2}}\\
 &=\sum_{k=1}^{\log (n+1)}a_{2^{k-1}}2^{\frac{k-1}{2}}\leq \sqrt{2}\sqrt{\log
 (n+1)}\leq 2\sqrt{\log n}.
\end{align*}
Thus it remains to rename the double indices by $s(k,j)$ and to
use
$\beta_{s(j,k)}=\alpha_j^ka_{2^{k-1}}2^{\frac{k-1}{2}}\frac{\sqrt{|A_j^k|}}{\sqrt{2^{k-1}}}$.
Then  $|\varphi_{s(j,k)}\rangle=\frac{1}{\sqrt{|A_j^k|}}\sum_{i\in
I_k}A_j^k(i)|i\rangle$ is a characteristic function with support in
$I_k$.\end{proof}

\begin{remark}
The general case in which the coefficients of $|\psi\rangle$ are
not necessary positive follows by decomposition in real and imaginary
and then positive and negative part.
\end{remark}

As an application of Lemma \ref{diagonal} we show that for games (even Bell inequalities satisfying a very weaker positivity condition) the largest violation is attained, up to a logarithmic factor, for a maximally entangled state.

\begin{theorem}\label{positive coefficients}

Let $M=(M_{a,b}^{x,y})_{a,b,x,y}$ be a Bell inequality with
positive coefficients. Suppose there exists an state $\rho$ acting
on a $n$-dimensional Hilbert space $H$ and verifying $\langle M,
\rho\rangle=C$. Then, there exists $k\leq n$ such that
 \begin{align}\label{posceof}
 \langle\psi_k| M |\psi_k\rangle\geq\frac{C}{4\log n},
 \end{align}
where $|\psi_k\rangle=\frac{1}{\sqrt{k}}\sum_{i=1}^k|ii\rangle$ is the maximally entangled state in dimension $k$.
\end{theorem}

\begin{proof} By hypothesis there exist some POVM's $\{E_x^a\}_{x,a},
\{F_y^b\}_{y,b}$ acting on a Hilbert space $H$ of dimension $n$ and
a state $\rho:H\otimes_2 H\rightarrow H\otimes_2 H$ such that
$\sum_{x,y,a,b}M_{x,y}^{a,b}tr(E_x^a\otimes F_y^b\rho)= C$. This implies that the operator  $\tilde{M}=\sum_{x,y,a,b}M_{x,y}^{a,b} E_x^a\otimes F_y^b$ is positive (and this is the weaker positivity assumption we need).  By a convexity argument, we may
assume that $\rho$ is a  pure state $\rho=|\psi\rangle\langle\psi|$ on $H\otimes_2 H$. Using the
Hilbert-Schmidt decomposition, we may even assume
that $|\psi\rangle$ is a diagonal element with positive coefficients given by the singular values $|\psi\rangle=\sum_i \al_i |ii\rangle$. According to  Lemma \ref{diagonal} we find a decomposition
 \[  |\psi\rangle\lel  \sum_{p}\beta_p |\varphi_p\rangle\]
where $(\beta_p)_p$ are positive coefficients satisfying
$\sum_p \beta_p\leq 2\sqrt{\log n}$ and
 \[  |\varphi_p\rangle\lel \frac{1}{\sqrt{|A_p|}}\sum_{i\in
 A_p}|pp\rangle \]
holds for all $p$ and some sets $A_p\subset \{1,\cdots,
n\}$. Thus we find  $p_0$ and $q_0$ such that
\begin{equation}\label{assymetric}
\langle
\varphi_{q_0}|\tilde{M}|\varphi_{p_0}\rangle\geq\frac{C}{4\log n}.
\end{equation}
Now, we use the positivity and the  Cauchy-Schwartz inequality:
\begin{equation}\label{Cauchy-Schwartz}
\frac{C^2}{4^2(\log n)^2}\leq\langle
\varphi_{q_0}|\tilde{M}|\varphi_{p_0}\rangle^2\leq \langle
\varphi_{q_0}|\tilde{M}|\varphi_{q_0}\rangle\langle
\varphi_{p_0}|\tilde{M}|\varphi_{p_0}\rangle.
\end{equation}
Thus \eqref{posceof} must be satisfied by $\langle
\varphi_{q_0}|\tilde{M}|\varphi_{q_0}\rangle$ or $\langle
\varphi_{p_0}|\tilde{M}|\varphi_{p_0}\rangle$.
\end{proof}




The previous theorem is particularly interesting once we know that
there exist violations of Bell inequalities with positive
coefficients (actually games) of polynomial order in the
dimension (\cite{Raz}, \cite{Brassard-review}, \cite{BGW}).

\begin{remark}\label{asymmaxent} For  arbitrary $M$ we still have \eqref{assymetric}. In fact, previously we obtained large violation with $|q_0|=1$ and $|p_0|=n-1$. Furthermore, as a consequence of the Polarization identity, we find some  $k\in \{0,1,2,3\}$ such that
\begin{align*}
\langle \xi|M|\xi\rangle\geq\frac{C}{4^2\log n}
\end{align*} for the non normalized state $|\xi\rangle=|\varphi_A\rangle+i^k|\varphi_B\rangle$.
\end{remark}

\begin{remark}  It turns out that the operator
$\sum_{x,y,a,b}M_{x,y}^{a,b} E_x^a\otimes F_x^a$ is positive for
all POVMs if and only if $M_{x,y}^{a,b}$ represents a positive
element in $NSG\ten NSG$ (see Subsection \ref{NSG}).
\end{remark}


\section{Geometric interpretation of violation of Bell inequalities and the $NSG$ space}\label{geometric}

\subsection{Geometric interpretation of violation of Bell inequalities}\label{geometric_int}

We have seen in Section \ref{State of the Main Results} that for certain applications in QIT (see  \cite{DKLR}, \cite{JPPVW}, and \cite{JPPVW2}) the value  $\sup_{P\in\mathcal{Q}}\nu(P)$ is very interesting. In this section we want to provide a geometric interpretation of this value in terms of convex sets of probabilities. We begin recalling some basic notions of convex geometry.

\begin{definition}\label{Minkowski}
\

\begin{enumerate}
\item[i)] We say that a set $S\subset \R^k$ is \emph{absolutely convex} if for every $x_1, x_2\in S$ and $\lambda_1,\lambda_2\in \R$ with $|\lambda_1|+|\lambda_2|\leq 1$ we have $\lambda_1x_1+\lambda_2x_2\in S$.
\item[ii)]  Let $S\subset \mathbb{R}^k$ be a convex set. Then
     \[ \rho_S(x):=\inf\{\lambda\geq 0: x\in \lambda S\} \]
     is the Minkowski functional of $S$.
\end{enumerate}
\end{definition}

The bipolar theorem tells us that
 \begin{equation}\label{bipolar}
 \rho_S(x) \lel \sup\{ B(x): \sup_{s\in S} B(s)\le 1\} \pl .
 \end{equation}





We say that the set $S$ is \emph{absorbing} if we have $\rho_S(x)<\infty$ for every $x\in \R^k$. Clearly $S$ is always absorbing if it is restricted to the linear space $[S]$ generated by $S$. In particular, given two bounded absolutely convex sets $S\subseteq Q\subset [S]\subset \R^k$, we may consider the homothetic distance:

\begin{align}\label{distance sets}
d(S,Q):= \sup_{x\in Q}\rho_S(x)=\inf\{\lambda\geq 0: Q\subseteq \lambda S\}.
\end{align}

\begin{center}
\begin{figure}[h]
  \includegraphics[width=3cm]{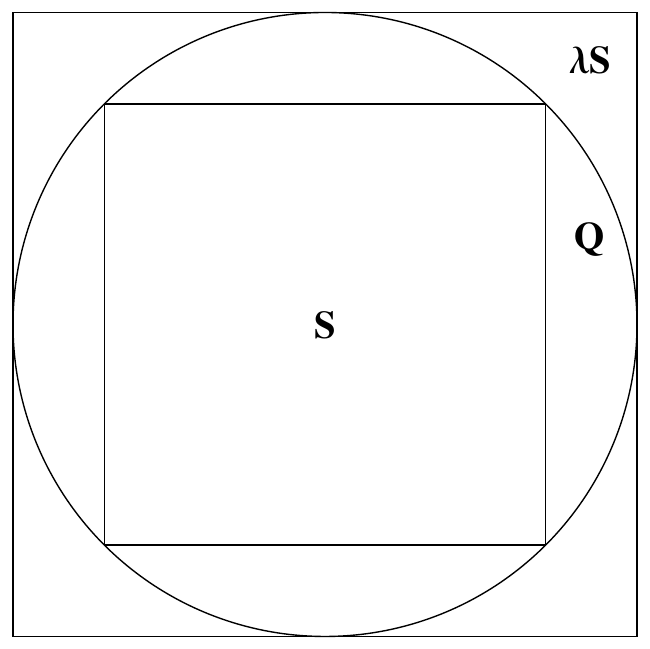}
  \caption{Geometric meaning of $d(S,Q)$.}\label{figure2}
\end{figure}
\end{center}
\vspace{-0.5cm}
Equation (\ref{Minkowski}) tells us that we can also define the previous distance by duality. Indeed, if we consider

\begin{align}\label{distance functionals}
\zeta(B)=\frac{\sup_{x\in Q} B(x)}{\sup_{y\in S}B(y)}
\end{align}
for every $B\in \R^k$, we have $\sup_{B\in\R^k}\zeta(B)=d(S,Q)$.

\begin{remark}
The fact that  $ Q\subset [S]$ guarantees that $\sup_{y\in S}B(y)=0$ implies $\sup_{x\in Q}B(x)=0$. So we just define $\frac{0}{0}=0$ in Equation (\ref{distance functionals}).
\end{remark}

It is easy to see that the distance $d$ in Equation (\ref{distance sets}) is the one considered in the case of correlation Bell inequalities (see \cite{Tsirelson}). In fact, it can be seen that the set of classical (resp. Quantum) correlation matrices  $M_C$ (resp. $M_Q$) is exactly the unit ball of $\ell_\infty^N\otimes_\pi \ell_\infty^N$ (resp. $\ell_\infty^N\otimes_{\gamma_2} \ell_\infty^N$) (see for instance \cite{Def} for the definition of such norms). Thus, in the case of correlation Bell inequalities the physical definition of violation of Bell inequalities coincides with the natural geometric definition. In particular, Equation (\ref{distance functionals}) coincides with the largest violation of the Bell inequality $B$ and the maximum of these values is the distance described in Equation (\ref{distance sets}): $d(M_C,M_Q)$.


General Bell inequalities (where we work with the whole probability distribution) presents some additional problems  because the corresponding sets are not centrally symmetric. Indeed, our sets $\mathcal{L}$ and $\mathcal{Q}$ are contained in a $d=(NK-N+1)^2-1$ dimensional affine subspace $A$ of $\R^{ N^2K^2}$ (see \cite{Tsirelson}). Actually, it is not difficult to see that $\mathcal{L}\varsubsetneq \mathcal{Q}\varsubsetneq\mathcal{C}\subset A={\rm Aff}(\mathcal{L})$ with equality in the last inclusion if we consider only those elements in ${\rm Aff}(\mathcal{L})$ that are probability distributions. Here, $${\rm Aff}(\mathcal{L})=\left\{\sum_{i=1}^N \alpha_iP_i:N\in \N, P_i\in \mathcal{L},\alpha_i\in\R, \sum_{i=1}^N\alpha_i=1\right\}$$ denotes the affine hull of the space $\mathcal{L}$. Thus, in order to define a ``good distance'' in this situation we must be more careful.

\

Consider a subset $S$ contained in an affine subspace ${\rm Aff}(S)\subset \R^k$ of dimension $d$. It is standard in convex geometry to consider the \emph{absolutely convex hull} of $S$:

\begin{align*}
\tilde{S}:=conv(S\cup -S).
\end{align*}$\tilde{S}$ is an absolutely convex set contained in a linear space of dimension $d+1$. It is easy to see that
\begin{align*}
\rho_{\tilde{S}}(x)= \inf\left\{\sum_{i=1}^N |\alpha_i|: x=\sum_{i=1}^N \alpha_is_i:N\in \N, s_i\in S,\alpha_i\in\R\right\}
\end{align*} for every $x\in [\tilde{S}]$.



Therefore, if we have two sets $S\subseteq Q\subset {\rm Aff}(S)$, we can naturally define a distance between them by using their absolutely convex hull $\tilde{S}\subseteq \tilde{Q}$ with the same geometric interpretation as in Equation (\ref{distance sets}):
\begin{align*}
d_1(S,Q)= d(\tilde{S}, \tilde{Q}).
\end{align*}
\
\
\

\vspace{-2.5cm}
\begin{center}\setlength{\unitlength}{0.5mm}
\begin{picture}(0,90)(0,0)
\put(-50,0){\vector(1,0){100}}
\put(0,-50){\vector(0,1){100}}
\drawline(20,40)(40,20)\drawline(-40,-20)(-20,-40)\drawline(-40,-20)(20,40)\drawline(40,20)(-20,-40)
\drawline(5,25)(25,5)\drawline(-25,-5)(-5,-25)\drawline(-25,-5)(5,25)\drawline(25,5)(-5,-25)
\put(16,17){\tiny$\displaystyle S$}
\put(-23.5,-21.5){\tiny $\displaystyle -S$}
\put(8.5,23){\tiny $\displaystyle Q$}
\put(-31.5,-14.5){\tiny $\displaystyle -Q$}
\put(4,6){\tiny $\displaystyle \overline{S}$}
\put(30,18){\tiny $\displaystyle \lambda\overline{S}$}
\put(-5,10){\tiny $\displaystyle \overline{Q}$}
\put(45,-12){\tiny $\displaystyle {\rm Aff}(S)$}
\dashline[3]{2}(-20,50)(50,-20)
\thicklines
\drawline(10,20)(20,10)
\drawline(-20,-10)(-10,-20)
\drawline(-20,-10)(10,20)
\drawline(20,10)(-10,-20)
\put(0,-55){\makebox(0,0){Figure 3. Geometric meaning of $d_1(S,Q)$.}}
\end{picture}
\end{center}
\vspace{2cm}

\

As before, the dual point of view defines a distance in terms of functionals. That is, for any linear functional $B$ on $\R^k$, we have

\begin{align*}
\zeta_1(B)=\frac{\sup_{q\in\tilde{Q}}B(q)}{\sup_{s\in\tilde{S}}B(s)}=\frac{\sup_{q\in Q}|B(q)|}{\sup_{s\in S}|B(s)|}.
\end{align*}

We immediately deduce that this is the distance that we have considered in Section \ref{State of the Main Results} in the particular case of  $S=\mathcal L$ and $Q=\mathcal Q$. Specifically,

\begin{align*}
\sup_{P\in\mathcal{Q}}\nu(P)=d_1(\mathcal L, \mathcal Q),
\end{align*}
and for every linear functional $M$ on $\R^{K^2N^2}$,
\begin{align*}
LV(M)=\zeta_1(M).
\end{align*}
Therefore, Theorem \ref{Theorem 1} can be stated as follows:

\begin{theorem}
With the same notation as in Section \ref{State of the Main Results}, if $N=n$, $K=n$ and $d=n$,
$$d_1(\mathcal L,\mathcal Q)\succeq \frac{\sqrt{n}}{\log n}.$$
\end{theorem}

However, from a purely geometric point of view $d_1(\mathcal L,\mathcal Q)$ presents two problems:

\begin{enumerate}

\item[a)] The sets we are using to define the distance, $\tilde{Q}$, $\tilde{S}$, are ``much'' bigger than the sets $S$ and $Q$. In particular, the previous definition involves to consider an extra dimension.

\item[b)] In order to measure distances between two affine subspaces we would like to have a measure invariant under translations. The previous one does not verify this property.
\end{enumerate}

There exists a natural way to obtain a completely convex set from an affine convex subset $S$ which solves the previous problems. Indeed, consider the set
\begin{align*}
\hat{S}=S-S.
\end{align*}
\vspace{-2.5cm}
\begin{center}\setlength{\unitlength}{0.5mm}
\begin{picture}(0,90)
\put(-50,0){\vector(1,0){100}}
\put(0,-50){\vector(0,1){100}}
\dashline[3]{2}(-20,50)(50,-20)
\dashline[3]{2}(-35,35)(35,-35)
\thicklines
\drawline(5,25)(25,5)
\drawline(-20,20)(20,-20)
\put(16,17){\tiny $\displaystyle S$}
\put(2,2){\tiny $\displaystyle \hat{S}$}
\put(45,-12){\tiny $\displaystyle \text{Aff}(S)$}
\put(30,-28){\tiny $\displaystyle (\text{Aff}(S))_{lin}$}
\thicklines
\put(0,-55){\makebox(0,0){Figure 4. Geometric meaning of $\hat{S}$.}}
\end{picture}
\end{center}
\vspace{3.5cm}

The new absolutely convex set $\hat{S}$ is contained in a $d$-dimensional linear space and it is invariant under translations of $S$. As before, given two convex sets $S\subseteq Q\subset {\rm Aff}(S)$, we can naturally define a distance between them by using the sets $\hat{S}\subseteq\hat{Q}$ with the same geometric interpretation as in Equation (\ref{distance sets}):

\begin{align*}
d_2 (S, Q)=d(\hat{S}, \hat{Q}).
\end{align*}

Furthermore, the dual formulation of $d_2$ defines a very nice distance in terms of functionals. Given an affine subset $S$ of $\R^k$, for any linear functional $\Psi$ on $\R^k$ we define the \emph{$\Psi$-width} of $S$ as
\begin{align*}
\omega_\Psi(S)=\sup_{x\in \hat{S}}\Psi(x)=\sup_{s\in S}\Psi(s)-\inf_{s'\in S}\Psi(s').
\end{align*}

\begin{center}
\setcounter{figure}{4}
\begin{figure*}[h]
  \includegraphics[width=7cm]{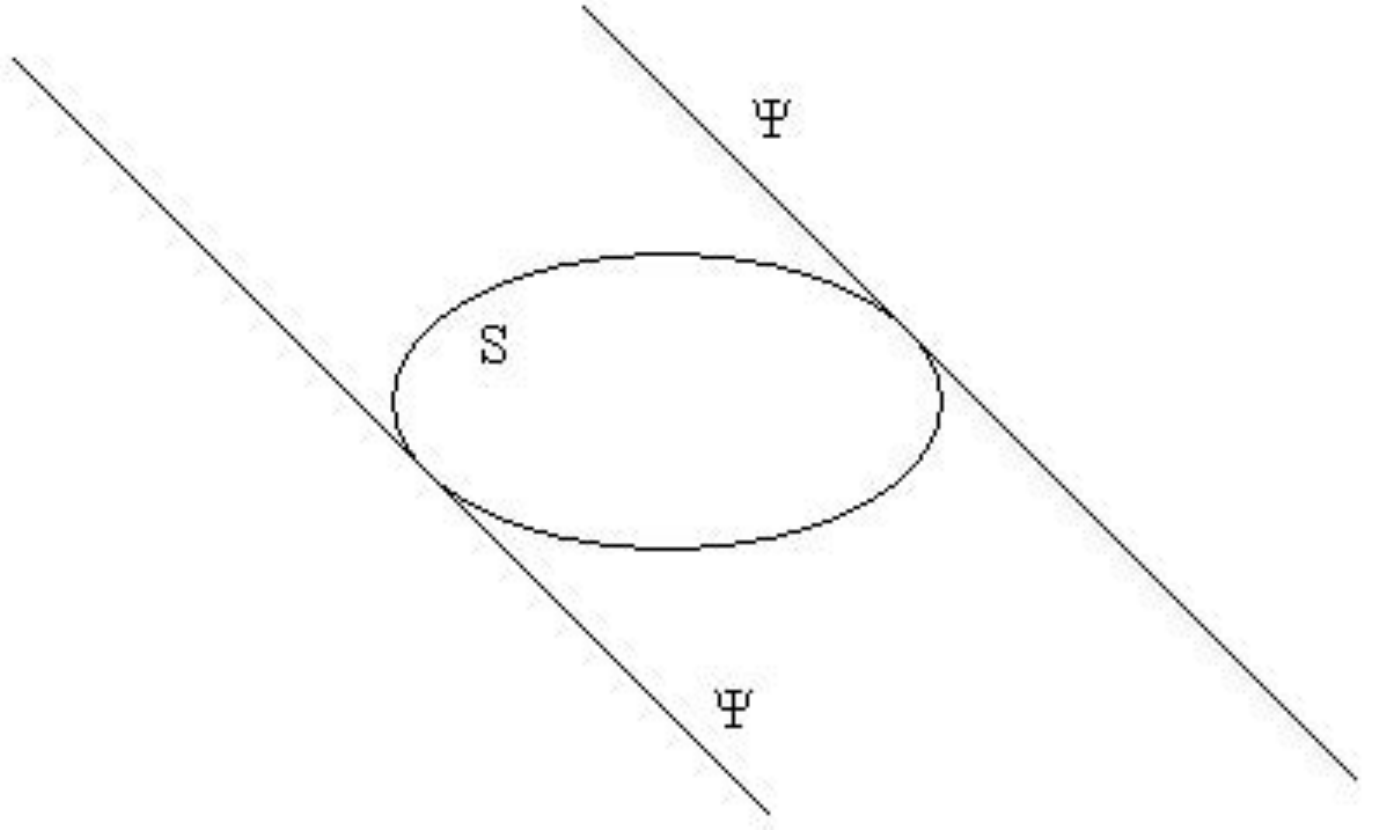}\\
  \caption{Geometric meaning of $\omega_\Psi(S)$.}
\end{figure*}
\end{center}

Then, we define the \emph{band-width distance} between $S$ and $Q$ as

\begin{align*}
d_2(S,Q)=\sup_{\Psi\in (R^k)^*}\frac{\omega_\Psi(Q)}{\omega_\Psi(S)}.
\end{align*}

\begin{center}
\setcounter{figure}{5}
\begin{figure*}[h]
  \includegraphics[width=7cm]{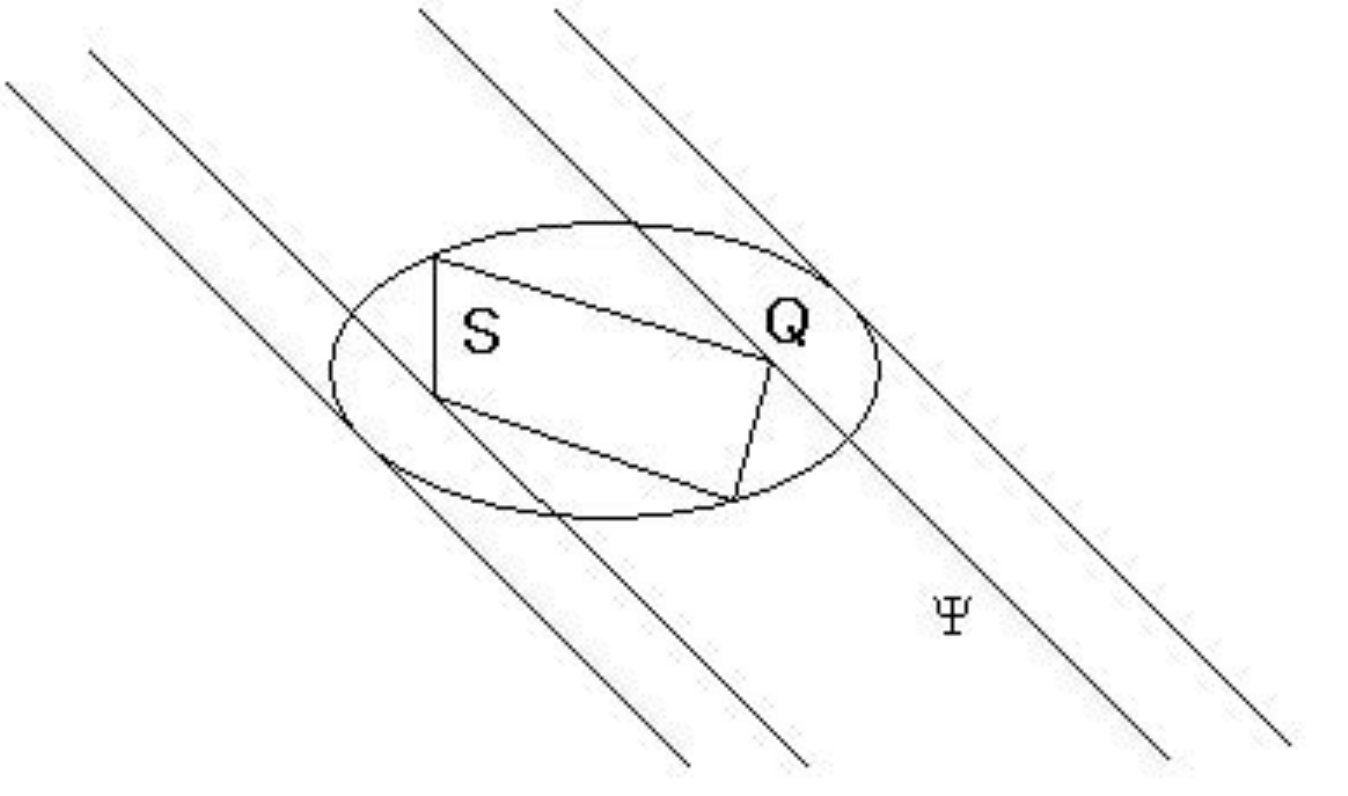}\\
  \caption{Geometrical meaning of $\frac{\omega_\Psi(Q)}{\omega_\Psi(S)}$.}
\end{figure*}
\end{center}

\begin{remark}
Again, $\omega_\Psi(S)=0$ implies $\omega_\Psi(Q)=0$, so we can define $\frac{0}{0}=0$.
\end{remark}

\begin{remark}\label{third distance}
In fact, there exists another standard way to obtain an absolutely convex set $S_h$ from $S$ in dimension $d$ (which is not, in general, invariant under translations). Consider the set
\begin{align*}
S_h=\tilde{S}\cap ({\rm Aff}(S))_{lin},
\end{align*}where $({\rm Aff}(S))_{lin}$ is the linear space associated to ${\rm Aff}(S)$ (that is, $({\rm Aff}(S))_{lin}={\rm Aff}(S)-x_0$ for any $x_0\in {\rm Aff}(S)$). However, in our particular situation both definitions are equivalent since
\begin{equation}\label{intersection}
\hat{S}= 2S_h.
\end{equation}
\end{remark}

\

Of course, we can not compare the sets $\tilde{S}$ and $\hat{S}$ because they have different dimensions. However, it is easy to see that they are comparable when we consider functionals which vanish at some point $x_0\in S$. We have:

\begin{lemma}\label{equivalence between d+1 and d}
Let $S$ be a set contained in an affine subspace ${\rm Aff}(S)$ of $\R^k$ and let $\Psi$ a linear functional on $\R^k$ such that $\Psi(x_0)=0$ for some $x_0\in S$. Then,

\begin{align}
\sup_{s\in \tilde{S}}\Psi(s)\leq \omega_\Psi(S)\leq 2\sup_{s\in \tilde{S}}\Psi(s).
\end{align}
In particular, given two sets $S\subseteq Q\subset Aff(S)\subset \R^k$ and given a linear functional $\Psi$  on $\R^k$ such that $\Psi(x_0)=0$ for some $x_0\in S$, we have

\begin{enumerate}
\item[1)] $\frac{\sup_{q\in \tilde{Q}}\Psi(q)}{\sup_{s\in \tilde{S}}\Psi(s)}\leq 2 \frac{\omega_\Psi(Q)}{\omega_\Psi(S)}$.

\item[2)] $\frac{\omega_\Psi(Q)}{\omega_\Psi(S)}\leq 2 \frac{\sup_{q\in \tilde{Q}}\Psi(q)}{\sup_{s\in \tilde{S}}\Psi(s)}$.

\end{enumerate}

\end{lemma}

\begin{proof}
To prove the first inequality just note that
\begin{align*}
\sup_{s\in \tilde{S}}\Psi(s)=\sup_{s\in S}|\Psi(s)|=\sup_{s\in S}|\Psi(s-x_0)|\leq \sup_{t\in \hat{S}}|\Psi(t)|=\sup_{t\in \hat{S}}\Psi(t).
\end{align*}
On the other hand, the second inequality note that
\begin{align*}
\omega_\Psi(S)=\sup_{s\in S-S}|\Psi(s)|\leq \sup_{s\in S}|\Psi(s)|+ \sup_{s'\in S}|\Psi(s')|= 2\sup_{s\in S}|\Psi(s)|=2\sup_{s\in \tilde{S}}\Psi(s).
\end{align*}
The second part of the lemma follows straightforward.
\end{proof}

\begin{remark}
Note that, in particular, $\tilde{S}$ and $\hat{S}$ are comparable if $0\in S$.
\end{remark}

Now, it is trivial to check that the element $M$ given in Theorem \ref{Theorem 1}  (see construction in Section \ref{The construction}) verifies that $M(P_0)=0$ if we define $ P_0(a,b|x,y)=P(a|x)P(b|y)\in \mathcal L$ by $P(a|x)=0$ if $a=1,\cdots ,k$ and $P(k+1|x)=1$ for every $x=1,\cdots ,N$. Therefore, the previous lemma allows us to state Theorem \ref{Theorem 1} as follows:

\begin{theorem}
With the same notation as in Section \ref{State of the Main Results}, if $N=n$, $K=n$ and $d=n$,
$$d_2(\mathcal L,\mathcal Q)\succeq \frac{\sqrt{n}}{\log n}.$$
\end{theorem}

In fact, we can get an improvement of Lemma \ref{equivalence between d+1 and d}. Indeed, in our particular situation the distances $d_1$ and $d_2$ are equivalent in the following sense:

\begin{lemma}
The following statements hold:
\begin{enumerate}
\item[a)] If $d_1(\mathcal L,\mathcal Q)\leq \lambda$, then $d_2(\mathcal L,\mathcal Q)\leq 2\lambda.$

\item[b)] If $d_2(\mathcal L,\mathcal Q)\leq \lambda$, then $d_1(\mathcal L,\mathcal Q)\leq 4\lambda +1.$
\end{enumerate}

\end{lemma}

\begin{proof}
Part a) follows by Equation (\ref{intersection}) in Remark \ref{third distance}. Indeed,

\begin{align*}
d_1(\mathcal L,\mathcal Q)\leq \lambda \Leftrightarrow \tilde{\mathcal Q}\subseteq \lambda \tilde{\mathcal L}\Rightarrow \mathcal Q_h\subseteq \lambda \mathcal L_h\Rightarrow \hat {\mathcal Q}\subseteq 2\lambda \hat {\mathcal L}\Leftrightarrow d_2(\mathcal L,\mathcal Q)\leq 2\lambda.
\end{align*}

To see part b), let's define

\begin{enumerate}
\item[] $P_0=P_0(a,b|x,y)=\frac{1}{K^2}$ for every $x,y=1,\cdots ,N; a,b=1,\cdots ,K$ and
\item[]$\varphi_0=\varphi_0(a,b|x,y)=\frac{1}{N^2}$ for every $x,y=1,\cdots ,N; a,b=1,\cdots ,K.$
\end{enumerate}
Then, we obtain an element $P_0\in \mathcal L$ and a linear functional $\varphi_0$ on $\R^{N^2K^2}$ such that $\varphi_0(P_0)=1$ and $\sup_{Q\in \mathcal Q}\varphi_0(Q)=1$.

Now, for any linear functional $\psi$ on $\R^{N^2K^2}$, we can write $\psi=\psi_1+\psi(P_0)\varphi_0$ where $\psi_1=\psi-\psi(P_0)\varphi_0$. Note that $\psi_1$ is a linear functional on $\R^{N^2K^2}$ such that $\psi_1(P_0)=0$. Therefore, according to Lemma \ref{equivalence between d+1 and d}, if $d_2(\mathcal L,\mathcal Q)\leq \lambda$ we have
\begin{align*}
 &\sup_{Q\in \tilde{\mathcal Q}}\psi(Q)=\sup_{Q\in \tilde{\mathcal Q}}(\psi_1(Q)+\psi(P_0)\varphi_0(Q))\leq
   \omega_{\Psi_1}(\mathcal Q)+\psi(P_0) \\
 &\leq \lambda\omega_{\Psi_1}(\mathcal L)+\psi(P_0)\leq 2\lambda \sup_{P\in \tilde{\mathcal L}}\psi_1(P)+ \psi(P_0)\leq (4\lambda +1)\sup_{P\in \tilde{\mathcal L}}\psi(P). \qedhere
\end{align*}

\end{proof}


\subsection[The $NSG$ space and upper bounds]{The $NSG$ space and upper bounds}\label{NSG}

In this section we will show that there exits a canonical identification between the problems of computing the largest violation of a Bell inequality $LV(M)$  and the quotient $$\frac{\|M\|_{\ell_1^n(\ell_\infty^n)\otimes_{\min} \ell_1^n(\ell_\infty^n)}}{\|M\|_{\ell_1^n(\ell_\infty^n)\otimes_\epsilon \ell_1^n(\ell_\infty^n)}}.$$This means an improvement of \cite[Lemma 1]{JPPVW} (see Theorem \ref{connection:min-epsilon}), where just the implication

\begin{align*}
\frac{\|M\|_{\ell_1^n(\ell_\infty^n)\otimes_{\min} \ell_1^n(\ell_\infty^n)}}{\|M\|_{\ell_1^n(\ell_\infty^n)\otimes_\epsilon \ell_1^n(\ell_\infty^n)}}\preceq LV(\tilde{M})
\end{align*} is shown. As a consequence of this, we will provide upper bounds for the violation of a Bell inequality as a function of the number of inputs, the number of outputs and the dimension of the Hilbert space. This will show that our Theorem \ref{Theorem-construction} is almost optimal in all the parameters of the problem.


According to the previous subsection, the problem of computing the  largest violation of Bell inequalities is exactly the same as computing the distance $d(\tilde{\mathcal L}, \tilde{\mathcal Q})$, where we denote $\tilde{\mathcal L}=conv(\mathcal L \cup -\mathcal L)$ and $\tilde{\mathcal Q}=conv(\mathcal Q \cup -\mathcal Q)$.

Let's start defining the complex linear space

\begin{align*}
NSG(N,K)=\{\{R(a|x)\}_{x,a=1}^{N,K}\in \C^{NK}: \sum_{a=1}^KR(a|x)={\rm constant}\in \C  \text{   for every   } x\}.
\end{align*}
It is not difficult to see that $dim(NSG(N,K))=NK-N+1$. We will identify the algebraic dual space $NSG(N,K)^*$ with $\C^{NK}/NSG(N,K)^ \perp$, where

\begin{align*}
NSG(N,K)^ \perp=\{B\in \C^{NK}:B(R)=0 \text{   for every   } R\in NSG(N,K)\}
\end{align*}is the orthogonal space of $NSG(N,K)$.

On the other hand, we will consider the family
\begin{align*}
 I=\{\{E_x^a\}_{x,a=1}^{N,K}\big|E_x^a \gl 0 , \text{  and  } \sum_{a=1}^KE_x^a=1  \text{ for every } x\}.
\end{align*}Then, it is not difficult to see that the map $$J:NSG(N,K)^*\rightarrow \bigoplus_{\{E_x^ a\}\in I}B(H_{\{E_x^ a\}})$$ given by $$\{\rho(x|a)\}_{x,a}\mapsto \big(\sum_{x,a=1}^{N,K}\rho(x|a)E_x^ a\big)_{\{E_x^a\}\in I}$$ is well defined and it defines an operator system structure on $NSG(N,K)^*$ (see \cite{Paulsen} for the definition of operator system). Then, as we explained in Section \ref{Mathematical Tools}, $NSG(N,K)$ has a natural operator space structure as dual space of  $NSG(N,K)^*$.

\begin{remark}
Duality in the category of operator system is in general a tricky point and we will disregard this problem here.
\end{remark}

\begin{remark}\label{FreeProduct}
In \cite{JNPPSW} the authors show that the map $$\iota:NSG(N,K)^*\rightarrow \star_{i=1}^N\ell_\infty^K$$defined by $\iota(e_{x,a})=\pi_x(e_a)$, where $e_a$ is the $a$-th canonical vector in $\ell_\infty^N$ and $\pi_x:\ell_\infty^N\hookrightarrow \star_{i=1}^N\ell_\infty^K$ is the canonical embedding of $\ell_\infty^N$ into the $x$-th position of the free product, is a completely isometric embedding. Furthermore, the operator system structure on $NSG(N,K)^*$ is exactly the one defined by this embedding.
\end{remark}

The following theorem shows that the operator space  $NSG(N,K)$ is, actually, just a little distortion of $ \ell_\infty^N(\ell_1^{K-1})$.

\begin{theorem}\label{NSG isomorphism}
The map $$T:NSG(N,K)\rightarrow \ell_\infty^N(\ell_1^{K-1})\oplus_\infty \C$$ defined as
\begin{align*}
T(\{R(x|a)\}_{x=1,a=1}^{N,K})=\big(\{R(x|a)\}_{x=1,a=1}^{N;K-1}, \sum_{a=1}^{K}R(x|a)\big)
\end{align*}for every $\{R(x|a)\}_{x=1,a=1}^{N;K}\in NSG(N,K)$ is a completely isomorphism with $\|T\|_{cb}\leq 1$ and $\|T^{-1}\|_{cb}\leq 9$. Here, $T^{-1}$  is defined as
\begin{align*}
T^{-1}\big((\{R(x|a)\}_{x=1,a=1}^{N;K-1}, R)\big)=\big\{\{R(x|a)\}_{a=1}^{K-1}, R-\sum_{a=1}^{K-1}R(x|a)\big\}_{x=1}^{N}.
\end{align*}
\end{theorem}

\begin{proof}
Clearly, the map $S_1:NSG(N,K)\rightarrow \C$ defined by $S_1(\{R(x|a)\})=\sum_{a=1}^{K}R(x|a)$, verifies that $\|S_1\|=\|S_1\|_{cb}\leq 1$.
Therefore it suffices to study the cb-norm of the map $P_1:NSG(N,K)\rightarrow \ell_\infty^N(\ell_1^{K-1})$ defined by
\begin{align*}
P_1(\{R(x|a)\}_{x=1,a=1}^{N;K})=\{R(x|a)\}_{x=1,a=1}^{N;K-1}
\end{align*}for every $\{R(x|a)\}_{x=1,a=1}^{N;K}\in NSG(N,K)$. Note that, by definition,
\begin{align*}
\|P_1\|_{cb}=\|\sum_{x,a=1}^{N,K-1}e_{x,a}\otimes (e_x\otimes e_a)\|_{NSG(N,K)^*\otimes_{\min}\ell_\infty^N(\ell_1^{K-1})}\\=\sup_{\{E_x^ a\}_{x,a=1}^{N,K}\in I}\|\sum_{x,a=1}^{N,K-1}E_x^ a\otimes (e_x\otimes e_a)\|_{B(H_{\{E_x^ a\}})\otimes_{\min}\ell_\infty^N(\ell_1^{K-1})}\leq 1.
\end{align*}
Here the last inequality follows by the fact that for any $\{E_x^ a\}_{x,a=1}^{N,K}\in I$, the map $\ell_1^N(\ell_\infty^{K-1})\rightarrow B(H)$ defined by $e_x\otimes e_a\mapsto E_x^ a$ is a completely contraction (see \cite[Section 8]{JPPVW}).

To study $\|T^{-1}\|_{cb}$, again it is clear that the map $S_2:\C \rightarrow NSG(N,K)$ defined by $S_2(R)=(\overbrace{0,\cdots ,0}^K,R)_{x=1}^N$ verifies $\|S_2\|=\|S_2\|_{cb}\leq 1$. Therefore, it suffices to show that $\|P_2\|_{cb}\leq 8$, where
$P_2:\ell_\infty^N(\ell_1^{K-1})\rightarrow NSG(N,K)$ is defined as
\begin{align*}
P_2(\{R(x|a)\}_{x=1,a=1}^{N;K-1})=\big\{\{R(x|a)\}_{a=1}^{K-1}, -\sum_{a=1}^{K-1}R(x|a)\big\}_{x=1}^{N}.
\end{align*}
To see this, consider an element $x=\sum_{x,a=1}^{N,K-1}E_x^ a\otimes (e_x\otimes e_a)$ of norm $1$. By using the same argument as in \cite[Theorem 6]{JPPVW} we can assume, up to a constant $C=4$ in the norm, that the operators $\{E_x^ a\}_{x,a=1}^{N,K-1}$ are positive. Now, we have to check that $\|(\id\otimes P_2)(x)\|_{B(H)\otimes_{\min} NSG(N,K)}\leq 2$. Equivalently, we will show that the associated operator to $(\id\otimes P_2)(x)$,  $\hat{x}:NSG(N,K)^*\rightarrow B(H)$, defined as $$\hat{x}\big(\{\rho(x,a)\}\big)=\sum_{x=1}^N\sum_{a=1}^{K-1}E_x^ a(\rho(x|a)-\rho(x|K)),$$verifies $\|\hat{x}\|_{cb}\leqslant 2$.  To see this, consider the maps $U,V:NSG(N,K)^*\rightarrow B(H)$ defined as follows:
 \begin{align*}
 U(\{\rho(x,a)\})= \sum_{x,a=1}^{N,K}\rho(x,a)\hat{E}_x^a,
 \end{align*}
where $\hat{E}_x^a=E_x^ a$ for every $a=1,\cdots,K-1$ and $\hat{E}_x^K=1-\sum_{a=1}^{K-1}\hat{E}_x^a$ for every $x=1,\cdots,N$ and
\begin{align*}
V(\{\rho(x,a)\})=\sum_{x,a=1}^{N,K}\rho(x,a)\hat{F}_x^a,
\end{align*}where $\hat{F}_x^a=0$ for every $a=1,\cdots,K-1$ and $\hat{F}_x^K=1$ for every $x=1,\cdots,N$. It is trivial to see that $\{\hat{E}_x^a\}_{x,a}$ (reps. $\{\hat{F}_x^a\}_{x,a}$) is a family of positive operators such that $\sum_{a=1}^{K}\hat{E}_x^a=1$ (resp. $\sum_{a=1}^{K}\hat{F}_x^a=1$). Then, by the very definition of the operator space $NSG(N,K)^*$ it is clear that $\|U\|_{cb},\|V\|_{cb}\leq 1$. On the other hand, we have $\hat{x}=U-  V$. Therefore, $\|\hat{x}\|_{cb}\leq 2.$
\end{proof}

The following corollary is a direct consequence of the metric mapping property of the norms $\pi$ and $\wedge$ in their corresponding categories (see \cite{Def}, \cite{Pisierbook}).

\begin{corollary}\label{NSGTensor}
For every natural numbers $N,K$,

\begin{enumerate}

\item[a)] $\tilde{T}:=T\otimes T: NSG(N,K)\otimes_\pi NSG(N,K)\rightarrow (\ell_\infty^N(\ell_1^{K-1})\bigoplus_\infty \C)\otimes_\pi (\ell_\infty^N(\ell_1^{K-1})\bigoplus_\infty \C)$ defines an isomorphism with $\|\tilde{T}\|\|\tilde{T}^{-1}\|\leq 81$.

\item[b)] $\tilde{T}:=T\otimes T: NSG(N,K)\otimes_\wedge NSG(N,K)\rightarrow (\ell_\infty^N(\ell_1^{K-1})\bigoplus_\infty \C)\otimes_\wedge (\ell_\infty^N(\ell_1^{K-1})\bigoplus_\infty \C)$ defines a completely isomorphism with $\|\tilde{T}\|_{cb}\|\tilde{T}^{-1}\|_{cb}\leq 81.$
\end{enumerate}

\end{corollary}

\begin{remark}
Actually, it is easy to see that $\|T^{-1}\|\leq 3$ in Theorem \ref{NSG isomorphism} when we restrict to real Banach spaces. In particular, in that case we obtain $\|\tilde{T}\|\|\tilde{T}^{-1}\|\leq 9$ in part a) of Corollary \ref{NSGTensor}.
\end{remark}

The following lemma shows that $NSG$ is the suitable space to describe the sets $\tilde{\mathcal L}$  and $\tilde{\mathcal Q}$. However, since $\tilde{\mathcal L}$  and $\tilde{\mathcal Q}$ are real sets we have to restrict the previous space to the real part. Taking the $\pi$ tensor product for real Banach spaces is well-defined and provides the correct dual. The correct way to understand the tensor product of the real space  $NSG(N,K)_{\R}$ is to use the operator system  $V=NSG(N,K)^*\ten_{\min}NSG(N,K)^*$, and then consider the real part $V_{sa}$. Then we may define the tensor product
of real coefficients as
 \[  NSG(N,K)_{\R}\pl \hat{\ten}_{\R} NSG(N,K)_{\R} \pl:=\pl
  V_{sa}^* \pl .\]
Using the fact that every element $\xi\in V$ can be written as
$\xi=\xi_1+i\xi_2$ with $\max\{\|\xi_1\|,\|\xi_2\|\}\le \|\xi\|$ we deduce that
 \begin{align*}
 &\frac{1}{2} B_{NSG(N,K)_{\R}\pl \hat{\ten}_{\R} NSG(N,K)_{\R}}
  \subset  B_{NSG(N,K)\ten_{\wedge}NSG(N,K)} \\
 & \lel B_{V^*} \subset B_{V_{sa}^*}+i B_{V_{sa}^*}
 \lel     B_{NSG(N,K)_{\R}\pl \hat{\ten}_{\R} NSG(N,K)_{\R}}
  + i  B_{NSG(N,K)_{\R}\pl \hat{\ten}_{\R} NSG(N,K)_{\R}} \pl .
  \end{align*}

In order to describe the set $\tilde{\mathcal L}$, it can be seen that for every element $R\in NSG(N,K)_{\R}$ we have

\begin{align*}
||R||_{NSG(N,K)}=\inf \{|\lambda|+|\mu|:R=\lambda P+\mu Q: P,Q\in S(N,K)\},
\end{align*} where

\begin{align*}
S(N,K)=\{\{P(x|a)\}_{x,a=1}^{N,K}: P(x|a)\geq 0 \text{   for every   } x,a   \text{   and   } \sum_{a=1}^KP(x|a)=1  \text{   for every   } x\}.
\end{align*}That is, the norm $||\cdot||_{NSG(N,K)}$ coincides with the Minkowski functional of the set $\tilde{S}=conv(S\cup -S)$ when we consider real coefficients.


With this at hand, we have

\begin{lemma}\label{connectionNSG}
\

\begin{enumerate}
\item[a)] $\tilde{\mathcal L}=B_{NSG(N,K)_{\R}\otimes_\pi NSG(N,K)_{\R}}\big(=B_{\big(NSG(N,K)_\R^*\otimes_{\epsilon} NSG(N,K)_\R^*\big)^*}\big)$.
\item[b)] $\tilde{\mathcal Q}=B_{NSG(N,K)_{\R} \hat{\otimes}_{\R} NSG(N,K)_{\R}}:=B_{ V_{sa}^*}$.

\end{enumerate}

\end{lemma}

\begin{proof} To prove part a), let's denote, for a couple of sets $A$ and $B$, $A\otimes B=\{a\otimes b:a\in A, b\in B\}$.

Now, as we have said, $B_{NSG(N,K)_\R}=conv(S\cup -S)$. On the other hand, by definition $$\tilde{\mathcal{L}}=conv(conv(S\otimes S) \cup -conv(S\otimes S))=conv(S\otimes S \cup -S\otimes S).$$Then, it follows by the well known fact $B_{X\otimes_\pi Y}=conv(B_X\otimes B_Y)$ (see \cite{Def}) that

\begin{align*}
B_{NSG(N,K)_\R\otimes_\pi NSG(N,K)_\R}=conv(conv(S\cup -S)\otimes conv(S\cup -S))\\=conv(S\otimes S \cup -S\otimes S)=\tilde{\mathcal{L}}.
\end{align*}The second equality follows trivially by the duality between the $\epsilon$ and $\pi$ tensor norms.

In order to prove part b), recall that the map

\begin{align*}
J\otimes J: NSG(N,K)^*\otimes_{min} NSG(N,K)^*\longrightarrow \oplus_{\{E_x^a\}\in I}B(H_{\{E_x^a\}})\otimes_{min} \oplus_{\{F_y^b\}\in I}B(H_{\{F_y^ b\}})\\ \subset \oplus_{(\{E_x^a\}\times \{F_y^b\})\in I\times I}B(H_{\{E_x^a\}}\otimes H_{\{F_y^b\}})
\end{align*}is a (completely) isometry. Then, the result follows easily just reasoning by duality on real elements $(M_{x,y}^{a,b})_{x,y;a,b}$.
\end{proof}

\begin{remark}
The space $NSG(N,K)$ represents the formalization of the comments after Theorem \ref{complemented} . Actually, in \cite{JNPPSW} it is shown the utility of this space to study the set of probability distributions when we assume another very interesting model of Nature.
\end{remark}

To finish this section, we will use Corollary \ref{NSGTensor}  and Lemma \ref{connectionNSG}  to obtain upper bounds for the largest violation of Bell inequalities.

\begin{theorem}\label{upperbounds}
With the notation of subsection \ref{geometric_int}, if we consider $N$ inputs, $K$ outputs and Hilbert space dimension $d$, we have that $$d(\tilde{\mathcal L}, \tilde{\mathcal Q})\preceq \min\{N,K,d\}.$$
\end{theorem}

The upper bound regarding the Hilbert space dimension $d$ was proved in \cite[Proposition 2]{JPPVW}. On the other hand, Theorem \ref{upperbounds} improves the previously known $O(K^2)$ upper bounds for $d(\tilde{\mathcal L}, \tilde{\mathcal Q})$ as a function of the number of outputs (see \cite[Proposition 25]{DKLR}). To our knowledge, no nontrivial upper bound for $d(\tilde{\mathcal L}, \tilde{\mathcal Q})$
as a function of the number of inputs was known.

\begin{proof} The upper bound in terms of  $d$ was proved in \cite[Proposition 2]{JPPVW}.  It remains to show the upper bound
in terms of the outputs or the inputs.  Since we allow absolute constants we may work with complex Banach spaces. Now, according to Corollary \ref{NSGTensor} and Lemma \ref{connectionNSG} it suffices to show that
\begin{align*}
\|\id\otimes \id:\ell_1^N(\ell_\infty^K)\otimes_\epsilon \ell_1^N(\ell_\infty^K)\rightarrow \ell_1^N(\ell_\infty^K)\otimes_{\min} \ell_1^N(\ell_\infty^K)\|\preceq \min\{N,K\}.
\end{align*}
Recall that (\cite{Tomczak}) $d(\ell_1^K,\ell_{\infty}^K)\kl \sqrt{K}$.
Thus there exist $u:\ell_\infty^k\to \ell_1^k$ such that  $\|u^{-1}\|\le 1$ and
 \[  \|u:\ell_\infty^k\to \ell_1^k\|\kl \sqrt{K} \pl .\]
Moreover, we know that $\|u^{-1}\|_{cb}\lel \|u^{-1}\|\le 1$.
 According to Grothendieck's inequality we also have
  \[ \|u:\ell_{\infty}^k\to \ell_1^k\|_{cb}\le K_G\|u\| \kl K_G \sqrt{K}\pl .\]
We use the amplifications $\tilde{u}=\id\otimes u:\ell_1^N(\ell_\infty^K)\rightarrow \ell_1^N(\ell_1^K)$ and $\tilde{u}^{-1}=\id \otimes u^{-1}:\ell_1^N(\ell_1^K)\rightarrow \ell_1^N(\ell_\infty^K)$. Then,  we have a factorization
\begin{align*}
 \id\otimes \id=(\tilde{u}^{-1}\otimes\tilde{u}^ {-1})\circ (\id\otimes \id)\circ (\tilde{u}\otimes \tilde{u}):\ell_1^N(\ell_\infty^K)\otimes_\epsilon \ell_1^N(\ell_\infty^K)\rightarrow \ell_1^N(\ell_\infty^K)\otimes_{\min} \ell_1^N(\ell_\infty^K)\pl. \end{align*}
 This implies
 \begin{align*}
 \|\id\otimes \id:\ell_1^N(\ell_\infty^K)\otimes_\epsilon \ell_1^N(\ell_\infty^K)\rightarrow \ell_1^N(\ell_\infty^K)\otimes_{\min} \ell_1^N(\ell_\infty^K)\|\preceq K.
\end{align*}
To prove the upper bound as a function of the number of inputs $N$, recall (see for instance \cite{Pisierbook2}) that
 \begin{align*}
\|\id\otimes \id:\ell_\infty^N(\ell_\infty^K)\rightarrow  \ell_1^N(\ell_\infty^K)\|_{cb}\leq N.
\end{align*}
Then, the factorization
\begin{align*}
 \id\otimes \id:\ell_1^N(\ell_\infty^K)\otimes_\epsilon \ell_1^N(\ell_\infty^K)\rightarrow \ell_\infty^N(\ell_\infty^K)\otimes_{\min} \ell_1^N(\ell_\infty^K)\rightarrow \ell_1^N(\ell_\infty^K)\otimes_{\min} \ell_1^N(\ell_\infty^K)
\end{align*}
implies
\begin{align*}
 \|\id\otimes \id:\ell_1^N(\ell_\infty^K)\otimes_\epsilon \ell_1^N(\ell_\infty^K)\rightarrow \ell_1^N(\ell_\infty^K)\otimes_{\min} \ell_1^N(\ell_\infty^K)\|&\leq N. \qedhere
\end{align*}
\end{proof}


\section{A relaxation of the problem: The $\gamma_2^*$ norm}\label{section-gamma}

The problem of computing or approximating the classical and quantum value of Bell inequalities has been studied from different points of view. On the one hand, the problem of studying the quantum value of Bell inequalities (related to the study of deciding whether a given probability distribution belongs to the quantum set $\mathcal{Q}$), has captured the interest of many researchers in QIT (see e.g. \cite{LD}, \cite{DLTW}, \cite{S.Wehner}, \cite{NPA1}, \cite{NPA2}). On the other hand, since any two-prover one-round game can be seen as a Bell inequality, Game Theory and, in general Computer Science, can be considered as a very important source of results about the (mostly) classical value of Bell inequalities.

\

The study of two-prover one-round games is of great interest in Computer Science owing to the fact that many of the most important problems in Complexity Theory can be stated in terms of these kinds of games. In particular, they are extremely useful to study problems of hardness of approximation. One example of this is the so called \emph{unique game conjecture} (see \cite{K}, \cite{KV}), which has become one of the crucial problems in Complexity Theory since it implies hardness of approximation results for several important problems (MaxCut, Multicut and Sparsest Cut, Vertex Cover,...) which are difficult to obtain by standard complexity assumptions. Although the results in Complexity Theory mainly focus on the classical value of games, recently some of the most relevant problems in the field have been studied in the context of quantum physics or under the assumption of the non-signally condition. Some examples of this are the study of hardness of approximation of the quantum value (commonly called entangled value) of games (\cite{KKMTV},\cite{IKM}), the unique game conjecture in the quantum context (\cite{KRT}) or the parallel repetition theorem (\cite{Holenstein}, \cite{KRT}, \cite{CSUU}, \cite{KR}). The standard way to tackle these kinds of problems is to show that the entangled value (resp. non signally value) of the considered games can be approximated by some semidefinite programming (SDP) relaxation with some good properties. In this sense, the following relaxation has been shown to be very useful (see \cite{BHHRRS} and \cite{KRT} for details).

\

For every $M$ consider the following optimization problem (OP), which maximizes over complex vectors $\{u_x^a\}_{x,a=1}^n, \{v_y^b\}_{y,b=1}^n$ and $z$:

\begin{OP}\label{OP}

$$\omega_{op}(M):=\max  \Big\{\big|\sum_{x,y,a,b=1}^nM_{x,y}^{a,b}\langle u_x^a, v_y^b\rangle\big|\Big\}$$

subject to: $$\left\{\begin{array}{c}
                         \hspace{-5.5cm} \|z\|=1,\\
                         \hspace{-2.2cm}\forall x,y, \sum_au_x^a=\sum_bv_y^b=z, \\
                         \forall x,y, \forall a\neq b, \langle u_x^a, u_x^b\rangle=0,  \langle v_y^a, v_y^b\rangle=0. \\
                         \end{array}\right.$$

\end{OP}

Following the notation in \cite{KRT}, the relaxation we are considering here verifies $$\omega_{sdp3}(M)\geq \omega_{op}(M)\geq \omega_{sdp1}(M)\geq \sup\{|\langle M,Q\rangle|:Q\in \mathcal Q\}.$$ Indeed, it is easy to see that $SDP_3$ and $SDP_1$ can be stated equivalently for real or complex vectors. $SDP_3$ and $SDP_1$ have been shown to be very useful in the study of different problems (see \cite{KRT}, \cite{KR}, \cite{NPA2}). In particular, they can be used to approximate the entangled value of unique games.

Actually, $SDP_1$ is obtained when we consider the extra restriction of $\langle u_x^a, v_y^b\rangle\geq 0$ for every $x,y, a, b$ in our problem OP. Note that this $SDP_1$ was already considered in the context of Bell inequalities in \cite{NPA1, NPA2} (certificate of order 1). As far as we know, it is an open question whether the quantum value of a Bell inequality (in particular, the entangled value of a game) can be efficiently approximated up to a universal constant. $SDP1$ (and the certificates of higher order in \cite{NPA1, NPA2}) seems to be the best known candidate to approximate such a value by using semidefinite programming.

In this section, we will study the problem $OP$ \ref{OP} to compute the classical and the quantum value of Bell inequalities. Our main theorem states:

\begin{theorem}\label{SDP}

There exists a Bell inequality $M$ with $n$ inputs and $n+1$ outputs such that

\begin{enumerate}
\item[a)]$\sup_{P\in\mathcal{L}} |\langle M,P\rangle|\preceq 1,$
\item[b)]$\frac{\sqrt{n}}{\log n}\preceq\sup_{Q\in \mathcal{Q}}|\langle M,Q\rangle|\preceq\sqrt{n}$ and
\item[c)]$\frac{n}{\log n}\preceq\omega_{op}(M)\preceq \frac{n}{(\log n)^{\beta}}$ for certain universal constant $\beta$.
\end{enumerate}
\end{theorem}

Actually, the element $M$ in Theorem \ref{SDP} is exactly the same as the one considered in Theorem \ref{Theorem 1}. Furthermore, the key point to prove this result is again the complemented copy of $\ell_2^n$ provided in Theorem \ref{complemented}. Thus, we have a \emph{canonical object} $M$ which allows us to ``separate'' the three different models: classical, quantum and $OP$ (see explanation after Theorem \ref{complemented}).

In terms of distance, we immediately deduce:

\begin{corollary}\label{distances_SDP}
There exists a Bell inequality $M$ with $n$ inputs and $n+1$ outputs such that
\begin{enumerate}
\item[a)]$\frac{\omega_{op}(M)}{\sup_{P\in\mathcal{L}} |\langle M,P\rangle|}\succeq \frac{n}{\log n}$.
\item[b)]$\frac{\omega_{op}(M)}{\sup_{Q\in \mathcal{Q}}|\langle M,Q\rangle|}\succeq \frac{\sqrt{n}}{\log n}$.
\end{enumerate}

\end{corollary}

Thus, even when this $OP$ is very useful to approximate the entangled value of some kinds of games, we can not use it to approximate the value of a general Bell inequality. Moreover, we will show that part a) in Corollary \ref{distances_SDP} is essentially optimal. Specifically,

\begin{prop}\label{optimality_SDP_inputs}
For every Bell inequality $M$ with $n$ inputs and $n$ outputs we have
\begin{align*}
\frac{\omega_{op}(M)}{\sup_{P\in\mathcal{L}} |\langle M,P\rangle|}\preceq n.
\end{align*}
\end{prop}

Furthermore, as a consequence of a deep result in \cite{BCLT}, we can prove a much sharper result about the previous optimality when we consider the rank of the operator $M$. Indeed, we will show

\begin{theorem}\label{optimality-SDP rank}

Let's denote, for each $n\in \N$,

\begin{align*}
A_n=\sup\left\{A: \frac{\omega_{op}(M)}{\sup_{P\in\mathcal{L}} |\langle M,P\rangle|}\leq A n\right\},
\end{align*}where the $\sup$ runs over all Bell inequalities $M$ with rank $n$.

Then,

\begin{align*}
\frac{1}{\log n}\preceq A_n \preceq \frac{1}{(\log n)^\beta}
\end{align*}
for certain universal constant $\beta$.
\end{theorem}

As we did in the previous sections (see Lemma \ref{connectionNSG}), we will regard Bell inequalities as elements in $(NSG(n,n)^*\otimes NSG(n,n)^*)_{\R}$. Actually, since we are working up to universal constants, we can deal with the complex linear space $NSG(n,n)^*\otimes NSG(n,n)^*$ and, via Corollary \ref{NSGTensor}, with $\ell_1^n(\ell_\infty^n)\otimes \ell_1^n(\ell_\infty^n)$. In the same way as we described the classical (resp. Quantum) value of a Bell inequality via the tensor norm $\epsilon$ (resp. $min$), we will be able to describe the value $\omega_{op}$ via another tensor norm. In this way, the results of this section emphasize again on the importance of the Banach space techniques to study this kind of problems.

We will need the following two definitions (see \cite{Pisierbook4}, \cite{Pisierbook} respectively):

\begin{definition}
Let $X$ and $Y$ be two Banach spaces. For every $z\in X\otimes Y$ we define

\begin{align*}
\gamma_2^*(z)=\sup \{\|(u\otimes v)(z)\|_{H\otimes_\pi H}\},
\end{align*}
where the $\sup$ is taken over all Hilbert spaces $H$ and all contractions $u:X \rightarrow H$, $v:Y \rightarrow H$.
\end{definition}




\begin{definition}
Given two operator spaces $X$, $Y$, for every $z=\sum_{i=1}^nx_i\otimes y_i\in X\otimes Y$ we define

\begin{align*}
\|z\|_h=\sup\{\|\sum_{i=1}^nu(x_i)v(y_i)\|_{B(H)}\},
\end{align*}where the $\sup$ is taken over all Hilbert spaces $H$ and all completely contractions $u:X \rightarrow B(H)$, $v:Y \rightarrow B(H)$.
\end{definition}

The following lemma will be very helpful in this section:

\begin{lemma}
For every element $z\in \ell_1^N(\ell_\infty^K)\otimes \ell_1^N(\ell_\infty^K)$ we have $\|z\|_h\leq \gamma_2^*(z)\leq K_G^2\|z\|_h$.
\end{lemma}

\begin{proof}
The first inequality follows from the fact that for any pair of Banach spaces $X$ e $Y$, the map $\id:min(X)\otimes_h min(Y)\rightarrow X\otimes_{\gamma}Y$ defines an isometry (see for instance \cite{BlPa}). Here, given a Banach space $Z$ we denote by $\min(Z)$ the operator space structure endowed by any embedding $Z \hookrightarrow C(K)$. To see the second inequality, consider an element $z$ such that $\|z\|_h\leq 1$. Now, according to Lemma \ref{Grothendieck}, given any contraction $u:\ell_1^N(\ell_\infty^K)\rightarrow \ell_2^n$ (resp. $v:\ell_1^N(\ell_\infty^K)\rightarrow \ell_2^n$) it verifies $\|u:\ell_1^N(\ell_\infty^K)\rightarrow R_n\|_{cb}\leq K_G$ (resp. $\|v:\ell_1^N(\ell_\infty^K)\rightarrow C_n\|_{cb}\leq K_G$). Then, \begin{align*}
\|(u\otimes v)(z)\|_{\ell_2^n\otimes_\pi \ell_2^n}=\|(u\otimes v)(z)\|_{R_n\otimes_h C_n}\leq K_G^2.
\end{align*}
\end{proof}

\begin{remark}\label{NSGHaagerup}
According to Theorem \ref{NSG isomorphism}, we know that for any $z\in NSG(N,K)^*\otimes NSG(N,K)^*$ we have $\|z\|_h\leq \gamma_2^*(z)\leq C\|z\|_h$ for certain universal constant $C$.
\end{remark}

The following lemma shows that $\gamma_2^*(M)$ is, up to a universal constant, the same as $\omega_{op}(M)$.

\begin{lemma}\label{SDP NSG}
For every Bell inequality $M$ with $N$ inputs and $K$ outputs we have

\begin{align*}
\omega_{op}(M)\simeq \|M\|_{NSG(N,K)^*\otimes_{\gamma_2^*} NSG(N,K)^*}.
\end{align*}
\end{lemma}

\begin{proof}

Suppose on the one hand that $\|M\|_{NSG(N,K)^*\otimes_{\gamma_2^*} NSG(N,K)^*}\leq 1$ and consider some vectors $\{u_x^a\}_{x,a}$ and $\{v_y^b\}_{y,b}$ in a Hilbert space $H$ verifying the restrictions in OP \ref{OP}. Then, it is easy to see that the map $u:NSG(N,K)^*\rightarrow H$ given by $u(e_{x,a})=u_x^a$ (resp. $v:NSG(N,K)^*\rightarrow H$ given by $v(e_{y,b})=v_y^b$) is well defined and verifies $\|u\|\leq 1$ (resp. $\|v\|\leq 1$). Thus, we have that
\begin{align*}
|\sum_{x,y,a,b=1}^nM_{x,y}^{a,b}\langle u_x^a, v_y^b\rangle|=|\sum_{x,y,a,b=1}^nM_{x,y}^{a,b}\langle u(e_{x,a}), v(e_{y,b})\rangle|\leq 1.
\end{align*}Therefore, $\omega_{op}(M)\leq \|M\|_{NSG(N,K)^*\otimes_{\gamma_2^*} NSG(N,K)^*}$.

The proof of the second inequality requires a more sophisticated argument. Note that, according to Remark \ref{NSGHaagerup}, we can consider $NSG(N,K)^*\otimes_h NSG(N,K)^*$. On the other hand, the application defined in Remark \ref{FreeProduct},

\begin{align*}
\iota:NSG(N,K)^*\rightarrow *_{i=1}^N\ell_\infty^K,
\end{align*}defines a completely isometric embedding. Then, we know (see \cite[Theorem 5.13]{Pisierbook}) that

\begin{align*}
\iota\otimes \iota:NSG(N,K)^*\otimes_h NSG(N,K)^*\rightarrow (*_{i=1}^N\ell_\infty^K) * (*_{i=1}^N\ell_\infty^K)= *_{i=1}^{2N}\ell_\infty^K
\end{align*}defines a (completely) isometric embedding.

Therefore, for any $M=\sum_{x,y,a,b}M_{x,y}^{a,b}e_{x,a}\otimes e_{y,b}\in NSG(N,K)^*\otimes NSG(N,K)^*$,

\begin{align*}
\|M\|_h=\|(\iota\otimes \iota)(M)\|_{*_{i=1}^{2N}\ell_\infty^K}=\|\sum_{x,y,a,b}M_{x,y}^{a,b}(\iota\otimes \iota)(e_{x,a}\otimes e_{y,b})\|_{*_{i=1}^{2N}\ell_\infty^K}.
\end{align*} Thus, it follows that $\|M\|_h= \sup\{\|\sum_{x,y,a,b}M_{x,y}^{a,b}E_x^aF_b^y\|_{B(H)}\}$, where the $\sup$ is taken over all families of positive operators $\{E_x^ a\}_{x,a}\subset B(H)$ (resp. $\{F_y^ b\}_{y,b}\subset B(H)$) verifying $\sum_a E_x^a=1$ for every $x$ (resp. $\sum_b F_y^b=1$ for every $y$) and $E_x^a\perp E_x^{a'}$ for every $a\neq a'$ (resp. $F_y^b\perp F_y^{b'}$ for every $b\neq b'$). We conclude the proof by using the polarization identity in order to take the supremum of $|\sum_{x,y,a,b}M_{x,y}^{a,b}\langle E_x^a\xi, F_b^y\xi\rangle|$ over all operators as before and $\xi$ in the unit sphere of $H$. Indeed, it is trivial to check that the elements $u_x^a= E_x^a\xi$ and $v_y^b=F_b^y\xi$ verify the restriction in OP \ref{OP}.
\end{proof}

\begin{remark}
In the case of correlation Bell inequalities, this norm exactly describes the set of quantum correlations. Specifically, given a matrix $(T_{i,j})_{i,j=1}^n$, we have (\cite{Tsirelson})

\begin{align*}
\sup\{ \sum_{i,j=1}^nT_{i,j}\gamma_{i,j}:(\gamma_{i,j})_{i,j=1}^n   \text{   is a quantum correlation matrix   } \}=\|\sum_{i,j=1}^nT_{i,j}e_i\otimes e_j\|_{\ell_1^n\otimes_{\gamma_2^*}\ell_1^n}.
\end{align*}

\end{remark}

\begin{remark}
The tensor norm $\gamma_2^*$ has been shown to be a very important tool in Communication Complexity (see \cite{LMSS} and \cite{LS} and the references therein).
\end{remark}

According to Lemma \ref{SDP NSG} and Corollary \ref{NSGTensor}, Theorem \ref{SDP} follows from the next result:

\begin{theorem}\label{gamma}

There exists an element $M\in \ell_1^n(\ell_\infty^n)\otimes \ell_1^n(\ell_\infty^n)$ such that

\begin{enumerate}
\item[a)]$\|M\|_\epsilon\preceq 1,$
\item[b)]$\frac{\sqrt{n}}{\log n}\preceq\|M\|_{min}\preceq\sqrt{n}$ and
\item[c)]$\frac{n}{\log n}\preceq\gamma_2^*(M)\preceq \frac{n}{(\log n)^{\beta}}$.
\end{enumerate}
\end{theorem}

Indeed, once we have an element $M$ as in Theorem \ref{gamma}, we can obtain $\hat{M}$ verifying Theorem \ref{SDP} just adding some extra zeros as it was explained in Theorem \ref{connection:min-epsilon}. Actually, it is easy to see that this is exactly the same as taking
\begin{align*}
\hat M=(T^*|_{\ell_1^n(\ell_\infty^n)}\otimes T^*|_{\ell_1^n(\ell_\infty^n)})(M)\in NSG(n,n+1)^*\otimes NSG(n,n+1)^*,
\end{align*}where $T:NSG(n,n+1)\rightarrow \ell_\infty^n(\ell_1^n)\oplus \C$ is the map defined in Theorem \ref{NSG isomorphism}.

Furthermore, Proposition \ref{optimality_SDP_inputs} and Theorem \ref{optimality-SDP rank} are equivalent, respectively, to

\begin{prop}\label{optimality_SDP_inputsII}
For every element $M\in \ell_1^n(\ell_\infty^n)\otimes \ell_1^n(\ell_\infty^n)$ we have
\begin{align*}
\frac{\|M\|_{\ell_1^n(\ell_\infty^n)\otimes_{\gamma_2^*} \ell_1^n(\ell_\infty^n)}}{\|M\|_{\ell_1^n(\ell_\infty^n)\otimes_\epsilon \ell_1^n(\ell_\infty^n)}}\preceq n.
\end{align*}
\end{prop}and

\begin{theorem}\label{optimality-SDP rankII}

Let's denote, for each $n\in \N$,

\begin{align*}
A_n=\sup\left\{A: \frac{\|M\|_{\ell_1^N(\ell_\infty^K)\otimes_{\gamma_2^*} \ell_1^N(\ell_\infty^K)}}{\|M\|_{\ell_1^N(\ell_\infty^K)\otimes_\epsilon \ell_1^N(\ell_\infty^K)}}\leq A n\right\},
\end{align*}where the $\sup$ runs over all $N,K\in \N$ and $M\in \ell_1^N(\ell_\infty^K)\otimes \ell_1^N(\ell_\infty^K)$ such that $rank(M)=n$.

Then,

\begin{align*}
\frac{1}{\log n}\preceq A_n \preceq \frac{1}{(\log n)^\beta}.
\end{align*}
Here $\beta$ is a universal constant.

\end{theorem}

Note that parts a) and b) in Theorem \ref{gamma} were already proved in Sections \ref{Main result} and \ref{The construction}. So it remains to show that our element $M=\frac{1}{n^2}\sum_{x,y,a,b,k=1}^n\epsilon_{x,a}^k\epsilon_{y,b}^ke_{x,a}\otimes e_{y,b}\in \ell_1^n(\ell_\infty^n)\otimes \ell_1^n(\ell_\infty^n)$ in Section \ref{The construction} verifies part c). By construction, our signs $(\epsilon_{x,a}^k)_{x,a,k=1}^n$ verify Lemma \ref{second} and Lemma \ref{second}. In particular, the map $G^*(e_i\otimes e_j)=\sum_{k=1}^n\epsilon_{i,j}^ke_k$ for every $i,j=1,\cdots ,n$ verifies
$ \|G^*:\ell_1^n(\ell_\infty^n)\rightarrow \ell_2^n\|\preceq n$. Then,

\begin{align*}
\gamma_2^*(M)\succeq \frac{1}{n^4}\|\sum_{x,y,a,b,k=1}^n\epsilon_{x,a}^k\epsilon_{y,b}^kG^*(e_{x,a})\otimes G^*(e_{y,b})\|_{\ell_2^n\otimes_\pi\ell_2^n} \geq \frac{1}{n^4}\sum_{p=1}^n\sum_{x,y,a,b,k=1}^n\epsilon_{x,a}^k\epsilon_{y,b}^k\epsilon_{x,a}^p\epsilon_{y,b}^p
\\=\frac{1}{n^4}\sum_{k,p=1}^n(\sum_{x,a=1}^n\epsilon_{x,a}^k\epsilon_{x,a}^p)^2\geq \frac{1}{n^4}\sum_{k=1}^n(\sum_{x,a=1}^n\epsilon_{x,a}^k\epsilon_{x,a}^k)^2=n.
\end{align*}Therefore, we have proved the first inequality in part c) of Theorem \ref{gamma}.

\begin{remark}
As in the case of $\|M\|_{min}$, the key point of the previous estimate is Theorem \ref{complemented}. However, as we already did in Section \ref{The construction}, we could construct explicitly the elements that we have to use in order to norm $M$. In this particular case we have shown that we can obtain $\gamma_2^*(M)\succeq n$ by using the vectors $u_x^a=v_x^a=\sum_{p=1}^n\epsilon_{x,a}^pe_p\in \ell_2^n$ in the definition of $\gamma_2^*$.
\end{remark}

On the other hand, since the element $M$ we are considering has rank $n$, the upper bound in (Theorem \ref{gamma}, part c)) follows by Theorem \ref{optimality-SDP rankII}. To prove this result, we will need the following theorem:

\begin{theorem}\label{Bourgain}(\cite{BCLT})
There exist universal constants $\alpha, \beta> 0$ such that for every $n\in \N$ and every pair of contractions $u:\ell_1(c_0)\rightarrow \ell_2^n$ and $v:\ell_2^n\rightarrow \ell_1(c_0)$ verifying $u\circ v=\id_{\ell_2^n}$, we have $\|u\|\|v\|\geq \alpha (\log n)^\beta$.
\end{theorem}

\begin{remark}
Note that Theorem \ref{Bourgain} says that if we have a $C$-complemented copy of $\ell_2^n$ into $\ell_1(c_0)$, then $C\geq \alpha (\log n)^\beta$. On the other hand, in Theorem \ref{complemented} we provided a $\sqrt{\log n}$-complemented copy of $\ell_2^n$ into $\ell_1^n(\ell_\infty^n)$.
\end{remark}

With this at hand, we can prove Theorem \ref{optimality-SDP rankII}.

\begin{proof}[Proof of Theorem \ref{optimality-SDP rankII}]
The first inequality is consequence of the comments above. Indeed, we have shown that our particular element $M$ verifies $\frac{\gamma_2^*(M)}{\|M\|_\epsilon}\succeq\frac{n}{\log n}$.

For the second inequality, let $z$ be a rank $n$ element in $\ell_1(c_0)\otimes \ell_1(c_0)$ such that $\|z\|_\epsilon\leq 1$. Assume that $\|z\|_{\gamma_2^*}= C_nn$, where $C_n$ is a constant which may depend on $n$. We want to show that
\begin{align*}
C_n\preceq\frac{1}{(\log n)^\beta}.
\end{align*}

If we denote $T_z\in \ell_\infty(\ell_1)\rightarrow \ell_1(c_0)$ the associated operator to $z$, we know by hypothesis that there exit contractions $u:\ell_2^n\rightarrow \ell_\infty(\ell_1)$ and $v:\ell_1(c_0)\rightarrow \ell_2^n$ such that

\begin{align*}
C_nn=|tr(v\circ T_z\circ u:\ell_2^n\rightarrow \ell_\infty(\ell_1)\rightarrow \ell_1(c_0)\rightarrow \ell_2^n)|.
\end{align*}
Indeed, this is an immediate consequence of the very definition of the norm $\gamma_2^*$ and the fact that the operator $T_z$ has rank $n$.
On the other hand, it is immediate that

\begin{align*}
C_nn=|tr(v\circ T_z\circ u)|\leq \sum_{i=1}^na_i(v\circ T_z\circ u)=\sum_{i=1}^{\delta n}a_i(v\circ T_z\circ u)+  n a_{\delta n}(v\circ T_z\circ u),
\end{align*}
where $a_i(v\circ T_z\circ u)$ denotes the $i^{th}$ singular value of the operator $v\circ T_z\circ u$ and $\delta=\frac{C_n}{2}$.

Now, it is very easy to see that $\sum_{i=1}^{\delta n}a_i(v\circ T_z\circ u)\leq \delta n \|v\circ T_z\circ u\|\leq \delta n$. It follows that

\begin{align*}
C_nn\leq \delta n+  n a_{\delta n}(v\circ T_z\circ u),
\end{align*}
so $a_{\delta n}(v\circ T_z\circ u)\geq \frac{C_n}{2}$.

This means that there exit operators $p:\ell_2^{\delta n}\rightarrow \ell_2^n$ and $q:\ell_2^n\rightarrow \ell_2^{\delta n}$, such that $\|p\|\|q\|\leq \frac{2}{C_n}$ and $q\circ (v\circ T_z\circ u)\circ p=\id_{\ell_2^{\delta n}}$. In particular, we have the following factorization:

$$ \xymatrix@R=1.5cm@C=1 cm {
{} & {\ell_1(c_0)}\ar[rd]^{q\circ v} & { }\\ {\ell_2^{\delta n}}\ar[rr]^{\id}\ar[ru]^{T_z\circ u\circ p} & {} &
{\ell_2^{\delta n}.}}$$

According to Theorem \ref{Bourgain}, we know that

\begin{align*}
\alpha (\log \frac{C_n n}{2})^\beta\leq \|q\circ v\|\|T_z\circ u\circ p\|\leq \|p\|\|q\|\leq \frac{2}{C_n}.
\end{align*}

That is, $(\frac{2}{C_n})^\frac{1}{\beta}\geq \alpha^\frac{1}{\beta} (\log \frac{C_n n}{2})$. On the other hand,

\begin{align*}
\alpha^\frac{1}{\beta} (\log \frac{C_n n}{2})=\alpha^\frac{1}{\beta} (\log n+\log (\frac{C_n}{2}))
\succeq \log n,
\end{align*}where we have used that $C_n\succeq \frac{1}{\log n}$ in the last inequality. We conclude that $C_n\preceq\frac{1}{(\log n)^\beta}.$

\end{proof}

We will finish this section with the proof of Proposition \ref{optimality_SDP_inputsII}.

\begin{proof}[Proof of Proposition \ref{optimality_SDP_inputsII}]

We have to show that
\begin{align*}
\|\id\otimes \id:\ell_1^n(\ell_\infty^n)\otimes_\epsilon \ell_1^n(\ell_\infty^n)\rightarrow\ell_1^n(\ell_\infty^n)\otimes_{\gamma_2^*} \ell_1^n(\ell_\infty^n)\|\preceq n.
\end{align*}To see this estimate, recall that Grothendieck's theorem can be stated (see \cite{Def} for details) as:

\begin{align*}
\|\id\otimes \id: \ell_1\otimes_\epsilon \ell_1\rightarrow \ell_1 \otimes_{\gamma_2^*} \ell_1\|\leq K_G.
\end{align*} Then, the result follows easily considering the same factorization as in the proof of Theorem \ref{upperbounds}:

\begin{align*}
\id\otimes \id=(\tilde{u}^{-1}\otimes\tilde{u}^ {-1})\circ (\id\otimes \id)\circ (\tilde{u}\otimes \tilde{u}):\ell_1^N(\ell_\infty^K)\otimes_\epsilon \ell_1^N(\ell_\infty^K)\rightarrow \ell_1^N(\ell_\infty^K)\otimes_{\min} \ell_1^N(\ell_\infty^K).
\end{align*}

\end{proof}


\centerline{\sc Acknowledgment}
We would like to thank Timur Oikhberg for drawing our attention to Theorem \ref{Bourgain}.



\begin{thebibliography}{99}

\bibitem{Acin1}  A. Acin, N. Brunner, N. Gisin, S.
Massar, S. Pironio, V. Scarani, \emph{Device-independent security
of quantum cryptography against collective attacks}, Phys. Rev.
Lett. 98, 230501 (2007).


\bibitem{ADGL}  A. Acín, T. Durt, N. Gisin and J. I. Latorre, \emph{Quantum nonlocality in two three-level systems}, Phys. Rev. A 65:052325, (2002).

\bibitem{AGG} A. Acín, R. Gill and N. Gisin, \emph{Optimal Bell tests do not require maximally entangled states}. Phys. Rev. Lett. 95:210402, (2005).


\bibitem{Acin} A. Acin, L. Masanes, N. Gisin, \emph{From Bell's Theorem to Secure Quantum Key
Distribution}, Phys. Rev. Lett. 97, 120405 (2006).


\bibitem{BHHRRS} B. Barak, M. Hardt, I. Haviv, A. Rao, O. Regev, and D. Steurer, \emph{Rounding parallel repetitions
of unique games}, Proc. 49th Annual IEEE Symp. on Foundations of Computer Science (FOCS), 374-383 (2008).

\bibitem{Bell} J.S. Bell, \emph{On the Einstein-Poldolsky-Rosen paradox}, Physics, \textbf{1}, 195 (1964).


\bibitem{Ben-Or} M. Ben-Or, A. Hassidim, H. Pilpel, \emph{Quantum Multi Prover Interactive Proofs with Communicating
Provers}, Proceedings of 49th Annual IEEE Symposium on Foundations of Computer Science (FOCS 2008), arXiv:0806.3982.

\bibitem{BlPa} D. P. Blecher and V. I. Paulsen, \emph{Tensor products of operator spaces}, J. Funct. Anal. 99, 262-292 (1991).

\bibitem{Brassard-review}  G. Brassard, A. Broadbent, A. Tapp, \emph{Quantum Pseudo-Telepathy} Foundations of Physics, Volume 35, Issue 11, 1877 - 1907 (2005).

\bibitem{BBLV} J. Briet, H. Buhrman, T. Lee, T. Vidick, Multiplayer XOR games and quantum communication complexity with clique-wise entanglement,
arXiv:0911.4007


\bibitem{Briet} J. Briët, H. Buhrman, B. Toner, \emph{A generalized Grothendieck inequality and entanglement in
XOR games}, arXiv:0901.2009.

\bibitem{Brunner} N. Brunner, N. Gisin, V. Scarani, C. Simon, \emph{Detection loophole in asymmetric Bell
experiments}, Phys. Rev. Lett., 98, 220403 (2007).


\bibitem{Brunner2} N. Brunner, S. Pironio, A. Acin, N. Gisin, A. A. Methot, V. Scarani, \emph{Testing the Hilbert space
dimension}, Phys. Rev. Lett. 100, 210503 (2008).

\bibitem{BCLT} J. Bourgain, P. G. Casazza, J. Lindenstrauss and L. Tzafriri, \emph{Banach spaces with a unique unconditional basis, up to a permutation}, Memoirs Am. Math. Soc. No. 322, Providence, (1985).

\bibitem{Buhrman} H. Buhrman, R. Cleve, S. Massar, R. de Wolf, \emph{Non-locality and Communication Complexity},  to appear in
Reviews of Modern Physics.

\bibitem{BGW} H. Buhrman, G. Scarpa, R. de Wolf, \emph{Better Non-Local Games from Hidden Matching}, arXiv:1007.2359.


\bibitem{CMM} M. Charikar, K. Makarychev, and Y. Makarychev, \emph{Near-optimal algorithms for unique games}, Proc. 38th ACM STOC, pag. 205-214 (2006).


\bibitem{Cleve2} R. Cleve, D. Gavinsly, R. Jain, \emph{Entanglement-Resistant Two-Prover Interactive Proof Systems and Non-Adaptive Private Information Retrieval Systems}, quant-ph/07071729, (2007).


\bibitem{Cleve} R. Cleve, P. Høyer, B. Toner, and J. Watrous, \emph{Consequences and Limits of Nonlocal Strategies}, Proceedings
of the 19th IEEE Annual Conference on Computational Complexity (CCC 2004), pp. 236- 249 (2004).

\bibitem{CSUU} R. Cleve, W. Slofstra, F. Unger, and S. Upadhyay, \emph{Perfect parallel repetition theorem for quantum XOR proof systems}, Proc. 22nd IEEE Conference on Computational Complexity, 109-114 (2007).



\bibitem{CGLMP} D. Collins, N. Gisin, N. Linden, S. Massar and S. Popescu, \emph{Bell inequalities for arbitrarily
high-dimensional systems}, Phys. Rev. Lett. 88:040404, (2002).

\bibitem{DKLR} J. Degorre, M. Kaplan, S. Laplante, J. Roland, \emph{The communication complexity of non-signaling distributions}, Proc. 34th Int. Symp. of the MFCS, 270-281 (2009).


\bibitem{Def} A. Defant and K. Floret, \emph{Tensor Norms and Operator Ideals}, North-Holland, Amsterdam (1993).


\bibitem{DLTW}  A. C. Doherty, Y-C. Liang, B. Toner, S. Wehner, \emph{The quantum moment problem and
bounds on entangled multi-prover games}, Proceedings of IEEE Conference on Computational Complexity 2008, pages 199-210.


\bibitem{DHR} M. J. Donald, M. Horodecki and O. Rudolph, \emph{The uniqueness theorem for entanglement
measures}, Jour. Math. Phys. 43:4252-4272, (2002).


\bibitem{Eber} P. Eberhard, \emph{Background level and counter efficiencies required for a loophole free
Einstein-Podolsky-Rosen experiment}, Phys. Rev. A 47:R747-R750, (1993).

\bibitem{EffrosRuan} E. G. Effros and Z.-J. Ruan, {\it Operator
spaces}, London Math. Soc. Monographs New Series, Clarendon Press,
Oxford, (2000).

\bibitem{EPR} A. Einstein, B. Podolsky, N. Rosen, \emph{Can Quantum-Mechanical Description of Physical Reality Be
Considered Complete?}, Phys. Rev., {\textbf 47}, 777 (1935).


\bibitem{Gisin} N. Gisin, \emph{Hidden quantum nonlocality revealed by local filters}, Phys. Lett. A 210, 151 (1996).


\bibitem{Holenstein} T. Holenstein, \emph{Parallel repetition: simplifications and the no-signaling
case}, Proceedings of the thirty-ninth annual ACM symposium on Theory of computing (STOC) 2007.


\bibitem{IKM} T. Ito, H. Kobayashi, and K. Matsumoto, \emph{Oracularization and two-prover one-round
interactive proofs against nonlocal strategies}, Proc. 24th IEEE Conference on Computational Complexity, pages 217-228 (2009).


\bibitem{Jain} R. Jain, Z. Ji, S. Upadhyay, J. Watrous, \emph{QIP = PSPACE}, arXiv:0907.4737.

\bibitem{MariusHabi} M. Junge, \emph{Factorization theory for Spaces of Operators}, Habilitationsschrift Kiel, (1996); see also: Preprint server of the university of southern Denmark 1999, IMADA preprint: PP-1999-02.


\bibitem{JNPPSW} M. Junge, M. Navascues, C. Palazuelos, D. Perez-Garcia, V. B. Scholz and R. F. Werner, \emph{ Connes' embedding problem and Tsirelson's problem }, arXiv:1008.1142.


\bibitem{JPPVW} M. Junge, C. Palazuelos, D. Pérez-García, I. Villanueva and M.M. Wolf, \emph{Unbounded violations of bipartite Bell Inequalities via Operator Space theory}, To appear in Commun. Math. Phys.

\bibitem{JPPVW2} M. Junge, C. Palazuelos, D. Pérez-García, I. Villanueva and M.M. Wolf, \emph{Operator Space theory: a natural framework for Bell inequalities}, Phys. Rev. Lett. 104, 170405 (2010).

\bibitem{JuPa} M. Junge and J. Parcet, \emph{Mixed-norm inequalities and operator space $L_p$ embedding theory}, Mem. Amer. Math. Soc. 952, (2010).


\bibitem{JuPi} M. Junge and G. Pisier, \emph{Bilinear forms on exact operator spaces and $B(H)\otimes B(H)$}, Geometric and Functional Analysis, 5, 329-363 (1995).

\bibitem{KKMTV} J. Kempe, H. Kobayashi, K. Matsumoto, B. Toner and T. Vidick, \emph{Entangled games are
hard to approximate}, arXiv:0704.2903v2 (2007).


\bibitem{KR} J. Kempe, O. Regev, \emph{No Strong Parallel Repetition with Entangled and Non-signaling Provers}, arXiv:0911.0201.


\bibitem{KRT} J. Kempe, O. Regev, B. Toner, \emph{The Unique Games Conjecture with Entangled Provers is False}, Proceedings of 49th Annual IEEE Symposium on Foundations of Computer Science (FOCS 2008), quant-ph/0710.0655 (2007).



\bibitem{K} S. Khot, \emph{On the power of unique $2$-prover $1$-round games}, Proc. $34$th ACM Symp.
on Theory of Computing, pages 767-775 (2002).




\bibitem{KV} S. Khot and N. K. Vishnoi. \emph{The unique games conjecture, integrality gap for cut problems
and embeddability of negative type metrics into $\ell_1$}, Proc. $46$th IEEE Symp. on
Foundations of Computer Science, pages 53-62 (2005).


\bibitem{Kw-Woj} S. Kwapien and W. Woyczynski, \emph{Random Series and Stochastic Integrals: Single and Multiple}, Probab. Appl., Birkhäuser Boston, Boston, MA (1982).


\bibitem{LedouxTalagrand} M.~Ledoux, M.~Talagrand, \emph{Probability in Banach Spaces}, Springer-Verlag, (1991).


\bibitem{LD} Y.-C. Liang, A. Doherty, \emph{Bounds on Quantum Correlations in Bell Inequality Experiments}, Phys. Rev. A, 75, 042103 (2007).


\bibitem{LMSS} N. Linial, S. Mendelson, G. Schechtman and A. Shraibman, \emph{Complexity Measures of Sign Matrices}, Combinatorica, 27(4), 439-463 (2007).

\bibitem{LS} N. Linial, A. Shraibman, \emph{Lower Bounds in Communication Complexity Based on Factorization Norms}, Random Structures and Algorithms, 34, 368-394 (2009).

\bibitem{MarcusPisier} M.B. Marcus, G. Pisier, \emph{Random Fourier series with applications to Armonic Analysis}, Annals of Math. Studies,
\textbf{101}, Princeton Univ. Press, (1981).


\bibitem{Mas} L. Masanes, \emph{Universally-composable privacy amplification from causality constraints}, Phys. Rev. Lett. 102, 140501 (2009).


\bibitem{Mas2}  Ll. Masanes, R. Renner, A. Winter, J. Barrett, M.
Christandl, \emph{Security of key distribution from causality
constraints}, quant-ph/0606049 (2006).


\bibitem{Massar} S. Massar, Nonlocality, closing the detection loophole, and communication complexity. Physical
Review A, 65:032121 (2002).

\bibitem{MV} A. A. Methot and V. Scarani, \emph{An anomaly of non-locality}, Quant. Inf. Comput. 7, 157 (2007).

\bibitem{NPA1} M. Navascues, S. Pironio, and A. Acin, \emph{Bounding the Set of Quantum Correlations}, Phys. Rev. Lett. 98, 010401 (2007).


\bibitem{NPA2} M. Navascués, S. Pironio and A. Acín, \emph{A convergent hierarchy of semidefinite programs characterizing the set of quantum correlations}, New J. Phys. 10, 073013 (2008).


\bibitem{Paulsen} V. I. Paulsen, Completely Bounded Maps and Operator Algebras, Cambridge Studies in
Advanced Mathematics 78, Cambridge University Press, Cambridge, 2003.


\bibitem{PWJPV} Pérez-García, D., Wolf, M. M., Palazuelos, C., Villanueva, I. and Junge M., \emph{Unbounded violation of tripartite Bell inequalities}, Commun. Math. Phys. 279 (2), 455-486 (2008).


\bibitem{Pisierbook} G. Pisier, \emph{An Introduction to Operator Spaces}, London Math. Soc. Lecture Notes Series 294, Cambridge University
Press, Cambridge (2003).

\bibitem{Pisierbook2} G. Pisier, \emph{Non-Commutative Vector Valued Lp-Spaces and Completely p-Summing Maps}, Asterisque, 247 (1998).

\bibitem{Pisierbook3} G. Pisier, \emph{The volume of convex bodies and Banach space geometry}, Cambridge Tracts in Mathematics, vol. 94, Cambridge University Press, Cambridge, (1989).

\bibitem{Pisierbook4} G. Pisier, \emph{Factorization of linear operators and geometry of Banach spaces}, CBMS 60 (1986).


\bibitem{Pisierbook5} G. Pisier, \emph{The operator Hilbert space OH, complex interpolation and tensor norms}, Mem. Am. Math.
Soc. 122, (585) (1996).

\bibitem{PiVar} G. Pisier, \emph{Probabilistic methods in the geometry of Banach spaces}, Probability and Analysis (Varenna. Italy, 1985), Lecture Notes Math. 1206, 167-241 (1986).


\bibitem{Popescu} S. Popescu, \emph{Bell's Inequalities and Density Matrices: Revealing ``Hidde'' Nonlocality}, Phys. Rev. Lett. 74, 2619 (1995).


\bibitem{Rao} A. Rao, \emph{Parallel repetition in projection games and a concentration bound}, STOC (2008).


\bibitem{Raz} R. Raz, \emph{A Parallel Repetition Theorem}, SIAM Journal on Computing 27, 763-803 (1998).

\bibitem{Raz2} R. Raz, \emph{A counterexample to strong parallel repetition}, In 49th Annual IEEE Symposium on Foundations of Computer Science, 369-373, (2008).




\bibitem{SGBMPA} V. Scarani, N. Gisin, N. Brunner, L. Masanes, S. Pino, A. Acin, \emph{Secrecy extraction from no-signalling correlations}, Phys. Rev. A 74, 042339 (2006).


\bibitem{Tomczak} N. Tomczak-Jaegermann, {\it Banach-Mazur Distances and Finite Dimensional Operator Ideals}, Pitman Monographs and
Surveys in Pure and Applied Mathematics 38, Longman Scientific and Technical, (1989).

\bibitem{Tsirelson} B.S. Tsirelson, Hadronic Journal Supplement 8:4, 329-345 (1993).

\bibitem{Vertesi}  T. Vertesi, K.F. Pal, \emph{Bounding the dimension of bipartite quantum
systems}, Phys. Rev. A 79, 042106 (2009).


\bibitem{S.Wehner} S. Wehner, \emph{Tsirelson bounds for generalized Clauser-Horne-Shimony-Holt inequalities}, Phys. Rev. A, 73, 022110 (2006)


\bibitem{Wehner} S. Wehner, M. Christandl, A. C. Doherty, \emph{ A lower bound on the dimension of a quantum system given measured
data}, Phys. Rev. A 78, 062112 (2008).

\bibitem{Werner} R.F. Werner, \emph{Quantum states with Einstein-Podolsky-Rosen correlations admitting a hidden-variable model}, Phys. Rev. A 40, 4277 (1989).


\bibitem{WernerWolf}  R.F. Werner, M.M. Wolf, \emph{Bell inequalities and Entanglement}, Quant. Inf. Comp., \textbf{1} no. 3, 1-25 (2001).



\end{thebibliography}
\end{document}